\documentclass[11pt,a4paper]{article}
\usepackage{amsmath,amssymb,amsthm}
\usepackage{hyperref}
\usepackage{cleveref}
\usepackage{fullpage}
\usepackage{microtype}
\usepackage{natbib}
\usepackage{todonotes}
\usepackage{bbm}
\usepackage{hyperref, color}
\hypersetup{colorlinks=true,citecolor=blue, linkcolor=blue, urlcolor=blue}
\usepackage[linesnumbered,boxed,ruled,vlined]{algorithm2e}
\usepackage{bm}
\usepackage{enumerate} 
\usepackage{enumitem}
\usepackage{footmisc}
\usepackage{footnote}
\usepackage{framed}
\usepackage{xcolor} 
\usepackage{tcolorbox} 

\newcommand{\DTV}[2]{d_{\mathrm{TV}}\left({#1},{#2}\right)}

\renewcommand{\epsilon}{\varepsilon}

\newcommand{\dis}{d_\mathrm{par}}
\newcommand{\Thre}{\theta}

\newcommand{\CC}{C^\text{TV}_{\mathrm{par}}}
\newcommand{\Zux}{\tilde{Z}_{S,\mu}^x (t)}
\newcommand{\Zvx}{\tilde{Z}_{S,\nu}^x (t)}
\newcommand{\ux}{\tilde{\mu}_{S,t}^x}
\newcommand{\vx}{\tilde{\nu}_{S,t}^x}

\newcommand{\TS}{T^{\textsf{sp}}}
\newcommand{\TC}{T^{\textsf{ct}}}
\newcommand{\TW}{T^{\textsf{wt}}}

\newtheorem{theorem}{Theorem}
\newtheorem{observation}[theorem]{Observation}
\newtheorem{claim}[theorem]{Claim}
\newtheorem*{claim*}{Claim}
\newtheorem{condition}[theorem]{Condition}

\newtheorem{lemma}[theorem]{Lemma}
\newtheorem{proposition}[theorem]{Proposition}
\newtheorem{corollary}[theorem]{Corollary}
\theoremstyle{definition}

\newtheorem{definition}[theorem]{Definition}
\newtheorem{remark}[theorem]{Remark}
\newtheorem{problem}[theorem]{Problem}
\newtheorem*{remark*}{Remark}

\def\Pr{\mathop{\mathbf{Pr}}\nolimits}

\renewcommand{\emptyset}{\varnothing}

\newcommand{\abs}[1]{\left\vert#1\right\vert}

 \newcommand{\tuple}[1]{\left(#1\right)} 
 
 \newcommand{\tp}{\tuple}

\newcommand{\defeq}{\triangleq}

\def\*#1{\mathbf{#1}} 
\def\+#1{\mathcal{#1}} 
\def\-#1{\mathrm{#1}} 

\usepackage{xifthen}

\renewcommand{\Pr}[2][]{ \ifthenelse{\isempty{#1}}
  {\mathbf{Pr}\left[#2\right]} {\mathbf{Pr}_{#1}\left[#2\right]} } 
\newcommand{\E}[2][]{ \ifthenelse{\isempty{#1}}
  {\mathbf{\mathbf{E}}\left[#2\right]}
  {\mathbf{\mathbf{E}}_{#1}\left[#2\right]} }
  \newcommand{\Var}[2][]{ \ifthenelse{\isempty{#1}}
  {\mathbf{\mathbf{Var}}\left[#2\right]}
  {\mathbf{\mathbf{Var}}_{#1}\left[#2\right]} }
\newcommand{\Ent}[2][]{ \ifthenelse{\isempty{#1}}
  {\mathbf{\mathbf{Ent}}\left[#2\right]}
  {\mathbf{\mathbf{Ent}}_{#1}\left[#2\right]} }
\title{Approximating the total variation distance between spin systems}
\date{}

\author{Weiming Feng \footnote{School of Computing and Data Science, The University of Hong Kong, emails: \texttt{wfeng@hku.hk} and \texttt{ymjessen02@connect.hku.hk}.} \and Hongyang Liu \footnote{State Key Laboratory for Novel Software Technology,
New Cornerstone Science Laboratory, Nanjing University, email: \texttt{liuhongyang@smail.nju.edu.cn}.} \and Minji Yang \footnotemark[1]}

\newboolean{conf}
\setboolean{conf}{false}

\begin{document}

\maketitle

\begin{abstract}
    Spin systems form an important class of undirected graphical models. For two Gibbs distributions \(\mu\) and \(\nu\) induced by two spin systems on the same graph \(G = (V, E)\), we study the problem of approximating the total variation distance \(\DTV{\mu}{\nu}\) with an \(\epsilon\)-relative error. We propose a new reduction that connects the problem of approximating the TV-distance to sampling and approximate counting. Our applications include the hardcore model and the antiferromagnetic Ising model in the uniqueness regime, the ferromagnetic Ising model, and the general Ising model satisfying the spectral condition. 
    
    Additionally, we explore the computational complexity of approximating the total variation distance \(\DTV{\mu_S}{\nu_S}\) between two marginal distributions on an arbitrary subset \(S \subseteq V\). We prove that this problem remains hard even when both \(\mu\) and \(\nu\) admit polynomial-time sampling and approximate counting algorithms.

  \end{abstract}

\section{Introduction}
The total variation distance (TV-distance) is a widely used metric for quantifying the difference between two distributions. For two discrete distributions \(\mu\) and \(\nu\) defined over the sample space \(\Omega\), the TV-distance is given by  
\[
\DTV{\mu}{\nu} \defeq \frac{1}{2}\sum_{\sigma \in \Omega} \abs{\mu(\sigma) - \nu(\sigma)} = \max_{A \subseteq \Omega} \tp{\mu(A) - \nu(A)}.
\]

Alternatively, the TV-distance can be characterized as the minimum probability of \(X \neq Y\), where \(X \sim \mu\) and \(Y \sim \nu\) form a \emph{coupling} of the two distributions. These different characterizations provide a rich set of tools for analyzing the TV-distance. It is also closely related to other measures of distance between distributions, such as the Wasserstein distance and the KL-divergence~\citep{mitzenmacher2017probability}. Consequently, the TV-distance is a fundamental quantity in many applications, including randomized algorithms, statistical physics, and machine learning.

The problem of \emph{computing} the TV-distance between two distributions has garnered considerable attention in machine learning and theoretical computer science. 
A straightforward approach is to compute the TV-distance directly from its definition. This algorithm runs in \(O(|\Omega|)\) time, assuming access to the probability mass at every point in \(\Omega\). However, the problem becomes particularly interesting when the distributions have succinct representation, as the sample space \(\Omega\) can be exponentially large relative to the size of the input.
In fact, the problem is intractable for many classes of distributions. For instance, ~\cite{sahai2003complete} showed that when two distributions are specified by circuits that sample from them, deciding whether their TV-distance is small or large is complete for SZK (statistical zero-knowledge). More recently, Bhattacharyya, Gayen, Meel, Myrisiotis, Pavan, and Vinodchandran~\citep{0001GMMPV23} proved that even for pairs of product distributions, the exact computation of their TV-distance is \textbf{\#P}-complete.

%
%

On the algorithmic side, the problem of \emph{approximating} the TV-distance between two distributions with \(\epsilon\)-relative error was first studied in~\cite{0001GMMPV23}, where the authors proposed an approximation algorithm for a restricted class of product distributions. Subsequently, Feng, Guo, Jerrum, and Wang~\citep{FGJW23} developed a simple FPRAS that works for general product distributions. Later, a deterministic FPTAS was also introduced for the same task~\citep{FengLL24}.
Beyond product distributions, however, the understanding of approximating the TV-distance remains very limited. The only progress was made in~\cite{0001GMM0V24}, which showed that for two Bayesian networks on DAGs, an FPRAS for the TV-distance exists if \emph{exact} probabilistic inference can be performed in polynomial time. As a consequence, their algorithm can be applied when the underlying DAG has bounded treewidth. 
However, the requirement of exact inference is a strong assumption. For many natural graphical models without the bounded treewidth property, exact inference is itself a \(\textbf{\#P}\)-complete problem.

In this paper, we further investigate the problem of approximating the TV-distance for spin systems, which are a fundamental class of undirected graphical models. Canonical examples include the hardcore model and the Ising model. Our contributions can be summarized as:

\begin{itemize}
    \item We give a new algorithm to reduce the TV-distance approximation to sampling and approximate counting. 
    As a result, our algorithm applies to a broad class of hardcore and Ising models, even if the underlying graph has unbounded treewidth.
    \item We analyze the computational complexity of approximating the TV-distance between two \emph{marginal distributions} of spin systems. We show that this problem is \(\textbf{\#P}\)-hard, even in parameter regimes where both sampling and approximate counting are tractable.
\end{itemize}

\subsection{Approximating the TV-distance between two Gibbs distributions}
Let $G = (V,E)$ be a graph. A spin system $\mathbb{S}$ (a.k.a. Markov random field) defines a distribution over $\{-1,+1\}^V$ (denoted by $\{\pm\}^V$ in short) in the following way. It defines a weight function $w=w_{\mathbb{S}}:= \{\pm\}^V \to \mathbb{R}_{\geq 0}$ that assigns each configuration $\sigma \in \{\pm\}^V$ a weight $w(\sigma)$. The spin system $\mathbb{S}$ induces a Gibbs distribution $\mu = \mu_{\mathbb{S}}$ over $\{\pm\}^V$ such that
\begin{align*}
 \forall \sigma \in \{\pm\}^V,\quad   \mu_{}(\sigma) \defeq \frac{w(\sigma)}{Z}, \quad \text{where } Z = Z_{\mathbb{S}} \defeq \sum_{\tau \in \{\pm\}^V}w_{}(\tau) \text{ is the \emph{partition function}.}
\end{align*}
The weight function $w$ of a spin system is a product of factors associated with each vertex and edge on graph $G$, which can be computed exactly and efficiently.
For a broad class of spin systems, 
the following sampling and approximate counting oracles with $\TS_G(\epsilon),\TC_G(\epsilon) = \mathrm{poly}({n}/{\epsilon})$ exist, where $n$ denotes the number of vertices in graph $G$. 
\begin{definition}[sampling and approximate counting oracles]\label{def:oracle}
Let $\mathbb{S}$ be a spin system on graph $G$ with Gibbs distribution $\mu$ and partition function $Z$.
Let $\TS_G,\TC_G:(0,1) \to \mathbb{N}$ be two non-increasing cost functions.
\begin{itemize}
    \item We say $\mathbb{S}$ admits a sampling oracle with cost function $\TS_G$ if given any $0 <\epsilon<1$, it returns a random sample $X \in \{\pm\}^V$ in time $\TS_G(\epsilon)$ with $\DTV{X}{\mu} \leq \epsilon$.
    \item We say $\mathbb{S}$ admits an approx. counting oracle with cost function $\TC_G$ if given any $0 <\epsilon<1$, it returns a random number $\hat{Z}$ in time $\TC_G(\epsilon)$ with $\Pr{(1-\epsilon)Z\leq\hat{Z}\leq (1+\epsilon)Z}\geq 0.99$.
\end{itemize}
\end{definition}

Given two spin systems $\mathbb{S}^\mu$ and $\mathbb{S}^\nu$ on the same graph $G$ that defines two Gibbs distributions $\mu$ and $\nu$ respectively and an error bound $0 <\epsilon < 1$, we study the following problem.
\begin{center}
\emph{Given the access to sampling and approximate counting oracles for both $\mathbb{S}^\mu$ and $\mathbb{S}^\nu$, can we efficiently approximate the TV-distance $\DTV{\mu}{\nu}$ within relative error }$(1\pm \epsilon)$?
\end{center}

In this paper,
we focus on the following two canonical and extensively studied spin systems. 
\begin{itemize}
    \item \textbf{Ising model}: Let $G = (V,E)$ be a graph. Let $J \in \mathbb{R}^{V \times V}$ be a \emph{symmetric interaction matrix} such that $J_{uv} \neq 0$ only if $\{u,v\} \in E$. Let $h \in (\mathbb{R} \cup \{\pm \infty\})^V$ be the \emph{external field}. Define \emph{Hamiltonian} function
  \begin{align*}
  \forall \sigma \in \{\pm\},\quad   H(\sigma) \defeq \sum_{\{u,v\} \in E} J_{uv}\sigma_u\sigma_v + \sum_{v \in V}\sigma_v h_v = \frac{1}{2}\langle \sigma, J\sigma \rangle + \langle \sigma, h \rangle.
  \end{align*}
The weight of a configuration $\sigma$ in Ising model is defined by $w(\sigma) \defeq \exp(H(\sigma))$.
\item \textbf{Hardcore model}: Let $G = (V,E)$ be a graph. Let $\lambda = (\lambda_v)_{v \in V} \in \mathbb{R}_{\geq 0}^V$ be external fields. A configuration $\sigma \in \{\pm\}^V$ is said to be an independent set if $S_\sigma = \{v \in V \mid \sigma_v = +1\}$ forms an independent set in graph $G$. 
The hardcore model is specified by the pair $(G,\lambda)$, which defines the weight function $w$ such that
\begin{align*}
   \forall \sigma \in \{\pm\}^V, \quad w(\sigma) \defeq \begin{cases}
   \prod_{v \in V:\sigma_v = +1}\lambda_v &\text{if $\sigma$ is an independent set};\\
   0 &\text{otherwise}.
   \end{cases}
\end{align*}
\end{itemize}
The problem of approximating the TV-distance can be formalized as follows.
\begin{problem}[Ising TV-distance approximation]\label{label:prob-Ising} 
The problem is defined as follows.
\begin{itemize}
    \item \emph{Input}: two Ising models $(G,J^\mu,h^\mu)$ and $(G,J^\nu,h^\nu)$ defined on the same graph $G = (V,E)$, which specifies two Gibbs distributions $\mu$ and $\nu$ respectively, and an error bound $0<\epsilon<1$.
    \item \emph{Output}: a number $\hat{d}$ such that $(1-\epsilon) \DTV{\mu}{\nu} \leq \hat{d} \leq (1 + \epsilon)\DTV{\mu}{\nu}$.
\end{itemize}
\end{problem}

\begin{problem}[Hardcore TV-distance approximation]\label{label:prob-hardcore} The problem is defined as follows.
\begin{itemize}
    \item \emph{Input}: two hardcore models $(G,\lambda^\mu)$ and $(G,\lambda^\nu)$ defined on the same graph $G = (V,E)$, which specifies two Gibbs distributions $\mu$ and $\nu$ respectively, and an error bound $0<\epsilon<1$.
    \item \emph{Output}: a number $\hat{d}$ such that $(1-\epsilon) \DTV{\mu}{\nu} \leq \hat{d} \leq (1 + \epsilon)\DTV{\mu}{\nu}$.
\end{itemize}
\end{problem}

We are interested in the FPRAS (fully polynomial randomized approximation scheme), which solves the above problems with probability at least $2/3$\footnote{The success probability can be boosted to $1-\delta$ by independently running the algorithm $O(\log \frac{1}{\delta})$ times and taking the median of the output.} in time $\mathrm{poly}({n}/{\epsilon})$.


We need a marginal lower bound condition for general results. 
For any subset $\Lambda \subseteq V$, any feasible partial configuration $\sigma \in \{\pm\}^\Lambda$, let $\mu^\sigma$ denote the distribution $\mu$ conditional on $\sigma$. 
For any $v \in V$, let $\mu^\sigma_v$ denote the marginal distribution on $v$ projected from  $\mu^\sigma$.

\begin{condition}[marginal lower bound]\label{def:marlow}
Let $0 < b < 1$ be a parameter. We say a distribution $\mu$ over $\{\pm\}^V$ is $b$-marginally bounded if
for any  feasible partial configuration $\sigma \in \{\pm\}^\Lambda$ on a subset $\Lambda \subseteq V$, any $v \in V$, any $c \in \{\pm\}$, it holds that $\mu_v^\sigma(c) \geq b$ if $\mu_v^\sigma (c) > 0$.
\end{condition}
The marginal lower bound condition is a natural condition for spin systems. Hardcore and Ising models with constant marginal lower bound were extensively studied in sampling and approximate counting~\citep{jerrum2003counting,Wei06,Sly10,CLV21}. 

\begin{theorem}[general result]\label{thm:Ising-1}
Let $0 < b < 1$ be a constant.
There exists a randomized algorithm that solves 
\Cref{label:prob-hardcore} and \Cref{label:prob-Ising} with probability at least $\frac{2}{3}$ in time 
\begin{align*}
     O_b\tp{\frac{N^2}{\epsilon^2} \TS_G \tp{\mathrm{poly}(b) \cdot \frac{\epsilon^2}{ N^2}} + \TC_G\tp{\mathrm{poly}(b) \cdot \frac{\epsilon}{ N}}},
\end{align*}
 if two input spin systems are both $b$-marginally bounded and both admit sampling and approximate counting oracles with cost functions $\TS_G(\cdot),\TC_G(\cdot)$, 
 where $O_b(\cdot)$ hides a factor depending only on $b$,  $N = n = |V|$ for hardcore model, and $N = n + m = |V| + |E|$ for Ising model.
\end{theorem}

For the hardcore and Ising models with a constant marginal lower bound, the theorem provides the first polynomial-time reduction from approximating the TV-distance to sampling and approximate counting. Moreover, the above theorem is a simplified version, and our technique yields a stronger result that also applies to some hardcore and Ising models with a smaller marginal lower bound $b = o(1)$. See \Cref{remark:b}\ifthenelse{\boolean{conf}}{ in the Appendix}{} or a more detailed discussion.

For Ising models $\mathbb{S} = (G,J,h)$, the following conditions are well-known.
\begin{condition}\label{cond:Ising}
$\mathbb{S}$ satisfies \emph{one of} the following conditions:
\begin{itemize}
    \item Spectral condition: $\lambda_{\max}(J) - \lambda_{\min}(J) \leq 1 - \eta$ for some constant $\eta > 0$, where  $\lambda_{\max}(J)$ and $\lambda_{\min}(J)$ denote the max and min eigenvalues of $J$ respectively.
    \item Ferromagnetic interaction with consistent field condition: $J_{uv} \geq 0$ for all edge $\{u,v\} \in E$ and $h_v \geq 0$ for all $v \in V$. 
    \item Anti-ferromagnetic interaction at or within the uniqueness threshold: $J_{uv}=\beta\leq 0$ for all $\{u,v\} \in E$ and $\exp(2\beta) \geq \frac{\Delta-2}{\Delta} $, where $\Delta$ is the maximum degree of $G$.
\end{itemize}
\end{condition}

Previous works \citep{AJKPV22, FengGW23,JS93,CCYZ24,JVV86,SVV09} gave $\mathrm{poly}(\frac{n}{\epsilon})$ time sampling and approximate counting oracles for Ising models satisfying \Cref{cond:Ising}.

We have the following corollary for hardcore and Ising models from \Cref{thm:Ising-1}.


\begin{corollary}\label{Cor:Ising}
There exists an FPRAS for \Cref{label:prob-Ising} if two input Ising models both satisfy \Cref{cond:Ising} and the marginal lower bound in \Cref{def:marlow} with $b = \Omega(1)$.
\end{corollary}


Consider a hardcore model  $(G,\lambda)$. Let $\Delta \geq 3$ denote the maximum degree of $G$. The hardcore model is said to satisfy the \emph{uniqueness condition} with a constant gap $0 < \eta < 1$ if
\begin{align}\label{eq:cond-hardcore}
    \forall v \in V, \quad \lambda_v \leq (1-\eta)\lambda_c(\Delta), \quad \text{ where } \lambda_c(\Delta) \defeq \frac{(\Delta-1)^{\Delta - 1}}{(\Delta-2)^\Delta} \approx \frac{e}{\Delta}.
\end{align}
Previous works~\citep{CFYZ22,CE22,SVV09} proved that for hardcore model satisfying the uniqueness condition in~\eqref{eq:cond-hardcore}, it admits $\text{poly}(\frac{n}{\epsilon})$ time sampling and approximate counting oracles. 
To obtain a $\text{poly}(\frac{n}{\epsilon})$ time algorithm for TV-distance estimation, one way is to consider hardcore models in the uniqueness regime with constant marginal lower bound $b =\Omega(1)$, which requires that $\Delta = O(1)$ and for all $v \in V$, either $\Omega(1) \leq \lambda_v \leq (1-\eta)\lambda_c(\Delta)$ or $\lambda_v = 0$\footnote{We allow $\lambda_v = 0$ because \Cref{def:marlow} only considers the lower bound for spins $c \in \{\pm\}$ with $\mu^\sigma_v(c) > 0$.}. If both two hardcore models satisfy this condition, then the algorithm in \Cref{thm:Ising-1} solves \Cref{label:prob-hardcore} in time
\begin{align}\label{eq:time-gen}
O_b\tp{\frac{N^2}{\epsilon^2} \TS_G \tp{\mathrm{poly}(b) \cdot \frac{\epsilon^2}{ N^2}} + \TC_G\tp{\mathrm{poly}(b) \cdot \frac{\epsilon}{ N}}} = \tilde{O} \tp{\frac{n^4}{\epsilon^2}}.  
\end{align}
The equation holds because $\TS_G(\delta) = O(\Delta n\log\frac{n}{\delta})$~\citep{CFYZ22,CE22}, $\TC_G(\delta) = \tilde{O}(\frac{\Delta n^2}{\delta^2})$~\citep{SVV09}, and $\Delta = O(1)$ due to the constant marginal lower bound assumption. For these hardcore models, we can give a faster algorithm than the general results in \Cref{thm:Ising-1}.
Compared to the running time in~\eqref{eq:time-gen}, the following improved algorithm reduces a factor of $n$ in the running time.

\begin{theorem}[faster algorithm further assuming constant marginal lower bound]\label{thm:hardcore-2}
Let $0 < \eta < 1$ and $\Delta \geq 3$ be two constants.
There exists an FPRAS in time $\tilde{O}(\frac{n^3}{\epsilon^2})$ for \Cref{label:prob-hardcore} if two input hardcore models are defined on graph $G$ with maximum degree $\Delta$ and all external fields satisfy that $\Omega(1) \leq \lambda^\pi_v \leq (1 - \eta)\lambda_c(\Delta)$ or $\lambda^\pi_v = 0$ for all $v \in V$ and $\pi \in \{\mu,\nu\}$.
\end{theorem}

\Cref{thm:hardcore-2} works for hardcore model in the uniqueness regime. However, it additionally requires a marginal lower bound $b = \Omega(1)$.
The following improved algorithm removes the marginal lower bound requirement, thus it works for the \emph{whole} uniqueness regime.

\begin{theorem}[improved algorithm for the
whole uniqueness regime]\label{thm:hardcore-1}
Let $0 < \eta < 1$ be a constant.
 There exists an FPRAS in time $\tilde{O}_\eta(\frac{\Delta n^7}{\epsilon^{5/2}})$ for \Cref{label:prob-hardcore}  if two input hardcore models both satisfy the uniqueness condition with a constant gap $\eta > 0$, where $n$ is the number of vertices, $\Delta$ is the maximum degree, and $\tilde{O}_\eta(\cdot)$ hides a constant factor depending on $\eta$ and a $\mathrm{polylog}(\frac{n}{\epsilon})$ factor.
\end{theorem}


The conditions required by \Cref{Cor:Ising}, \Cref{thm:hardcore-2}, and \Cref{thm:hardcore-1} are arguably significant conditions for TV-distance estimation.
Let $\Delta \geq 3$ be a constant.
For the hardcore model beyond the uniqueness regime such that $\lambda_v > \lambda_c(\Delta)$ for all $v \in V$, polynomial-time sampling and approximate counting oracles do not exist unless $\textbf{NP}=\textbf{RP}$~\citep{Sly10}. Moreover, the technique in~\cite{BGMMPV24ICLR} can show that unless $\textbf{NP}=\textbf{RP}$, there is no FPRAS for \Cref{label:prob-hardcore} if input hardcore models are beyond the uniqueness condition. For the completeness, we give a simplified proof in~\Cref{app:proof}. The hardness result holds even for \emph{additive-error} approximation.

To see the significance of condition in  \Cref{Cor:Ising}, consider the following family of Ising models.
Let $\mathbb{S}^\mu = (G,J,h^\mu)$ and $\mathbb{S}^\nu = (G,J,h^\nu)$ be two Ising models defined on the same graph $G$ and have the same interaction matrix $J$.
Assume that $G$ has constant maximum degree $\Delta$; 
for all edges $\{u,v\}$ in $G$,
$J_{uv} = \beta < 0$ is a unified negative constant; 
and for all vertices $v \in V$, $h^\mu_v$ and $h^\nu_v$ can take values from set $\{\pm\infty,0\}$. This family of Ising models has a constant marginal lower bound $b = \Omega(1)$. We have the following two results:
\begin{itemize}
    \item if $\exp(2\beta) \geq \frac{\Delta-2}{\Delta}$, then FPRAS for TV-distance exists by \Cref{Cor:Ising};
    \item if $\exp(2\beta) < \frac{\Delta-2}{\Delta}$, there is no FPRAS for TV-distance unless $\textbf{NP}=\textbf{RP}$, the hardness result holds even for approximating the TV-distance with \emph{additive} error.
\end{itemize}
Again, the hardness result can be proved by the technique in \cite{BGMMPV24ICLR}. We give a simple proof in \Cref{app:proof} for the completeness.

\subsection{Approximating the TV-distance between two marginal distributions}
One natural extension is to approximate the TV-distance between two marginal distributions on a subset of vertices.
We use the hardcore model as an example to state our results on this problem. The same results can be extended to Ising model using a similar reduction.

\begin{problem}[$k$-marginal TV-distance approximation]\label{label:prob-mar} 
Let $k:\mathbb{N} \to \mathbb{N}$ be a function. 
\begin{itemize}
    \item \emph{Input}: two hardcore models $(G,\lambda^\mu)$ and $(G,\lambda^\nu)$ defined on the same graph $G = (V,E)$, which specifies two Gibbs distributions $\mu$ and $\nu$ respectively, a subset $S \subseteq V$ such that $|S| = k(n)$, where $n = |V|$, and an error bound $\epsilon > 0$.
    \item \emph{Output}: a number $\hat{d}$ such that $ \frac{\DTV{\mu_S}{\nu_S}}{1+\epsilon} \leq \hat{d} \leq (1 + \epsilon)\DTV{\mu_S}{\nu_S}$, where $\mu_S$ and $\nu_S$ are marginal distributions on $S$ projected from $\mu$ and $\nu$ respectively.
\end{itemize}
\end{problem}
In particular, if the function $k(n) = n$, then \Cref{label:prob-mar} is the same as \Cref{label:prob-hardcore}.

We show that the problem is hard when $k(n) = 1$. In this case, the problem is to approximate the TV-distance between two marginal distributions at a single vertex. The hardness result holds even if two input hardcore models lie in the uniqueness regime, where both sampling and approximate counting are tractable in polynomial time.
\begin{theorem}\label{thm:one-vertex}
Let $k(n) = 1$ for all $n \in \mathbb{N}$ be a constant function.  The $k$-marginal TV-distance approximation is \textbf{\#P}-Hard when two input hardcore models both satisfy the uniqueness condition, the hardness result holds even if $\epsilon = \mathrm{poly}(n)$, where  $n$ is number of vertices in $G$.
\end{theorem}

The above theorem is for marginal distributions at one vertex. One can simply lift the result to the marginal distributions on a set of vertices.
In particular, we have the following corollary.

\begin{corollary}\label{thm:many-vertex}
Let $0 < \alpha < 1$ be a constant and $k(n) = n - \lceil n^\alpha \rceil$. The $k$-marginal TV-distance approximation is \textbf{\#P}-Hard when two input hardcore models both satisfy the uniqueness condition, the hardness result holds even if $\epsilon = \mathrm{poly}(n)$, where  $n$ is number of vertices in $G$.      
\end{corollary}

The proofs of the hardness results are given in \Cref{sec:hard}\ifthenelse{\boolean{conf}}{ of the Appendix}{}.
The proof constructs a Turing reduction that exactly counts the number $Z$ of independent sets in graphs with a maximum degree of 3, a problem known to be \textbf{\#P}-complete~\citep{DyerG00}. Specifically, we show that if one can efficiently solve the problem stated in \Cref{thm:one-vertex}, then it is possible to efficiently estimate the probability that a vertex $v$ is included in a uniformly random independent set, with an \emph{exponentially} small relative error of $4^{-n}$. This, in turn, solves the \emph{exact} counting problem by using the self-reducibility~\citep{JVV86} property of the hardcore model and the fact that $Z \leq 2^n$.

Now, let us compare our hardness results with that in \cite{BGMMPV24ICLR}. The hardness results in \cite{BGMMPV24ICLR} are for approximating the TV-distance between two \emph{entire Gibbs distributions}. In contrast, our hardness results are for approximating the TV-distance between two \emph{marginal distributions}. The hardness results in \cite{BGMMPV24ICLR} are \textbf{NP}-hard results. It considered the Gibbs distributions in a parameter regime where sampling and approximate counting are intractable. In contrast, our hardness results are \textbf{\#P}-hard results. Our hardness results holds even if sampling and approximate counting can both be solved in polynomial time.

Finally, consider approximating the TV-distance between two \emph{marginal distributions} with \emph{additive error} $\epsilon$. We show that this relaxed problem admits FPRAS if two input hardcore models both satisfy the uniqueness condition. 

\begin{theorem}\label{thm:many-vertex-alg}
There exists a randomized algorithm such that given two hardcore distributions $\mu$ and $\nu$ on the same graph $G=(V,E)$ with $n = |V|$, any subset $S \subseteq V$, and any $0 < \epsilon < 1$, if $\mu$ and $\nu$ both satisfy~\eqref{eq:cond-hardcore}, it returns a random number $\hat{d}$ in time $\frac{\Delta n^2}{\epsilon^4} \cdot \mathrm{polylog}(\frac{n}{\epsilon})$ such that \[\Pr{\big|\hat{d} - \DTV{\mu_S}{\nu_S}\big|\leq \epsilon} \geq \frac{2}{3}.\]
\end{theorem}

\Cref{thm:many-vertex-alg} together with two hardness results give a clear separation between the computational complexity for approximating the TV-distance with relative and additive error.

\subsection{Related work and open problems} 

There is a series of works on checking whether two given distributions are identical (e.g. \cite{CortesMR07,DoyenHR08,KieferMOWW11,BGMMPV24ICLR}), which can be viewed as the decision version of computing the TV-distance, i.e., checking whether it is zero.

There are also a long line of works (e.g., see a survey in \cite{Canonne15}) studying the identical testing problem in access model, where the algorithm can only access a set of random samples from the distributions.
This setting is different from our setting, where we assume that all the parameters of the spin systems are given as the input to the algorithm.

A series of works~\citep{0001GMV20,ChenK14,Kiefer18,CanonneR14} studied the algorithm and the hardness of approximating the TV-distance with additive error, which is an easier problem than the relative-error approximation.

\cite{devroye2018total,ArbasAL23,kontorovich2024tensorization,kontorovich2024sharp} found \emph{closed-form formulas}, which approximate the TV-distance between two high-dimensional distributions with a \emph{fixed} relative error.
We study \emph{algorithms} achieving an \emph{arbitrary} $\epsilon$-relative error approximation for spin systems.

This work raises several open problems: for the Ising model, it remains open whether the marginal lower bound $b = \Omega(1)$ in \Cref{Cor:Ising} can be removed, as our current techniques rely on this condition; for the hardcore model, the $\tilde{O}(\Delta n^7)$ runtime in \Cref{thm:hardcore-1} may admit further improvements; finally, while this paper focuses on 2-spin systems with pairwise interactions, extending these results to more general models like the Potts model or higher-order Markov random fields presents an interesting direction for future work.


\section{Algorithm overview}
In this section, we give an overview of our algorithm.
Let $G=(V,E)$ be a graph. 
Let $\mathbb{S}^\mu$ and $\mathbb{S}^\nu$ be two spin systems (either two hardcore models or two Ising models) defined on the same graph $G$.
Let $\mu$ and $\nu$ denote Gibbs distributions of $\mathbb{S}^\mu$ and $\mathbb{S}^\nu$ respectively.

We first introduce the distance between parameters of two spin systems. For a vector $a \in \mathbb{R}^V$, denote $\Vert a \Vert_\infty = \max_{v \in V}|a_v|$. For a matrix $A \in  \mathbb{R}^{V \times V}$, denote $\Vert A \Vert_{\max} = \max_{u,v \in V}|A_{uv}|$.
\begin{definition}[parameter distance]\label{def:dis}
The parameter distance $\dis(\mathbb{S}^\mu,\mathbb{S}^\nu)$ between $\mathbb{S}^\mu$ and $\mathbb{S}^\mu$, which is denoted by  $\dis(\mu,\nu)$ for simplicity, is defined by  
\begin{itemize}
    \item Hardcore model: for two hardcore models $\mathbb{S}^\mu = (G,\lambda^\mu)$ and  $\mathbb{S}^\nu = (G,\lambda^\nu)$, \[\dis(\mu,\nu) \defeq \Vert \lambda^\mu - \lambda^\nu \Vert_{\infty}.\]
    \item Soft-Ising model:  for two soft-Ising models $\mathbb{S}^\mu = (G,J^\mu,h^\mu)$ and $\mathbb{S}^\nu = (G,J^\nu,h^\nu)$, \[\dis(\mu,\nu) \defeq \max\left\{ \Vert J^\mu - J^\nu \Vert_{\max}, \max_v \frac{|h^\mu(v)-h^\nu(v)|}{\deg_v+1} \right\},\]
    where $\deg_v$ is the degree of $v$ in graph $G$.
\end{itemize}
\end{definition}

The above parameter distance can be computed easily.
The following lemma gives the relation between parameter distance $\dis(\mu,\nu)$ and the total variation distance $\DTV{\mu}{\nu}$. 
An Ising model $(G=(V,E),J,h)$ is said to be \emph{soft} if $h \in \mathbb{R}^V$ (instead of $h \in (\mathbb{R} \cup \{\pm \infty\})^V$).
We now focus on soft-Ising model in this overview. 
We will show how to reduce general Ising models to soft models in \Cref{sec:proof-main}\ifthenelse{\boolean{conf}}{ of the Appendix.}{.}
\begin{lemma}[TV-distance lower bound]\label{lem:TV-lower}
It holds that $\DTV{\mu}{\nu} \geq \CC\cdot \dis(\mu,\nu)$ such that
\begin{itemize}
    \item Hardcore model: if both $\mu$ and $\nu$ satisfy the uniqueness condition in~\eqref{eq:cond-hardcore}, $\CC = \frac{1}{5000}$.
    \item Hardcore model: if both $\mu$ and $\nu$ are $b$-marginally bounded, $\CC = b^3$.
    \item Soft-Ising model: if both $\mu$ and $\nu$ are $b$-marginally bounded, $\CC = \frac{b^2}{2}$.
\end{itemize}
\end{lemma}

By \Cref{lem:TV-lower}, if $\dis(\mu,\nu)$ is large, the TV-distance $\DTV{\mu}{\nu}$ is also large. Let $n = |V|$ and $m = |E|$ denote the numbers of vertices and edges in $G$. Define the a threshold $\theta = \frac{c_b}{\text{poly}(n)}$ for the parameter distance for soft-Ising and hardcore models, where $c_b$ is some parameter depending only on the marginal lower bound $b$\footnote{The value of $b$ is not given in the input, but we can compute $b$ efficiently (see \Cref{lem:alg-b}).}. The specific value of $\theta$ can be found in \Cref{sec:proof-main}\ifthenelse{\boolean{conf}}{ of the Appendix}{}.
Our algorithm first compute $\dis(\mu,\nu)$ in time $O(m+n)$, and then compares $\dis(\mu,\nu)$ to the threshold $\Thre$. The algorithm considers the following two cases.
\begin{itemize}
    \item Case $\dis(\mu,\nu) \geq \Thre$: By \Cref{lem:TV-lower}, the TV-distance $\DTV{\mu}{\nu} = \CC \dis(\mu,\nu) \geq \Omega_b(\frac{1}{\mathrm{poly}(n)})$ is large. In this case, the task of approximating $\DTV{\mu}{\nu}$ with relative error is the same (up to $O_b(\mathrm{poly}(n))$ running time) as the task of approximating $\DTV{\mu}{\nu}$ with \emph{additive} error. We give a general FPRAS that achieves the additive-error approximation assuming polynomial-time sampling and approximate counting oracles for both $\mu$ and $\nu$.
    \item Case $\dis(\mu,\nu) < \Thre$: The algorithm for this case is our main technical contribution.  Let $w_{\mu}(\cdot)$ and $w_\nu(\cdot)$ denote the weight functions for spin systems $\mathbb{S}^\mu$ and $\mathbb{S}^\nu$ respectively.
    By the definition of parameter distance, we know that for any $\sigma \in \{\pm\}^V$, $w_\mu(\sigma) \approx w_\nu(\sigma)$. 
    We will utilize this property to design an efficient approximation algorithm for the TV-distance. 
\end{itemize}
In \Cref{sec:add-alg} and \Cref{sec:mul-alg}, we will explain the main ideas of our algorithms for the above two cases.
The formal description of the algorithms are given in Sections \ref{sec:add} and \ref{sec:alg-main}\ifthenelse{\boolean{conf}}{ of the Appendix}{}.

\subsection{Warm-up: additive-error approximation algorithm}\label{sec:add-alg}
Let us first consider the easy case.
This case is solved by generalizing the algorithm in~\cite{0001GMV20}.
The paper~\citep{0001GMV20} focused on algorithms that have sample access to the distributions. However, their technique can be extended to our setting.
By the definition of total variation distance, we can write
\begin{align}\label{eq:TV-1}
  \DTV{\mu}{\nu} &= \sum_{\sigma \in \{\pm\}^V:\mu(\sigma)>\nu(\sigma)}|\mu(\sigma)-\nu(\sigma)| = \sum_{\sigma \in \{\pm\}^V} \mu(\sigma) \max \left( 0, 1 - \frac{\nu(\sigma)}{\mu(\sigma)} \right).
\end{align}
Define the random variable $X \defeq \max \left( 0, 1 - \frac{\nu(\sigma)}{\mu(\sigma)} \right)$, where $\sigma \sim \mu$. Note that $0 \leq X \leq 1$. We have $\E[]{X} = \DTV{\mu}{\nu}$ and $\Var[]{X} \leq 1$. Additive-error approximation can be achieved if random samples of the variable $X$ can be efficiently generated.

However, since we can only access sampling and approximate counting oracles, we cannot compute $\mu(\sigma)$ and $\nu(\sigma)$ exactly. Instead, our algorithm uses an alternative estimator $\hat{X}$ to approximate the random variable $X$.
Let $Z_\mu, w_{\mu}(\cdot)$ and $Z_\nu, w_\nu(\cdot)$ denote the partition functions and weight functions of $\mathbb{S}^\mu$ and $\mathbb{S}^\nu$, respectively. We first call the approximate counting oracle to obtain $\hat{Z}_\mu$ and $\hat{Z}_\nu$, which approximate $Z_\mu$ and $Z_\nu$ with a relative error of $O(\epsilon)$. The estimator $\hat{X}$ is then defined by the following process:

\begin{itemize}
    \item Call sampling oracle to generate approximate sample $\sigma$ from $\mu$;
    \item $\hat{X} = \max\left(0, 1 - \frac{\hat{\nu}(\sigma)}{\hat{\mu}(\sigma)}\right)$, where $\hat{\nu}(\sigma) = w_\nu(\sigma)/\hat{Z}_\nu$ and $\hat{\mu}(\sigma) = w_\mu(\sigma)/\hat{Z}_\mu$.
\end{itemize}
The error of $\hat{X}$ arises from the errors in the approximate sample $\sigma$ and the probabilities $\hat{\nu}(\cdot)$ and $\hat{\mu}(\cdot)$. To illustrate the main idea of the algorithm, let us ignore the sampling error and assume $\sigma \sim \mu$. 
Note that the true value $\frac{\nu(\sigma)}{\mu(\sigma)} = \frac{w_\nu(\sigma)}{w_{\mu}(\sigma)} \cdot \frac{Z_{\mu}}{Z_\nu}$. 
By the assumption of approximate counting oracle, $\frac{\hat{\nu}(\sigma)}{\hat{\mu}(\sigma)} \geq (1-O(\epsilon))\frac{\nu(\sigma)}{\mu(\sigma)}$, which implies $\mathbf{E}[\hat{X}] \leq \mathbf{E}_{\sigma \sim \mu}[\max(0,1-(1-O(\epsilon))\frac{\nu(\sigma)}{\mu(\sigma)})]$, which is at most $\mathbf{E}_{\sigma \sim \mu}[\max(0,1-\frac{\nu(\sigma)}{\mu(\sigma)})] +O(\epsilon)\mathbf{E}_{\sigma \sim \mu} \frac{\nu(\sigma)}{\mu(\sigma)} = \E[]{X} +O(\epsilon)$.  A similar analysis gives the lower bound $\mathbf{E}[\hat X] \geq \E[]{X} - O(\epsilon)$. Since  $0 \leq \hat{X} \leq 1$ so the variance is at most 1, this achieves the additive error approximation of the TV-distance. 

The algorithm in \Cref{thm:many-vertex-alg} estimates the TV-distance between two marginal distributions with additive error, which follows a similar approach. The key difference is that it requires approximating the marginal probabilities $\nu_S(\sigma)$ and $\mu_S(\sigma)$ for $\sigma \in \{\pm\}^S$, where $S \subseteq V$. For the hardcore model within the uniqueness regime, we have not only an efficient approximate counting oracle but also efficient oracles for approximating the conditional partition function $Z^\sigma = \sum_{\tau \in \{\pm\}^V: \tau_S = \sigma} w(\tau)$, which give the efficient approximation of $\mu_S(\sigma)$ and $\nu_S(\sigma)$.


\subsection{Approximation algorithm for instances with small parameter distance}\label{sec:mul-alg}
Consider the case where $\dis(\mu, \nu) < \theta$. The total variation distance $\DTV{\mu}{\nu}$ can be very small (e.g., $\exp(-\Omega(n))$). We cannot use the additive-error approximation algorithm to efficiently achieve a relative-error approximation. The main challenge lies in the term $1 - {\nu(\sigma)}/{\mu(\sigma)}$ in~\eqref{eq:TV-1}. With approximate counting oracles, we can approximate ${\nu(\sigma)}/{\mu(\sigma)}$ with relative error, but there is no guarantee of relative accuracy for $1 - {\nu(\sigma)}/{\mu(\sigma)}$. Due to this difficulty, previous works mainly studied graphical models on bounded treewidth graphs~\citep{0001GMM0V24} and product distributions~\citep{FGJW23}, where ${\nu(\sigma)}$ and ${\mu(\sigma)}$ can be computed \emph{exactly}.

We overcome this challenge by designing an alternative estimator with good concentration property. 
We use the hardcore model as an example to illustrate the main idea. 
Assume both $\mu$ and $\nu$ are hardcore models $(G,\lambda^\mu)$ and $(G,\lambda^\nu)$ satisfying uniqueness condition in~\eqref{eq:cond-hardcore}.
For simplicity of the overview, let us further assume that for any for $v \in V$, $0 < \lambda_v^\mu \leq \lambda_v^\nu$. In words, $\nu$ is obtained by increasing some external fields of $\mu$ by a small amount. We will deal with the general case in later technical sections.

\paragraph{Basic case without small external fields} 
Let us start with a simple case where $\lambda^\mu_v = \Theta(\frac{1}{\Delta})$ for all $v \in V$. The uniqueness condition guarantees that $\lambda^\mu_v = O(\frac{1}{\Delta})$. The assumption additionally requires that every $\lambda^\mu_v$ cannot be too small. By the definitions of total variation distance and Gibbs distributions, we can compute
\begin{align}\label{eq:TV-2}
     \DTV{\mu}{\nu} &= \frac{1}{2}\sum_{\sigma \in \{\pm\}^V}\mu(\sigma)\left|1-\frac{\nu(\sigma)}{\mu(\sigma)}\right| 
     = \frac{Z_\mu}{2Z_\nu}\cdot \sum_{\sigma \in \{\pm\}^V}\mu(\sigma)\left| \frac{Z_{\nu}}{Z_{\mu}} -\frac{w_\nu(\sigma)}{w_\mu(\sigma)}\right|.
\end{align}
The first term $\frac{Z_\mu}{2Z_\nu}$ can be approximated with relative error by approximate counting oracle. For the second term, we consider the following random variable, which also appears in the previous work of approximate counting~\citep{SVV09},
\begin{align*}
 W \defeq \frac{w_{\nu}(\sigma)}{w_{\mu}(\sigma)}, \text{ where } \sigma \sim \mu.
\end{align*}
Note that $\E[]{W} = \frac{Z_\nu}{Z_\mu}$. By~\eqref{eq:TV-2}, our task is reduced to estimate the value of $\E[]{\vert \E[]{W} - W \vert}$. We are in the case that $ 0\leq \lambda^\nu_v-\lambda^\mu_v \leq \dis(\mu,\nu) < \theta = \frac{1}{\mathrm{poly}(n)}$, so that for any $\sigma \in \{\pm\}^V$, $\frac{w_\mu(\sigma)}{w_\nu(\sigma)} \approx 1$.
We will utilize this property to show that $W$ has a good concentration property. Formally, 
\begin{align}\label{eq:w-bound}
 1\leq \frac{w_\nu(\sigma)}{w_{\mu}(\sigma)} &\leq  \prod_{v \in V: \sigma_v = +1} \frac{\lambda_v^\nu}{\lambda^\mu_v} \leq \prod_{v \in V: \sigma_v = +1}\left(1 + \frac{\dis(\mu,\nu)}{\lambda^\mu_v}\right) \notag\\
 (\star)\quad&\leq 1 + O \left( \sum_{v \in V: \sigma_v = +1} \frac{\dis(\mu,\nu)}{\lambda^\mu_v} \right) \leq 1 + O(n\Delta) \cdot \dis(\mu,\nu).
\end{align}
In inequality $(\star)$, we use the fact that $\lambda^\mu_v = \Omega(\frac{1}{\Delta})$ so that $\frac{\dis(\mu,\nu)}{\lambda^\mu_v} = O(\Delta \theta)$, 
since we choose $\theta = \frac{1}{\text{poly}(n)}$ small enough, then $ \frac{\dis(\mu,\nu)}{\lambda^\mu_v} \ll \frac{1}{n}$ and the inequality can be verified as follows
\begin{align*}
\prod_{v \in V: \sigma_v = +1}\left(1 + \frac{\dis(\mu,\nu)}{\lambda^\mu_v}\right) \leq \exp \underbrace{ \tp{\sum_{v \in V: \sigma_v = +1}\frac{\dis(\mu,\nu)}{\lambda^\mu_v}}}_{\ll n \cdot (1/n) = o(1)}   = 1 + O \left( \sum_{v \in V: \sigma_v = +1} \frac{\dis(\mu,\nu)}{\lambda^\mu_v} \right).
\end{align*}
By~\Cref{lem:TV-lower}, $ \DTV{\mu}{\nu} = \Omega(\dis(\mu,\nu))$. This implies $W$ enjoys a very good concentration property such that
\begin{align}\label{eq:w-var}
    1 \leq W \leq 1 + O(n\Delta) \cdot \DTV{\mu}{\nu} \quad \implies \quad  \Var[]{W} \leq O(n^2\Delta^2) \cdot (\DTV{\mu}{\nu})^2.
\end{align}
Now, our algorithm can be outline as follows.
\begin{tcolorbox}[colback=lightgray!20, colframe=lightgray!18, coltitle=black, title={}]
    \begin{itemize}
        \item Call approximate counting oracles to obtain $\hat{Z}_\mu$ and $\hat{Z}_\nu$ with relative error $O({\epsilon})$.
        \item Draw $T=\mathrm{poly}(n/\epsilon)$ samples $W_1,\dots, W_T$ from $W$ independently.
        \item Compute $\bar{W}=\frac{1}{T}\sum_{i=1}^T W_i$. \hfill \texttt{(approximate  $\E[]{W}$)}
        \item Compute $\bar{E}=\frac{1}{T}\sum_{i=1}^{T} |W_i-\bar{W}|$. \hfill \texttt{(approximate $\E[]{|W-\E[]{W}|}$)}
        \item Return $\hat{d}=\frac{\hat{Z}_\mu}{2\hat{Z}_\nu}\bar{E}$.
        \end{itemize}
    \end{tcolorbox}
Using the variance bound in~\eqref{eq:w-var}, we can prove that with high probability, $\bar{E}$ approximates $\E[]{|W-\E[]{W}|}$ with an additive error of $O(\epsilon) \cdot \DTV{\mu}{\nu}$.
By~\eqref{eq:w-bound}, we can also verify that $\frac{Z_\mu}{Z_\nu} = O(1)$. Hence, we can bound the error as follows
\begin{align*}
   \hat{d} &\leq (1+O(\epsilon))\frac{Z_\nu}{2Z_\mu}\cdot\left( \E[]{|W-\E[]{W}|} + O(\epsilon) \cdot \DTV{\mu}{\nu} \right)
   \leq (1+O(\epsilon))\DTV{\mu}{\nu}.
\end{align*}
A similar analysis gives the lower bound $\hat{d} \geq (1-O(\epsilon))\DTV{\mu}{\nu}$. This achieves the relative-error approximation of the TV-distance.
The above technique can be generalized to Ising models and hardcore models with a marginal lower bound.


\paragraph{General case containing small external fields} 
For the general case, there may exist vertex $v \in V$ such that the external field $\lambda^\mu_v \ll \dis(\mu,\nu)$. In this case, the inequality $(\star)$ in~\eqref{eq:w-bound} may not hold. In fact, it will cause a fundamental problem to the above algorithm. Consider the case that $\lambda^\mu_v = \exp(-n)$ and $\lambda^\nu_v = \lambda^\mu_v + D$ for all $v \in V$, where $D > 0$. If we draw polynomial number of samples $\sigma \sim \mu$, typically, every sample $\sigma$ corresponds to the empty set, i.e., $\sigma_u = -1$ for all $ u \in V$. Hence, with high probability, all $W_i$ in the above algorithm are $\frac{w_\nu(\emptyset)}{w_\mu(\emptyset)} = 1$ and the algorithm will return $\hat{d} = 0$. However, the true TV-distance $\DTV{\mu}{\nu} > 0$ is positive. 

Let us first consider a special case such that for all $v \in V$, $\lambda^\mu_v < \kappa$ and $\lambda^\nu_v < \kappa$, where $\kappa =  \mathrm{poly}(\frac{\epsilon}{n}) \ll \frac{1}{n}$ is very small. 
The total variation distance is the sum of $\frac{1}{2}|\mu(\sigma)-\nu(\sigma)|$ for all independent sets $\sigma \in \{\pm\}^V$.  
In this special case, all external fields are tiny, so that large independent sets appear with very low probability.
Let $\Vert \sigma\Vert_+$ be the number of $+1$ in $\sigma$.
We can show that there is a constant $t = O(1)$ depending on $\kappa$ such that
\begin{align*}
    \frac{1}{2}\sum_{\sigma \in \{\pm\}^V: \Vert \sigma\Vert_+ \geq t+1} |\mu(\sigma)-\nu(\sigma)| \leq O(\epsilon) \cdot \DTV{\mu}{\nu}.
\end{align*} 
In words, to approximate the total variation distance, we only need to consider the independent sets with size at most $t$.
Now, we define a distribution $\mu'$ as the distribution $\mu$ restricted on the independent sets with size at most $t$.
Similarly, we can define $\nu'$ from $\nu$. 
Since $t$ is a constant, we can enumerate all independent sets with size at most $t$ and compute the total variation distance between $\mu'$ and $\nu'$.
We can show that $\DTV{\mu'}{\nu'}$ approximate $\DTV{\mu}{\nu}$ with $\epsilon$ relative error.

For the most general case, we divide the vertices into two groups. The big group $B$ contains all vertices $v \in V$ such that $\min\{\lambda^\mu_v,\lambda^\nu_v\} > \kappa$. The small group $S$ contains all vertices $v \in V$ such that $\min\{\lambda^\mu_v,\lambda^\nu_v\} \leq \kappa$ for a small $\kappa = \mathrm{poly}(\frac{\epsilon}{n})$. The total variation distance is
\begin{align*}
    \DTV{\mu}{\nu} &= \frac{1}{2}\sum_{\sigma \in \{\pm\}^V} \left|\mu(\sigma)-{\nu(\sigma)}\right| = \frac{1}{2}\sum_{x \in \{\pm\}^B}\sum_{y \in \{\pm\}^S} \left|\mu_B(x)\mu_S^x(y)-\nu_B(x)\nu_S^x(y)\right|\\
    &= \sum_{x \in \{\pm\}^B}\mu_B(x) \cdot  \underbrace{\frac{1}{2}\sum_{y \in \{\pm\}^S} \left|\frac{\nu_B(x)}{\mu_B(x)}\nu_S^x(y)-\mu_S^x(y)\right|}_{f(x)}.
\end{align*}
In a high level, our algorithm will draw independent samples $x \sim \mu_B$ from the marginal distribution and approximately compute the value of $f(x)$. The algorithm finally outputs the average value of $f(x)$ over all samples. 
\ifthenelse{\boolean{conf}}{We need to deal with the following two main technical challenging.}{To make the idea work, we need to deal with the following two main technical challenging.}
\begin{itemize}
    \item We need to bound the variance of $f(x)$ where $x \sim \mu_B$ to show that polynomial number of samples are sufficient for approximation. To achieve a good bound, we use the \emph{Poincar\'e inequality} for hardcore model in uniqueness regime~\citep{ChenFYZ21} to control the variance.
    \item Given a sample $x \sim \mu_B$, we also need to approximately compute the value of $f(x)$. We need to (1) approximate two distributions $\mu_S^x$ and $\nu_S^x$ over $\{\pm\}^S$; and (2) approximate the ratio $\frac{\nu_B(x)}{\mu_B(x)}$. For the first task, note that both $\mu_S^x$ and $\nu_S^x$ are hardcore distributions on the induced subgraph $G[S]$ such that for all $v \in S$, the external fields are small. Hence, we can approximate them using distributions over independent sets of size at most $t = O(1)$. A similar idea will also be use for the second task: to estimate marginal probabilities $\mu_B(x)$ and $\nu_B(x)$, we also need to consider the total weight of all independent sets $I$ in $G[S]$ such that $I \cup \{v \in B \mid x_v = +1\}$ forms a independent set in $G$. Again, we show that the approximation algorithm only needs to consider all such $I$ with $|I| \leq t = O(1)$.
\end{itemize}
All the technical details for general case are given in \Cref{sec:var-main}\ifthenelse{\boolean{conf}}{ of the Appendix}{}.

\ifthenelse{\boolean{conf}}{\section{Proofs of \#P-hardness results}
In this section, we prove \Cref{thm:one-vertex}, and the proof of \Cref{thm:many-vertex} is deferred to \Cref{sec:hard}. Our starting point is the \#P-hardness for exactly counting the number of independent sets in a graph.

\ifthenelse{\boolean{conf}}{\begin{proposition}[\text{\citet[Theorem 4.2]{DyerG00}}]}{\begin{proposition}[\text{\cite[Theorem 4.2]{DyerG00}}]}\label{prop:hard}
The following problem \#\textsf{Ind}(3) is \#P-complete. Input: a graph $G=(V,E)$ with maximum degree $\Delta = 3$; Output: the exact number of independent sets in $G$.
\end{proposition}

The above problem is exactly computing the partition function of $(G,\lambda)$ with $\lambda_v = 1$ for all $v \in V$.
Let $n = |V|$ and $V = \{1,2,\ldots,n\}$. Let $\mu_{G,\bm{1}}$ denote the uniform distribution over all independent sets in $G$, which is the hardcore distribution in $G$ when $\lambda_v = 1$ for all $v \in V$. Define
\begin{align}\label{eq:def_pi}
    p_i = \Pr[X \sim \mu_{G,\bm{1}}]{X_i = 0 \mid \forall 1\leq j \leq i - 1, X_j = 0},
\end{align}
which is the probability that the vertex $i$ is not in a random independent set $X$ conditional on all $j < i$ not being in $X$. By definition, the total number of independent set is $Z = \frac{1}{\mu_{G,\boldsymbol{1}}(\boldsymbol{0})} = \prod_{i=1}^n \frac{1}{p_i}$.
Suppose for any $i \in [n]$, we can compute $\hat{p}_i$ such that 
\begin{align}\label{eq:tar}
   (1 - 4^{-n})p_i \leq \hat{p}_i \leq (1+4^{-n})p_i.
\end{align}
Let $\hat{Z} = \prod_{i=1}^n \frac{1}{\hat{p}_i}$ and it holds that $(1-3^{-n})Z\leq \hat{Z} \leq (1+3^{-n})Z$. Note that $Z \leq 2^n$. We have $|\hat{Z} - Z| \leq 3^{-n}Z \leq 1.5^{-n} < 0.01$. We can round $\hat{Z}$ to the nearest integer to recover $Z$. Hence, $\#\text{Ind}(3)$ can be reduced to the following high-accuracy marginal estimation problem.
\begin{problem}\label{problem:mar} The high-accuracy marginal estimation problem is defined by
\begin{itemize}
    \item Input: a graph $G=(V,E)$ with $n$ vertices and maximum degree $\Delta = 3$;
    \item Output: $n$ numbers $(\hat{p}_i)_{i \in [n]}$ such that for all $ i \in [n]$, $(1 - 4^{-n})p_i \leq \hat{p}_i \leq (1+4^{-n})p_i$.
\end{itemize}
\end{problem}

Let $k(n) = 1$ for all $n \in \mathbb{N}$ be a constant function. We show that if there is a $\mathrm{poly}(n)$ time algorithm for \Cref{label:prob-mar} if the input error bound $\epsilon = \mathrm{poly}(n)$ and both two input hardcore models satisfy the uniqueness condition, then \Cref{problem:mar} can also be solved in $\mathrm{poly}(n)$ time. \Cref{thm:one-vertex} follows from \Cref{prop:hard}.

Fix an integer $i \in [n]$. Let $G_i$ denote the induced graph $G[S_i]$, where $S_i =\{j \in [n] \mid j \geq i\}$ is the set of vertices with label at least $i$. Let $\mu^{(i)}$ denote the uniform distribution over all independent set in graph $G_i$. In other words, $\mu^{(i)}$ is the Gibbs distribution of hardcore model $(G_i,\boldsymbol{1})$. Then $p_i$ in~\eqref{eq:def_pi} is the marginal distribution on vertex $i$ projected from $\mu^{(i)}$. If the maximum degree of $G_i$ is at most $2$, then $G_i$ is a set of disconnected lines or circles and $p_i$ can be computed exactly in polynomial time. We can assume the maximum degree of $G_i$ is $3$.  

Let $\alpha \geq 0$. Define vector $\lambda^\alpha \in \mathbb{R}^{S_i}$ by
\begin{align}\label{eq:def-lambda-alpha}
    \lambda_j^\alpha = \begin{cases}
        \frac{\alpha}{1 - \alpha} &\text{if } j = i; \\
        0 &\text{if } j \neq i.
    \end{cases}
\end{align}
Let $\nu^\alpha$ denote the Gibbs distribution of $(G_i,\lambda^\alpha)$. Note that $\lambda_c(3) = 4 > 1$. The following observation is easy to verify.
\begin{observation}\label{ob:uniq}
Both $\mu^{(i)}$ and $\nu^\alpha$ satisfies the uniqueness condition in~\eqref{eq:cond-hardcore} if $\alpha \leq \frac{1}{2}$.
\end{observation}

By the definition of $\nu^\alpha$, it is easy to see $\nu^\alpha_i(+1) = \alpha$  and
\begin{align*}
    \DTV{\mu^{(i)}_i}{\nu^\alpha_i} = \vert \mu^{(i)}_i(+1) - \alpha \vert = \vert p_i - \alpha \vert.
\end{align*}
Let $\+A(\alpha)$ be the algorithm such that given $\alpha \in [0,\frac{1}{2}]$, it returns a number $\hat{d}$ such that $ \frac{d_{\text{TV}}(\mu^{(i)}_i,\nu^\alpha_i)}{1+\epsilon}\leq \hat{d} \leq (1+\epsilon)d_{\text{TV}}(\mu^{(i)}_i,\nu^\alpha_i)$, where $\epsilon = \mathrm{poly}(n)$. By \Cref{ob:uniq}, if the polynomial-time algorithm for the problem in \Cref{thm:one-vertex} exists, then $\+A(\alpha)$ runs in $\mathrm{poly}(n)$ time.
We then can use the following algorithm to solve \Cref{label:prob-mar} for $p_i$ in $\mathrm{poly}(n)$ time. Thus, the hardness result in \Cref{thm:one-vertex} follows from \Cref{prop:hard}.

\ifthenelse{\boolean{conf}}{\begin{algorithm2e}[ht]}{
\begin{algorithm}[ht]
}\label{alg:high-TV}
    \caption{Algorithm for high-accuracy marginal estimation}
    Let $\alpha \gets \frac{1}{2}$ and $\epsilon = \mathrm{poly}(n)$ be the parameter assumed by algorithm $\+A$\;
    \For{$t$ from 1 to $50n (1+\epsilon)^2$}{
        $\hat{d} \gets \+A(\alpha)$\;
        $\alpha \gets \alpha - \hat{d}/(1+\epsilon)$\;
        if the bit length of $\alpha$ is more than $100n$, then round $\alpha$ up to the nearest number that has bit length at most $100n$\;
    }
    \Return $\hat{p}_i = \alpha$.
\ifthenelse{\boolean{conf}}{\end{algorithm2e}}{
\end{algorithm}
}
The above algorithm runs in $\mathrm{poly}(n)$ time. We show that the output $\hat{p}$ satisfies \eqref{eq:tar}.

Let $\alpha_t$ be the value of $\alpha$ after the $t$-th iteration. We first show that $p_i \leq \alpha_t$ for all $t$. 
 At the beginning, $\alpha_0 = 1/2$. Since $\mu^{(i)}$ is a uniform distribution over all independent sets, we have $\mu^{(i)}_i(+1) \leq \frac{1}{2}$. By the assumption of algorithm $\+A$, $\+A(\alpha_{t})/(1+\epsilon) \leq  d_{\text{TV}}(\mu^{(i)}_i,\nu^{\alpha_t}_i) = \alpha_t - p_i$. Hence, $\alpha_{t+1} \geq \alpha_t -  \+A(\alpha_{t})/(1+\epsilon) \geq p_i $. 
 
We next bound the value of $\alpha_t - p_i$. At the beginning, $\alpha_0 = \frac{1}{2}$ so that $\alpha_0 - p_i \leq \frac{1}{2}$. Note that $\alpha_{t+1} < \alpha_t -  \frac{\+A(\alpha_{t})}{1+\epsilon}+2^{-90n} \leq \alpha_t -  \frac{\alpha_t-p_i}{(1+\epsilon)^2}+2^{-90n} $, where $2^{-90n}$ is an upper bound of rounding error. The inequality implies that 
\begin{align*}
    \alpha_{t+1} - p_i \leq \left( 1 - \frac{1}{(1+\epsilon)^2} \right)(\alpha_t - p_i) + 2^{-90n}.
\end{align*}
Note that $\hat{p}_i = \alpha_{50n(1+\epsilon)^2}$.
Solving the recurrence implies that
\begin{align*}
    0 \leq \hat{p}_i - p_i \leq \exp\left( -\frac{50n(1+\epsilon)^2}{(1+\epsilon)^2} \right)\cdot \frac{1}{2} + (1+\epsilon)^2 \cdot 2^{-90n} < 2^{-40n}.
\end{align*}
Note that $p_i$ is at least $1/2^n$. Hence, the output $\hat{p}$ satisfies \eqref{eq:tar}.

\begin{remark*}
In the above proof, while the TV-distance between the two Gibbs distributions \(\mu^{(i)}\) and \(\nu^{\alpha}\) is large (because their parameter distance is 1), the TV distance between their marginal distributions at vertex \(i\) can be arbitrarily small. This highlights the distinction between the TV-distance of marginal distributions and the TV-distance of the entire distribution.
\end{remark*}}

\ifthenelse{\boolean{conf}}{\paragraph{Organization of the Appendix}}{\paragraph{Organization of the paper}}
\Cref{lem:TV-lower} is proved in \Cref{sec:lower}. The additive error approximation algorithm is given in \Cref{sec:add}. The algorithm for instances with small parameter distance in given in \Cref{sec:alg-main}. In \Cref{sec:proof-main}, we put all the pieces together to prove all algorithmic results. 
\ifthenelse{\boolean{conf}}{Our hardness result in \Cref{thm:many-vertex} is proved in \Cref{sec:hard}.}{Our
 hardness results on approximating TV-distance for marginal distributions is proved in \Cref{sec:hard}.}   

\ifthenelse{\boolean{conf}}{\section*{Acknowledgment}
Weiming Feng and Hongyang Liu gratefully acknowledge the support of ETH Z\"urich, where part of this work was conducted.  
Weiming Feng acknowledges the support of Dr.\ Max R\"ossler, the Walter Haefner Foundation, and the ETH Z\"urich Foundation during his affiliation with ETH Z\"urich.}


\section{Parameter distance v.s. total variation distance}\label{sec:lower}
In this section, we prove \Cref{lem:TV-lower}. We first prove the result for the hardcore model in \Cref{sec:hardcore-proof}, and then prove the result for the soft-Ising model in \Cref{sec:Ising-proof}. 
\subsection{Analysis for the hardcore model}\label{sec:hardcore-proof}
Recall that our problem setting is: Let $G=(V,E)$ be a graph, and $(G,\lambda^\mu)$ and $(G,\lambda^\nu)$ are two hardcore models on the same graph, satisfying the uniqueness condition in~\eqref{eq:cond-hardcore}. Let $\mu$ and $\nu$ denote distributions of $(G,\lambda^\mu)$ and $(G,\lambda^\nu)$ respectively. The parameter distance $\dis(\mu,\nu)$ for hardcore model is defined by $\dis(\mu,\nu) \defeq \Vert \lambda^\mu - \lambda^\nu \Vert_{\infty}$.
Now we prove the first hardcore model part of \Cref{lem:TV-lower}: 
\[\DTV{\mu}{\nu} \geq \CC\cdot \dis(\mu,\nu), \quad\text{where } \CC = \frac{1}{5000}.\]

Let $i \in V$ be one vertex such that $|\lambda_i^\mu-\lambda_i^\nu|=\dis(\mu,\nu)$. Without loss of generality, we can assume that $\lambda_i^\mu-\lambda_i^\nu=\dis(\mu,\nu)$. Otherwise, we can flip the roles of $\mu$ and $\nu$ in the following proof.
Define two collections of independent sets of $V$:
\begin{align*}
    H_1:&=\{S\text{ is an independent set of }G\mid i\in S\},\\
    H_2:&=\{S \setminus \{i\} \mid S \in H_1 \}.
\end{align*}
In this proof, we use $\mu(S)$ to denote the probability of $\sigma \in \{\pm\}^V$ in $\mu$ such that $\sigma_i = +1$ if and only if $i\in S$.
Define $\mu(H_1) = \sum_{S \in H_1} \mu(S)$ and $\mu(H_2),\nu(H_1),\nu(H_2)$ in a similar way.
It is easy to verify that 
\begin{align*}
\frac{\mu(H_1)}{\mu(H_2)}=\lambda_i^\mu ,\quad \frac{\nu(H_1)}{\nu(H_2)}=\lambda_i^\nu.
\end{align*}

Then, we consider two cases depending on the value of $|\mu(H_2)-\nu(H_2)|$. The simple case is $|\mu(H_2)-\nu(H_2)|\geq \CC\cdot \dis(\mu,\nu)$, then $\DTV{\mu}{\nu} \geq|\mu(H_2)-\nu(H_2)|\geq \CC\cdot \dis(\mu,\nu)$.

The main case is $|\mu(H_2)-\nu(H_2)|< \CC\cdot \dis(\mu,\nu)$. In this case we show that $|\mu(H_1)-\nu(H_1)|\geq \CC\cdot \dis(\mu,\nu)$, which also implies $\DTV{\mu}{\nu} \geq|\mu(H_1)-\nu(H_1)|\geq \CC\cdot \dis(\mu,\nu)$.
We can lower bound the value of $|\mu(H_1)-\nu(H_1)|$ as follows:
\begin{align}\label{eq:mu-nu-diff}
\mu(H_1)-\nu(H_1)&=\lambda_i^\mu \cdot \mu(H_2)-\lambda_i^\nu \cdot \nu(H_2)\notag\\
&=(\lambda_i^\nu+\dis(\mu,\nu)) \mu(H_2)-\lambda_i^\nu \cdot \nu(H_2)\notag\\
&=\lambda_i^\nu(\mu(H_2)-\nu(H_2))+\dis(\mu,\nu) \cdot \mu(H_2)\notag\\
&>\dis(\mu,\nu) \cdot \mu(H_2)-\lambda_i^\nu \cdot \CC \cdot \dis(\mu,\nu).
\end{align}

Because $(G,\lambda^\mu)$ and $(G,\lambda^\nu)$ satisfy the uniqueness condition, we have for all vertex $ i\in V$, all  $x\in \{\mu,\nu\}$, $\lambda_i^x \leq \lambda_c(\Delta) = \frac{(\Delta-1)^{\Delta - 1}}{(\Delta-2)^\Delta}\leq 4$ because $\Delta \geq 3$.
Let $N(i)$ denotes the set of neighbors of $i$ in graph $G=(V,E)$. 
By definition of $H_2$, $\mu(H_2)$ is the probability of $j \notin S$ for all $j \in N(i) \cup \{i\}$ and $S \sim \mu$ is a random independent set from the hardcore model $(G,\lambda^\mu)$. Suppose there is a total ordering $<$ among all vertices in $V$. We have
\begin{align}\label{eq:mu-H2}
\mu(H_2)&= \Pr[S \sim \mu]{ i \notin S} \prod_{j \in N(i)}\Pr[S \sim \mu]{ j \notin S \mid (i \notin S) \land (\forall k \in N(i) \text{ with } k < j, k \notin S )} \notag\\
&\geq \frac{1}{1+\lambda_i^\mu}\prod_{j \in N(i)} \left(\frac{1}{1+\lambda_j^\mu}\right)\geq \left( \frac{1}{1+4/(\Delta-2)}\right)^{\Delta+1}\notag\\
&=\left( 1-\frac{4}{\Delta+2}\right)^{\Delta+1}> \frac{1}{1000}.
\end{align}
Also note that $\lambda_i^\nu \leq 4$ by the uniqueness condition. Recall $\CC = \frac{1}{5000}$. By~\eqref{eq:mu-nu-diff} and~\eqref{eq:mu-H2},  we have
\begin{align*}
\mu(H_1)-\nu(H_1)&>\dis(\mu,\nu) \cdot \mu(H_2)-\lambda_i^\nu \cdot \CC \cdot \dis(\mu,\nu)\\
&>\dis(\mu,\nu)\left(\frac{1}{1000}-\frac{4}{5000}\right)=\CC\cdot \dis(\mu,\nu).
\end{align*}

Next, we prove the second part of the hardcore model with $\CC = b^3$. Now, we do not have the uniqueness condition but we have a marginal lower bound $b$ and the the hardcore model is soft.
The proof is similar.
We only need to show how to lower bound the value in~\eqref{eq:mu-nu-diff}.
We first upper bound $\lambda^\nu_i$ in \eqref{eq:mu-nu-diff}. Consider the pinning that all neighbors of $i$ are fixed as $-1$, then the probability of $i$ taking $-1$ is $\frac{1}{1+\lambda^\nu_i} \geq b$. Thus, $\lambda_i^v \leq \frac{1-b}{b}$. Next, we lower bound the value of $\mu(H_2)$.
Let $N(i)$ denote all neighbors of $i$.
Note that $i$ takes value $+$ only if all neighbors of $i$ take the value $-$. Recall that we assume that $\lambda_i^\mu-\lambda_i^\nu=\dis(\mu,\nu)$ in the beginning of the proof. 
We can assume that $\dis(\mu,\nu) > 0$, otherwise $\DTV{\mu}{\nu} = 0$ and the lemma is trivial. 
We have $\lambda^\mu_i > 0$ so that $\mu_i(+) > 0$. 
By the marginal lower bound, $\mu_i(+) \geq b$.
The vertex $i$ takes value $+$ only if all neighbors take the value $-$. We have
\begin{align*}
    b\leq \mu_i(+) \leq \Pr[S\sim \mu]{ \forall j \in N(i), j \notin S}. 
\end{align*}
On the other hand, the value of $\mu(H_2)$ can be lower bound by
\begin{align*}
    \mu(H_2) &= \Pr[S \sim \mu]{\forall j \in N(i) \cup \{i\}, j \notin S}\\
    &= \Pr[S\sim \mu]{ \forall j \in N(i), j \notin S}\Pr[S \sim \mu]{i \notin S \mid \forall j \in N(i), j \notin S}\\
    &\geq b^2.
\end{align*}
We can set $\CC = b^3$ so that
\begin{align*}
\mu(H_1)-\nu(H_1)&>\dis(\mu,\nu) \cdot \mu(H_2)-\lambda_i^\nu \cdot \CC \cdot \dis(\mu,\nu)\\
&\geq\dis(\mu,\nu)\left(b^2-\frac{1-b}{b} \cdot b^3\right)=\CC\cdot \dis(\mu,\nu).
\end{align*}


\subsection{Analysis for the soft-Ising model} \label{sec:Ising-proof}

Recall that our problem setting is: Let $G=(V,E)$ be a graph, 
and $(G,J^\mu,h^\mu)$ and $(G,J^\nu,h^\nu)$ are two soft-Ising models on the same graph.
Let $\mu$ and $\nu$ denote distributions of $(G,J^\mu,h^\mu)$ and $(G,J^\nu,h^\nu)$ respectively. 
The parameter distance $\dis(\mu,\nu)$ for soft-Ising model is defined by 
\[\dis(\mu,\nu) \defeq \max\left\{ \Vert J^\mu - J^\nu \Vert_{\max}, \max_v \frac{|h^\mu_v-h^\nu_v|}{\deg_v+1} \right\}.\]

Now we prove the soft-Ising model part of \Cref{lem:TV-lower}: 
if both $\mu$ and $\nu$ are $b$-marginally bounded for $0<b< 1$, then
\[\DTV{\mu}{\nu} \geq f(b) \cdot \dis(\mu,\nu), 
\quad\text{where } f(b) = \frac{1}{2}b^2.\]
Before proving the above, we first present the following lemma.
\begin{lemma}
\label{lem:local-global-dtv}
Let $\mu$ and $\nu$ be any two distributions on $\{\pm\}^V$, and $S\subseteq V$ be a subset of vertices.
Let $0<\delta<1$. 
If for any $\sigma\in \{\pm\}^S$, $\DTV{\mu^\sigma}{\nu^\sigma}\geq \delta$, then $\DTV{\mu}{\nu}\geq \delta/2$.
\end{lemma}

\begin{proof}
For any two distributions $p,q$ on space $\Omega$, by the definition of total variation distance,
\begin{align*}
\sum_{x\in \Omega}\min(p(x),q(x)) &= 1-\DTV{p}{q},\\
\sum_{x\in \Omega}\max(p(x),q(x)) &= 1+\DTV{p}{q}.
\end{align*}

Under the assumption of the lemma, we have
\begin{align*}
1-\DTV{\mu}{\nu}&=\sum_{\sigma\in \{\pm\}^V}\min(\mu(\sigma),\nu(\sigma))
\\
&=\sum_{\sigma\in \{\pm\}^S}\sum_{\tau\in \{\pm\}^{V\setminus S}}\min(\mu(\sigma)\mu(\tau\mid\sigma),\nu(\sigma)\nu(\tau\mid\sigma))
\\
&\leq \sum_{\sigma\in \{\pm\}^S}\max(\mu(\sigma),\nu(\sigma))\sum_{\tau\in \{\pm\}^{V\setminus S}}\min(\mu(\tau\mid\sigma),\nu(\tau\mid\sigma))
\\
(\text{since }\DTV{\mu^\sigma}{\nu^\sigma}\geq \delta) \quad 
&\leq \sum_{\sigma\in \{\pm\}^S}\max(\mu(\sigma),\nu(\sigma))(1-\delta) 
\\
&= (1+\DTV{\mu_S}{\nu_S})(1-\delta),
\\
&\leq 1-\delta+\DTV{\mu_S}{\nu_S}
\end{align*}    
which implies
\begin{align}
\label{eq:local-global-dtv-1}
\DTV{\mu}{\nu}\geq \delta - \DTV{\mu_S}{\nu_S}.
\end{align}
On the other hand, we have 
\begin{align}
\label{eq:local-global-dtv-2}
\DTV{\mu}{\nu} \geq \DTV{\mu_S}{\nu_S}.
\end{align}
Then \Cref{lem:local-global-dtv} follows by combining~\eqref{eq:local-global-dtv-1} and~\eqref{eq:local-global-dtv-2}.
\end{proof}

According to the definition of the parameter distance, there are $2$ cases:
\begin{itemize}
\item It exists $\{u,v\}\in E$ such that $|J^\mu_{u,v}-J^\nu_{u,v}|=\dis(\mu,\nu)$.
\item It exists $v\in V$ such that $|h_v^\mu-h_v^\nu|=\dis(\mu,\nu) \cdot (\deg_v+1)$.
\end{itemize}

For the first case, we will show that
\begin{align}
\label{eq:Ising-dtv-lowerbound-1}
\forall \sigma \in \{\pm\}^{V\setminus\{u,v\}}, \DTV{\mu^\sigma}{\nu^\sigma}\geq 2f(b)\dis(\mu,\nu),
\end{align}
and for the second case, we will show that
\begin{align}
\label{eq:Ising-dtv-lowerbound-2}
\forall \sigma \in \{\pm\}^{V\setminus\{v\}}, \DTV{\mu^\sigma}{\nu^\sigma} \geq 2f(b)\dis(\mu,\nu).
\end{align}

Assuming~\eqref{eq:Ising-dtv-lowerbound-1} and~\eqref{eq:Ising-dtv-lowerbound-2} hold, using \Cref{lem:local-global-dtv}, 
we can directly derive the result of the total variation distance lower bound for soft-Ising model claimed in \Cref{lem:TV-lower}.

Next, we will prove~\eqref{eq:Ising-dtv-lowerbound-1} and~\eqref{eq:Ising-dtv-lowerbound-2},
respectively to complete the entire proof.

\begin{proof}\ifthenelse{\boolean{conf}}{\textbf{of eq.\eqref{eq:Ising-dtv-lowerbound-1}}}{[Proof of eq.\eqref{eq:Ising-dtv-lowerbound-1}]}
Recall that $\mu,\nu$ are both $b$-marginal bounded which is defined in \Cref{def:marlow}. Thus for any $\sigma \in \{\pm\}^{V\setminus \{u,v\}}$, for any $\tau \in \{\pm\}^{\{u,v\}}$, we have $\mu^\sigma(u=\tau_u,v=\tau_v) = \mu^\sigma_u(\tau_u)\cdot \mu^{\sigma\wedge u\gets\tau_u}_v(\tau_v) \geq b^2$. The same local lower bound also holds for the distribution $\nu$.

Define $c^\mu_u=h^\mu_u+\sum_{\{u,w\}\in E, w\neq v}J^\mu_{u,w}\sigma_w$ and $c^\mu_v=h^\mu_v+\sum_{\{v,w\}\in E, w\neq u}J^\mu_{v,w}\sigma_w$ be the influence coefficient of the external field $\sigma$ on $u,v$ in the distribution $\mu$. Similarly, let $c^\nu_u$ and $c^\nu_v$ denote the corresponding influence coefficients in the distribution $\nu$. 

According to the definition of the Ising model, when the external field $\sigma$ is fixed, the local distribution of $u,v$ in the distribution $\mu$ is follows:
\begin{align*}
    \mu^\sigma(u=+,v=+) &= \exp(J^\mu_{u,v}+c^\mu_u+c^\mu_v-s^\mu),
    \\
    \mu^\sigma(u=+,v=-) &= \exp(-J^\mu_{u,v}+c^\mu_u-c^\mu_v-s^\mu),
    \\
    \mu^\sigma(u=-,v=+) &= \exp(-J^\mu_{u,v}-c^\mu_u+c^\mu_v-s^\mu),
    \\
    \mu^\sigma(u=-,v=-) &= \exp(J^\mu_{u,v}-c^\mu_u-c^\mu_v-s^\mu),
\end{align*}
where $s^\mu = \log(\exp(J^\mu_{u,v}+c^\mu_u+c^\mu_v)+\exp(-J^\mu_{u,v}+c^\mu_u-c^\mu_v))+\exp(-J^\mu_{u,v}-c^\mu_u+c^\mu_v)+\exp(J^\mu_{u,v}-c^\mu_u-c^\mu_v))$. 

Using the same method, we define $s^\nu$, and the local distribution of $u,v$ in the distribution $\nu$ can be expressed in terms of $c^\nu_u,c^\nu_v$ and $s^\nu$, as in the above equation.

Let, $p=J^\mu_{u,v}+c^\mu_u+c^\mu_v-s^\mu$ and $q=J^\nu_{u,v}+c^\nu_u+c^\nu_v-s^\nu$, then 
\begin{align}\label{eq:Ising-dtv-exp-ineq}
    \DTV{\mu^\sigma}{\nu^\sigma} &\geq |\mu^\sigma(u=+,v=+)-\nu^\sigma(u=+,v=+)|
    \nonumber \notag\\&= |\exp(p)-\exp(q)| =\int_{x=\min\{p,q\}}^{\max\{p,q\}} \exp(x)dx
    \nonumber \notag\\&\geq  \int_{x=\min\{p,q\}}^{\max\{p,q\}} \exp(p)dx =\exp(p)\cdot |p-q|
    \nonumber \notag\\{\text{(by marginal lower bound)}}\quad&\geq b^2\cdot |p-q|.
\end{align}

Based on this, we can derive the following lower bound on the total variation distance between $\mu^\sigma$ and $\nu^\sigma$:
\begin{align}
\label{eq:Ising-dtv-local-edge-1}
\DTV{\mu^\sigma}{\nu^\sigma} &\geq b^2 \cdot |(J^\mu_{u,v}+c^\mu_u+c^\mu_v-s^\mu)-(J^\nu_{u,v}+c^\nu_u+c^\nu_v-s^\nu)|.
\end{align}

Using the same method, consider the distribution differences in $\mu^\sigma$ and $\nu^\sigma$ for the remaining three cases $(u=+,v=-),(u=-,v=+)$ and $(u=-,v=-)$, we can also obtain the following lower bounds:
\begin{align}
\label{eq:Ising-dtv-local-edge-2}
\DTV{\mu^\sigma}{\nu^\sigma} &\geq b^2 \cdot |-(-J^\mu_{u,v}+c^\mu_u-c^\mu_v-s^\mu)+(-J^\nu_{u,v}+c^\nu_u-c^\nu_v-s^\nu)|,
\\
\label{eq:Ising-dtv-local-edge-3}
\DTV{\mu^\sigma}{\nu^\sigma} &\geq b^2 \cdot |-(-J^\mu_{u,v}-c^\mu_u+c^\mu_v-s^\mu)+(-J^\nu_{u,v}-c^\nu_u+c^\nu_v-s^\nu)|,
\\
\label{eq:Ising-dtv-local-edge-4}
\DTV{\mu^\sigma}{\nu^\sigma} &\geq b^2 \cdot |(J^\mu_{u,v}-c^\mu_u-c^\mu_v-s^\mu)-(J^\nu_{u,v}-c^\nu_u-c^\nu_v-s^\nu)|.
\end{align}

Note that the absolute value operation satisfies the triangle inequality, by combining inequalities \eqref{eq:Ising-dtv-local-edge-1}, \eqref{eq:Ising-dtv-local-edge-2}, \eqref{eq:Ising-dtv-local-edge-3} and \eqref{eq:Ising-dtv-local-edge-4}, we have
\begin{align*}
4\DTV{\mu^\sigma}{\nu^\sigma} \geq b^2\cdot 4|J^{\mu}_{u,v}-J^{\nu}_{u,v}|.
\end{align*}

Then for the case that $|J^\mu_{u,v}-J^\nu_{u,v}|=\dis(\mu,\nu)$, the lower bound of local total variation distance in \eqref{eq:Ising-dtv-lowerbound-1} is proven.
\end{proof}

\begin{proof}\ifthenelse{\boolean{conf}}{\textbf{of eq.\eqref{eq:Ising-dtv-lowerbound-2}}}{[Proof of eq.\eqref{eq:Ising-dtv-lowerbound-2}]}
For the case that $|h_v^\mu-h_v^\nu|=\dis(\mu,\nu) \cdot (\deg_v+1)$. Define $c^\mu=h^\mu_v+\sum_{\{u,v\}\in E}J^\mu_{u,v}\sigma_v$ and $c^\nu=h^\nu_v+\sum_{\{u,v\}\in E}J^\nu_{u,v}\sigma_v$ be the influence
coefficient of the external field $\sigma$ on $v$ in the distribution $\mu$ and $\nu$ respectively, note that for each $\{u,v\}\in E$, $|J^\mu_{u,v}-J^\nu_{u,v}|$ is bounded by $\dis(\mu,\nu)$, then 
\begin{align*}
|c^\mu-c^\nu| &\geq |h^\mu_v-h^\nu_v|-\sum_{\{u,v\}\in E}|J^\mu_{u,v}-J^\nu_{u,v}| \\&\geq (\deg_v+1)\dis(u,v)- \sum_{\{u,v\}\in E}\dis(\mu,\nu) = \dis(\mu,\nu).
\end{align*}

According to the definition of the Ising model, when the external field $\sigma$ is fixed, the local distribution on $v$ in the distribution $\mu$ is follows:
\begin{align*}
    \mu^\sigma(v=+) = \exp(c^\mu-s^\mu), \quad \mu^\sigma(v=-) = \exp(-c^\mu-s^\mu),
\end{align*}
where $s^\mu = \log (\exp(c^\mu)+\exp(-c^\mu))$. 

Using the same method, we define $s^\nu$, and the local distribution of $u,v$ in the distribution $\nu$ can be expressed in terms of $c^\nu$ and $s^\nu$, as in the above equation. 

Recall that $\mu,\nu$ are both $b$-marginal bounded, applying the same method in \eqref{eq:Ising-dtv-exp-ineq}, we can derive the following lowerbound on the total variation distance between $\mu^\sigma$ and $\nu^\sigma$:
\begin{align}
\label{eq:Ising-dtv-local-vertex-1}
\DTV{\mu^\sigma}{\nu^\sigma} &\geq b \cdot |(c^\mu-s^\mu)-(c^\nu-s^\nu)|,
\\
\label{eq:Ising-dtv-local-vertex-2}
\DTV{\mu^\sigma}{\nu^\sigma} &\geq b \cdot |-(-c^\mu-s^\mu)+(-c^\nu-s^\nu)|.
\end{align}

By combining the above inequalities, we have 
\begin{align*}
2\DTV{\mu^\sigma}{\nu^\sigma} \geq b \cdot 2|c^\mu-c^\nu| \geq 2b\cdot \dis(\mu^\sigma,\nu^\sigma).
\end{align*}
Note $0<b<1$, then the lower bound of local total variation distance in \eqref{eq:Ising-dtv-lowerbound-2} is proven.
\end{proof}

\section{Additive-error approximation algorithm}\label{sec:add}

\subsection{TV-distance between two Gibbs distributions}

We first present an algorithm can achieves the additive-error approximation to the total variation distance between two  \emph{general} Gibbs distributions, which covers Ising and hardcore models as special cases.
%
Let $\mu$ over $\{0,1\}^V$ be a Gibbs distribution over graph $G=(V,E)$. For any configuration $\sigma \in \{0,1\}^V$, $\mu(\sigma) = w_{\mu}(\sigma)/Z_\mu$, where $w_{\mu}(\cdot)$ is the weight function and $Z_\mu = \sum_{\tau \in \{0,1\}^V}w_\mu(\sigma)$.
Let $\TW_{G} \in \mathbb{N}$.
We say the Gibbs distribution $\mu$ admits a weight oracle with cost $\TW_{G}$ is given any $\sigma \in [q]^V$, it returns the exact weight $w_\mu(\sigma)$ in time $\TW_{G}$. 
Note that both Ising and hardcore models, as well as most Gibbs distributions, admit weight oracle with cost $\TW_{G} = O(|V|+|E|)$.
Recall the sampling and approximate counting oracles are defined in~\Cref{def:oracle}.

\begin{theorem}\label{thm:Approximate-Gibbs}
There exists an algorithm such that, given two general Gibbs distributions $\mu$ and $\nu$ on the same graph $G=(V,E)$ and an error bound $\epsilon > 0$, if $\mu$ and $\nu$ both admit weight, sampling, and approximate counting oracles with cost $\TW_{G}$ and cost functions $\TS_G(\cdot)$ and $\TC_G(\cdot)$ respectively, then it returns a random $\hat{d}$ in time $O(\TC_G(\frac{\epsilon}{4})+ \frac{1}{\epsilon^2}(\TW_{G} + \TS_G(\frac{\epsilon}{4})))$ such that
\begin{align*}
    \Pr[]{\DTV{\mu}{\nu} - \epsilon \leq \hat{d} \leq \DTV{\mu}{\nu} + \epsilon} \geq \frac{2}{3}.
\end{align*}
\end{theorem}


\begin{proof}
Define a random variable $X \in [0,1]$ such that $X = \max\left(0,1-\frac{\nu(\sigma)}{\mu(\sigma)}\right)$, where $\sigma \sim \mu$.
\begin{align*}
 \DTV{\mu}{\nu} &= \sum_{\sigma:\mu(\sigma)>\nu(\sigma)}|\mu(\sigma)-\nu(\sigma)| =   \sum_{\sigma:\mu(\sigma)>\nu(\sigma)}\mu(\sigma)\left|1-\frac{\nu(\sigma)}{\mu(\sigma)}\right|\\
 &=\sum_{\sigma}\mu(\sigma)\max\left(0,1-\frac{\nu(\sigma)}{\mu(\sigma)}\right) = \E[]{X}.
\end{align*}
By the definition of $X$, we have $0\leq X \leq 1$, so $\Var[]{X}\leq 1$.
Ideally, we want to draw independent samples of $X$ and take average to approximate $ \DTV{\mu}{\nu}$. However, the main issue is that given a $\sigma \sim \mu$, we cannot compute neither $\mu(\sigma)$ nor $\nu(\sigma)$ exactly. Alternatively, we will define another random variable $\hat{X} \in [0,1]$ that approximate the random variable $X$. 

Call the approximate counting oracles of $\mu$ and $\nu$ to obtain $\hat{Z_\mu}$ and $\hat{Z_\nu}$ that approximate partition functions $Z_\mu$ and $Z_\nu$ with relative error bound $\frac{\epsilon}{4}$. 
We may assume both counting oracles succeed, which happens with probability at least 0.98.
Define the random variable $\hat{X} \in [0,1]$ by the following process.
\begin{enumerate}
    \item Call sampling oracle of $\mu$ to obtain one sample $\sigma \in \{\pm\}^V$ such that $\DTV{\sigma}{\mu}\leq \frac{\epsilon}{4}$.\label{step-I}
    \item Call weight oracles of both $\mu$ and $\nu$ to obtain exact weights $w_\mu(\sigma), w_\nu(\sigma)$. Compute $\hat{\mu}(\sigma)={w_\mu(\sigma)}/{\hat{Z_\mu}}$ and $\hat{\nu}(\sigma)={w_\nu(\sigma)}/{\hat{Z_\nu}}$.
    \item Define $\hat{X}=\max(0,1-{\hat{\nu}(\sigma)}/{\hat{\mu}(\sigma)})$, in particular, $\hat{X} = 0$ if $\hat\mu(\sigma) = 0$.
\end{enumerate}
Let $T = \frac{64}{\epsilon^2}$.
Our algorithm is the following simple process:
\begin{itemize}
    \item Draw $T$ independent samples $\hat{X}_1,\hat{X}_2,\ldots,\hat{X}_T$ of random variable $\hat{X}$.
    \item Output the average $\hat d = \frac{1}{T}\sum_{i=1}^T \hat{X}_i$.
\end{itemize}
It is easy to see the running time of our algorithm is $2\TC_G(\frac{\epsilon}{4})+T\cdot (\TS_G(\frac{\epsilon}{4}) + 2\TW_G +O(1))$.


    


    


To prove the correctness of our algorithm, we only need to show that
\begin{align}\label{eq:ex-error}
    \left\vert \E[]{\hat{X}} - \E[]{X} \right\vert = \left\vert \E[]{\hat{X}} - \DTV{\mu}{\nu} \right\vert \leq \frac{7\epsilon}{8}.
\end{align}
Note that $0 \leq \hat{X} \leq 1$ so that $\Var{\hat X} \leq 1$. By Hoeffding’s inequality, it is easy to show that with probability at least 0.9, $|\hat d - \mathbb{E}[\hat{X}] | \leq \epsilon/8$. Combining with~\eqref{eq:ex-error} proves the theorem. 

Now, we only need to verify~\eqref{eq:ex-error}.
We introduce a new random variable $X^*$ in analysis.
In the definition of $\hat{X}$,
assume we replace the sample $\sigma$ in \Cref{step-I} with a  perfect sampler of the distribution $\mu$. Let $X^*$ denote the resulting random variable. We first compare $\E[]{X^*}$ with $\E[]{X}$. The difference between $X^*$ and $X$ comes from the error of computing the ratio of $\mu(\sigma)$ and $\nu(\sigma)$.
Note that $\frac{\nu(\sigma)}{\mu(\sigma)} = \frac{w_\nu(\sigma)}{w_\mu(\sigma)}\cdot \frac{Z_\mu}{Z_\nu}$ and $\frac{\hat\nu(\sigma)}{\hat\mu(\sigma)} = \frac{w_\nu(\sigma)}{w_\mu(\sigma)}\cdot \frac{\hat Z_\mu}{\hat Z_\nu}$. By the definition of approximate counting oracle, for $\sigma$ with $\mu(\sigma) >0$,
\begin{align*}
 \left(1-\frac{5\epsilon}{8}\right)\frac{\nu(\sigma)}{\mu(\sigma)}   \leq  \frac{\hat{\nu}(\sigma)}{\hat{\mu}(\sigma)} \leq \left(1+\frac{5\epsilon}{8}\right)\frac{\nu(\sigma)}{\mu(\sigma)}.
\end{align*}
We can compute the expectation as 
\begin{align*}
    \E[]{X^*}&= \sum_{\sigma:\mu(\sigma)>0}\mu(\sigma)\max\left(0,1-\frac{\hat{\nu}(\sigma)}{\hat{\mu}(\sigma)}\right)  \leq\sum_{\sigma:\mu(\sigma)>0}\mu(\sigma)\max\left(0,1-\left(1-\frac{5\epsilon}{8}\right)\frac{\nu(\sigma)}{\mu(\sigma)}\right)\\
&\leq\sum_{\sigma:\mu(\sigma)>0}\mu(\sigma)\left(\max\left(0,1-\frac{\nu(\sigma)}{\mu(\sigma)}\right)+\frac{5\epsilon}{8}\frac{\nu(\sigma)}{\mu(\sigma)}\right)
\leq \DTV{\mu}{\nu}+\frac{5\epsilon}{8}\sum_{\sigma:\mu(\sigma)>0}\nu(\sigma)\\
    &\leq \DTV{\mu}{\nu}+\frac{5\epsilon}{8}.
\end{align*}
Using the same way, we could verify the other direction. We have 
\begin{align}\label{eq:ex-1}
|\E[]{X^*}-\E[]{X}| = |\mathbf{E}[X^*]-\DTV{\mu}{\nu}|\leq \frac{5\epsilon}{8}.    
\end{align}

Next, we compare $X^*$ to $\hat X$. The only difference is that they sample $\sigma$ from different distributions. Let $\mu'$ be the distribution defined by the approximate sampling oracle.
Define a function $f$ such that for any $\sigma \in \{\pm\}^V$, $f(\sigma) = \max(0,1-{\hat{\nu}(\sigma)}/{\hat{\mu}(\sigma)})$, where we set $f(\sigma) = 0$ if $\hat \mu (\sigma) = 0$.
We have $|\mathbf{E}(\hat X)-\mathbf{E}(X^*)|=\left | \mathbf{E}_{\sigma\sim \mu}[f(\sigma)]-\mathbf{E}_{\sigma\sim \mu'}[f(\sigma)] \right |$. 
Define $A = \{\sigma\mid \mu(\sigma)>\mu'(\sigma)\}$ and $B = \{\sigma \mid \mu(\sigma)<\mu'(\sigma)\}$.
We can write
\begin{align}\label{eq:ex2}
    |\mathbf{E}(\hat X)-\mathbf{E}(X^*)| &=
     \left | \sum_{\sigma \in A} (\mu'(\sigma)-\mu(\sigma))f(\sigma) + \sum_{\sigma \in B} (\mu'(\sigma)-\mu(\sigma))f(\sigma) \right|\notag\\
\text{($\star$)}\quad    &\leq \max \tp{\sum_{\sigma \in A} (\mu(\sigma)-\mu'(\sigma))f(\sigma),\sum_{\sigma \in B} (\mu'(\sigma)-\mu(\sigma))f(\sigma)}\notag\\
\text{(by $f(\sigma) \in [0,1]$)}\quad    &\leq \DTV{\mu}{\mu'}\leq \frac{\epsilon}{4},
\end{align}
where in inequality ($\star$), we use $\sum_{\sigma \in A} (\mu'(\sigma)-\mu(\sigma))f(\sigma) \leq 0$ and $\sum_{\sigma \in B} (\mu'(\sigma)-\mu(\sigma))f(\sigma) \geq 0$.
Finally,~\eqref{eq:ex-error} holds due to~\eqref{eq:ex-1} and \eqref{eq:ex2}.
\end{proof}

\Cref{thm:Approximate-Gibbs} implies the following corollary for concrete models.
\begin{corollary}\label{corollary:apps}
There exist an FPRAS for approximating TV-distances with additive error for following models: (1) Hardcore model satisfying uniqueness condition; (2) Ising model with spectral condition; (3) Ferromagnetic interaction with consistent field condition; and (4) Anti-ferromagnetic interaction at or within the uniqueness threshold.
\end{corollary}
The definitions of the conditions can be found in~\eqref{eq:cond-hardcore} and \Cref{cond:Ising}. 
In contrast to the relative-error approximation, the above corollary does not require a marginal lower bound for the Ising model.
Furthermore, since the proof of \Cref{thm:Approximate-Gibbs} does not use \Cref{lem:TV-lower}, \Cref{corollary:apps} holds for general (not necessarily soft) Ising models.

\subsection{TV-distance between two marginal distributions}

In this subsection, we present an algorithm that can achieves the additive-error approximation to the total variance distance between two marginal distributions. Let $\mu$ over $\{\pm\}^V$ be a Gibbs distribution over graph $G=(V,E)$, and $S$ is a subset of $V$. For any configuration $\sigma$, let $\sigma$ be a partial configuration over $\{\pm\}^S$. Recall that $\mu_S(\sigma)  \propto Z^{\sigma}$, where $Z^{\sigma}$ is the conditional partition function of $\sigma$ defined as
\begin{align*}
    Z^{\sigma}\defeq\sum_{\tau \in \{\pm\}^V: \tau_S = \sigma } w_\mu(\tau).
\end{align*}

\begin{definition}[approximate conditional counting oracle]\label{def:cond-count-oracle}
Let $\mathbb{S}$ be a spin system on graph G with Gibbs distribution $\mu$. 
Let $\TC_{G}:(0,1) \to \mathbb{N}$ be a function.
We say $\mathbb{S}$ admits a conditional counting oracle with cost function $\TC_{G}(\cdot)$ if given any $0<\epsilon<1$, and any partial configuration $\sigma \in \{\pm\}^S$ on a subset $S\subseteq V$, it returns a random number $\hat{Z}_\mu^{\sigma}$ in time $\TC_{G}(\epsilon)$ such that $Z_\mu^{\sigma}(1-\epsilon)\leq \hat{Z}_\mu^{\sigma}(\sigma) \leq Z_\mu^{\sigma}(1+\epsilon)$ with probability at least 0.99.
\end{definition}

The oracle above is stronger than the approximate counting oracle in \Cref{def:oracle}. The approximate counting oracle only answers the query for $S = \emptyset$, while the conditional counting oracle can answer the query for any subset $S\subseteq V$.

\begin{theorem}\label{thm:approx-margin-tv}
There exists an algorithm such that, given two general Gibbs distributions $\mu$ and $\nu$ on the same graph $G=(V,E)$ with $n = |V|$, any subset $S \subseteq V$, and any $\epsilon > 0$, if $\mu$ and $\nu$ both admit sampling and conditional counting oracles with cost functions $\TS_G(\cdot)$ and $\TC_{G}(\cdot)$ respectively, then it returns a random number $\hat{d}$ in time $O(\frac{1}{\epsilon^2}\log\frac{1}{\epsilon})\cdot(\TC_{G}(\frac{\epsilon}{8})+\TS_{G}(\frac{\epsilon}{8}))$ such that \[\Pr{\big|\hat{d} - \DTV{\mu_S}{\nu_S}\big|\leq \epsilon} \geq \frac{2}{3}.\]
\end{theorem}

\begin{proof}
Define a random variable $Y\in [0,1]$ such that $Y=\max \left(0,1-\frac{\nu_S(\sigma)}{\mu_S(\sigma)}\right)$, where $\sigma \sim \mu_S$.

\begin{align*}
\DTV{\mu_S}{\mu_S}&=\sum_{\sigma \in \{\pm\}^S}\max\left( 0,\mu_S(\sigma)-\nu_S(\sigma) \right)\\
&=\sum_{\sigma \in \{\pm\}^S}\mu_S(\sigma)\max\left(0,1-\frac{\nu_S(\sigma)}{\mu_S(\sigma)} \right)=\E[]{Y}.
\end{align*}

The random variable $Y$ satisfies that $0\leq Y \leq 1$ so $\Var{Y}\leq 1$. Similar to the proof of \Cref{thm:Approximate-Gibbs}, we want to draw independent sample of $Y$ and take average to approximate $\DTV{\mu_S}{\nu_S}$. However, here $\mu_S(\sigma) = Z^\sigma_\mu/Z_{\mu}$ and $\nu_S(\sigma) = Z^\sigma_\nu/Z_{\nu}$. An additional problem is that we cannot exactly compute the weight $Z^\sigma_\mu$ and $Z^\sigma_\nu$ for each partial configuration $\sigma\sim \mu_S$.

First note that we can boost the success probability of conditional counting oracle from $0.99$ to $1 - \delta$ by calling it independently for $O(\log \frac{1}{\delta})$ times and take the median. Let $\delta = \frac{\epsilon^2}{320}$.
Call conditional counting oracles with $S=V$ to obtain $\hat{Z_\mu}$ and $\hat{Z_\nu}$ that approximate partition functions $Z_\mu$ and $Z_\nu$ with relative error bound $\frac{\epsilon}{8}$. 
We may assume both counting oracles succeed, which happens with probability at least $1 - \delta$. Similarly, we first define the random variable $\hat{Y}$ by the following process:

\begin{enumerate}

\item Call sampling oracle of $\mu$ to obtain one sample $\sigma\in \{\pm\}^V$ such that $\DTV{\sigma}{\mu}\leq \frac{\epsilon}{8}$, and we use $\sigma_S \in \{\pm\}^S$ as our sample.

\item Call conditional counting oracles to obtain $\hat{Z}_\mu^{\sigma_S}$ and $\hat{Z}_\nu^{\sigma_S}$ with an error bound $\epsilon/8$ and success probability $1-\delta$. Compute $\hat{\mu}_S(\sigma_S)=\hat{Z}^{\sigma_S}_\mu/\hat{Z_\mu}$ and $\hat{\nu}_S(\sigma_S)=\hat{Z}^{\sigma_S}_\nu/\hat{Z_\nu}$.

\item Define $\hat{Y}=\max\left (0,1-\hat{\nu}_S(\sigma_S)/\hat{\mu}_S(\sigma_S) \right )$, and in particular, $\hat{Y}=0$ if $\hat{\mu}_S(\sigma)=0$.
\end{enumerate}
Let $T=\frac{64}{\epsilon^2}$, we present our algorithm by following process:

\begin{itemize}
\item Draw T samples $\hat{Y}_1,\hat{Y}_2,\dots,\hat{Y}_n$ of random variable $\hat{Y}$.

\item Output the average $\hat{d}=\frac{1}{T}\sum_{i=1}^{T}\hat{Y}_i$.
\end{itemize}

The running time of our algorithm is $(\TC_{G}(\frac{\epsilon}{8})+\TS_{G}(\frac{\epsilon}{8}))\cdot T \cdot O(\log \frac{1}{\delta})$.

We now analyze the approximation error. First note that the probability that all conditional counting oracles success is $(1-\frac{\epsilon^2}{320})^{2T+2}>0.98$ and $\DTV{\sigma}{\mu}\leq \frac{\epsilon}{8}$ implies $\DTV{\sigma_S}{\mu_S}\leq \frac{\epsilon}{8}$. 
Compared to the proof of \Cref{thm:Approximate-Gibbs}, the difference is that we can only compute $\hat{Z}_\mu^{\sigma_S}$ and $\hat{Z}_\nu^{\sigma_S}$ approximately. However, the error can still be bounded. 
We prove that: for $\sigma_S$ with $\mu_S(\sigma_S)>0$, $\left(1-\frac{5\epsilon}{8} \right)\frac{\nu_S(\sigma_S)}{\mu_S(\sigma_S)}\leq \frac{\hat{\nu}_S(\sigma_S)}{\hat{\mu}_S(\sigma_S)}\leq \left ( 1+\frac{5\epsilon}{8} \right ) \frac{\nu_S(\sigma_S)}{\mu_S(\sigma_S)}$. Since $\frac{\hat{\nu}_S(\sigma_S)}{\hat{\mu}_S(\sigma_S)}= \frac{\hat{Z}_\nu^{\sigma_S}}{\hat{Z}_\mu^{\sigma_S}}\frac{\hat{Z}_{\mu}}{\hat{Z}_{\nu}}$ and $\frac{\nu_S(\sigma_S)}{\mu_S(\sigma_S)}= \frac{Z_\nu^{\sigma_S}}{Z_\mu^{\sigma_S}}\frac{Z_{\mu}}{Z_{\nu}}$,
\begin{align*}
    \frac{\hat{\nu}_S(\sigma_S)}{\hat{\mu}_S(\sigma_S)}= \frac{\hat{Z}_\nu^{\sigma_S}}{\hat{Z}_\mu^{\sigma_S}}\frac{\hat{Z}_{\mu}}{\hat{Z}_{\nu}}
    \leq \frac{Z_\nu^{\sigma_S}(1+\frac{\epsilon}{8})}{Z_\mu^{\sigma_S}(1-\frac{\epsilon}{8})}\frac{Z_{\mu}(1+\frac{\epsilon}{8})}{Z_{\nu}(1-\frac{\epsilon}{8})}< \left(1+\frac{5\epsilon}{8}\right)\frac{\nu_S(\sigma_S)}{\mu_S(\sigma_S)}.
\end{align*}
The other side of the inequality can be proved similarly. The rest of the proof follows from the proof of \Cref{thm:Approximate-Gibbs}. 
\end{proof}

\Cref{thm:many-vertex-alg} is a simple corollary of \Cref{thm:approx-margin-tv}, which is proved in \cref{sec:marginthm}.



\section{The algorithm for instances with small parameter distance}\label{sec:alg-main}
Let $\mu$ and $\nu$ be two general Gibbs distributions (including hardcore and Ising models) on the same graph $G=(V,E)$ with $n = |V|$. Let $w_\mu(\sigma)$ and $w_\nu(\sigma)$ be the weights of $\mu$ and $\nu$ on configuration $\sigma \in \{\pm\}^V$. In this section, we focus on the case where the parameter distance $\dis(\mu,\nu)$ is small. We will first give a basic algorithm for instances satisfying \Cref{cond:meta}, and then verify the \Cref{cond:meta} for Ising models with small parameter distance. 
We next give a more advanced algorithm for hardcore models with small parameter distance.

\subsection{Basic algorithm}
Define random variable
\begin{align}\label{eq:defW}
    W \defeq \frac{w_\nu(\sigma)}{w_\mu(\sigma)}, \quad\text{where } \sigma \sim \mu.
\end{align}
\begin{condition}\label{cond:meta}
Let $K,L\geq 1$ be two parameters.
Two Gibbs distributions $\mu$ and $\nu$ satisfy that
\begin{itemize}
    \item $\nu$ is absolutely continuous with respect to $\mu$: for all $\sigma \in \{\pm\}^V$, if $\mu(\sigma) = 0$, then $\nu(\sigma) = 0$; 
    \item $ \sqrt{\Var{W}} \leq K {\DTV{\mu}{\nu}}$;
    \item $\E[]{W} \geq \frac{1}{L}$. 
\end{itemize}
\end{condition}

\begin{theorem}\label{thm:alg-main}
    There exists an algorithm such that given two Gibbs distributions $\mu$ and $\nu$ on the same graph $G=(V,E)$, and any $0 < \epsilon <1$, if $\mu$ and $\nu$ satisfy \Cref{cond:meta} with $K$ and $L$, and both admit sampling and approximate counting oracles with cost functions $\TS_G(\cdot)$ and $\TC_G(\cdot)$, then it returns a random number $\hat{d}$ in time $O(\TC_G(\frac{\epsilon}{4}) + T \cdot \TS_G(\frac{1}{100T}))$, where $T = O(\frac{L^2K^2}{\epsilon^2})$, such that 
    \begin{align*}
        \Pr[]{ (1-\epsilon)\DTV{\mu}{\nu} \leq \hat{d} \leq (1+\epsilon)\DTV{\mu}{\nu} } \geq \frac{2}{3}.
    \end{align*}
    \end{theorem}

\begin{proof}
    Since $\nu$ is absolutely continuous with respect to $\mu$ ($\nu \ll \mu$), we can compute that:
\begin{align}\label{eq:DTV}
\DTV{\mu}{\nu} &= \frac{1}{2}\sum_{\sigma\in \{\pm\}^V} \left \vert \mu(\sigma)-\nu(\sigma) \right\vert =\frac{1}{2}\sum_{\sigma\in \{\pm\}^V:\mu(\sigma)>0}\mu(\sigma)\left \vert 1-\frac{\nu(\sigma)}{\mu(\sigma)}\right \vert \notag \\
&=\frac{1}{2}\sum_{\sigma\in \{\pm\}^V:\mu(\sigma)>0}\mu(\sigma) \left \vert 1-\frac{w_\nu(\sigma)}{w_\mu(\sigma)}\frac{Z_\mu}{Z_\nu}\right \vert \notag =\frac{Z_\mu}{2Z_\nu}\sum_{\sigma\in \{\pm\}^V:\mu(\sigma)>0}\mu(\sigma)\left \vert \frac{Z_\nu}{Z_\mu}-\frac{w_\nu(\sigma)}{w_\mu(\sigma)} \right \vert \notag \\
& = \frac{Z_\mu}{2Z_\nu}\E[]{|\E[]{W}-W|},
\end{align}
where in the last step we use the fact that if $\nu \ll \mu$, then $\E[]{W}= \sum_{\sigma \in \{\pm\}^V:\mu(\sigma)>0} \frac{w_\mu(\sigma)}{Z_\mu} \frac{w_\nu(\sigma)}{w_\mu(\sigma)} = \frac{Z_\nu}{Z_\mu}$.
We next use sampling oracles to define a random variable $\hat{W}\in [0,+\infty)$, which serves as an approximation of the random variable $W$ in~\eqref{eq:defW}.
Define
\begin{align*}
   T \defeq \left \lceil \frac{10^4 L^2 K^2}{\epsilon^2} \right \rceil.
\end{align*}
\begin{enumerate}
\item Call the sampling oracle of $\mu$ to obtain a random $\sigma \in \{\pm\}^V$ such that $\DTV{\sigma}{\mu}\leq \frac{1}{100T}$. 

\item Compute $\hat{W}=\frac{w_\nu(\sigma)}{w_\mu(\sigma)}$. In particular, if $w_\mu(\sigma)=0$\footnote{This case can happen because our sampling oracle is approximate.}, then we set $\hat{W} = 0$.
\end{enumerate}
Given the random variable $\hat{W}$, our algorithm is given by the following processes:
\begin{tcolorbox}[colback=lightgray!20, colframe=lightgray!18, coltitle=black, title={\textbf{Basic algorithm for instances satisfying \Cref{cond:meta}}}]
    \begin{itemize}
        \item Call approximate counting oracles to obtain $\hat{Z}_\mu$ and $\hat{Z}_\nu$ with error $\frac{\epsilon}{4}$.
        \item Draw $T$ samples $\hat{W}_1,\dots, \hat{W}_T$ from $\hat{W}$ independently.
        \item Compute $\bar{W}=\frac{1}{T}\sum_{i=1}^T \hat{W}_i$.
        \item Compute $\bar{E}=\frac{1}{T}\sum_{i=1}^{T} |\hat{W}_i-\bar{W}|$.
        \item Return $\hat{d}=\frac{\hat{Z}_\mu}{2\hat{Z}_\nu}\bar{E}$.
        \end{itemize}
    \end{tcolorbox}
The total running time of the above algorithm is 
\begin{align*}
    2\TC_G\tp{\frac{\epsilon}{4}} + O\tp{T \cdot \TS_G\left(\frac{1}{100T}\right)}.
\end{align*}
We remark that if both sampling and approximate counting can be solved in polynomial time, i.e., for any $\delta \in (0,1)$, $\TC_G(\delta),\TS_G(\delta) = \mathrm{poly}(\frac{n}{\delta})$, where $n$ is the number of vertices in $G$, then the above running time is $\mathrm{poly}(\frac{nLK}{\epsilon})$.

We now prove the correctness of the algorithm.
First, due to the definition of approximate counting oracle, with probability at least $0.98$, we can bound the error from ${\hat{Z}_\mu}{\hat{Z}_\nu}$ as follows:
\begin{align}\label{eq:ErrorZ}
    \left(1-\frac{3\epsilon}{4}\right)\frac{Z_\mu}{Z_\nu}  \leq \frac{(1-\epsilon/4)Z_\mu}{(1+\epsilon/4)Z_\nu} \leq \frac{\hat{Z}_\mu}{\hat{Z}_\nu}\leq \frac{(1+\epsilon/4)Z_\mu}{(1-\epsilon/4)Z_\nu}<\left(1+\frac{3\epsilon}{4}\right)\frac{Z_\mu}{Z_\nu}.
\end{align}
Suppose we can access a perfect sampler of $\mu$ and we draw perfect samples $W_1,\dots,W_T$ of $W$.
For each pair of $W_i$ and $\hat{W}_i$, there exists a coupling of $W_i,\hat{W}_i$ such that $\Pr[]{W_i \neq \hat{W}_i}\leq\frac{1}{100T}$. Then with probability at least $0.99$, $W_i = \hat{W}_i$ for all $1 \leq i \leq T$.
Consider an ideal algorithm that can use the perfect samples $W_1,\dots,W_T$.
Our real algorithm can be coupled successfully with the ideal algorithm with probability at least $0.99$. If we can show the ideal algorithm outputs correct result with probability at least $0.96$, then our real algorithm outputs correct result with probability at least $0.95 > 2/3$.

Now we assume all $W_i$ are perfect samples of $W$. We compute $\bar{W}=\frac{1}{T}\sum_{i=1}^{T}W_i$ and similarly $\bar{E}$ and $\hat{d}$. We only need to prove that $(1-\epsilon)d\leq \hat{d}\leq (1+\epsilon)d$ with probability at least 0.9, where $d=\DTV{\mu}{\nu}$. For every random variable $W_i$, by triangle inequality, we have
\begin{align*}
    \vert \E[]{W} - W_i \vert - \vert \bar{W} - \E[]{W} \vert \leq \vert \bar{W} - W_i \vert \leq \vert \E[]{W} - W_i \vert + \vert \bar{W} - \E[]{W} \vert.
\end{align*}
Note that $\bar{E} = \frac{1}{T}\sum_{i=1}^{T}|\bar{W}-W_i|$. We have
\begin{align}\label{eq:barE}
    \tp{\frac{1}{T}\sum_{i=1}^T\vert \E[]{W} - W_i \vert} - \vert \bar{W} - \E[]{W} \vert \leq \bar{E} \leq  \tp{\frac{1}{T}\sum_{i=1}^T\vert \E[]{W} - W_i \vert} + \vert \bar{W} - \E[]{W} \vert.
\end{align}
By definition of $\bar{W}$, we have $\E[]{\bar{W}} = \E[]{W}$ and $\Var[]{\bar{W}} = \frac{\Var[]{W}}{T}$. By Chebyshev's inequality,
\begin{align}\label{eq:barW}
    \Pr[]{|\bar{W}-\E[]{W}|\geq \frac{\epsilon d}{10L}}\leq \frac{100L^2\Var[]{\bar{W}}}{\epsilon^2 d^2}  = \frac{100L^2\Var[]{W}}{T\epsilon^2 d^2} \leq \frac{100L^2K^2d^2}{T \epsilon^2 d^2} \leq 0.01,
\end{align}
where the second inequality follows from \Cref{cond:meta}.
Next, consider the random variable
\begin{align*}
    R \defeq \vert \E[]{W} - W \vert.
\end{align*}
By the definition of $R$ and the variance bound in \Cref{cond:meta}, we know that 
\begin{align*}
    \Var[]{R} \leq \E[]{R^2} = \E[]{(\E[]{W}-W)^2} = \Var[]{W} \leq K^2 d^2.
\end{align*}
Note that $\frac{1}{T}\sum_{i=1}^T\vert \E[]{W} - W_i \vert$ is the average of $T$ i.i.d. random samples of $R$.
Denote $\bar{R} \defeq \frac{1}{T}\sum_{i=1}^T\vert \E[]{W} - W_i \vert$. Note that $\Var[]{\bar{R}} = \frac{\Var[]{R}}{T}$. By Chebyshev's inequality, we have
\begin{align}\label{eq:barR}
     \Pr[]{\left\vert \bar{R} - \E[]{R} \right\vert \geq \frac{\epsilon d}{10L}} \leq \frac{100L^2\Var[]{R}}{T \epsilon^2 d^2} \leq \frac{100L^2 K^2 d^2}{T \epsilon^2 d^2} \leq 0.01.
\end{align}
Combining~\eqref{eq:barE}, \eqref{eq:barW},~\eqref{eq:barR}, and a union bound, we have
\begin{align}\label{eq:barE-R}
    \Pr[]{\vert \bar{E} - \E[]{R} \vert \leq \frac{\epsilon d}{5L}} \geq 0.98.
\end{align}

Assume two good events in~\eqref{eq:barE-R} and~\eqref{eq:ErrorZ} both hold, which happens with probability at least $0.96$. The final output $\hat{d} = \frac{\hat{Z}_\mu}{2\hat{Z}_\nu}\bar{E}$ satisfies 
\begin{align*}
    \hat{d} &= \frac{\hat{Z}_\mu}{2\hat{Z}_\nu}\bar{E} \leq \frac{(1+3\epsilon/4)Z_\mu}{2Z_\nu} \cdot \tp{ \E[]{R} + \frac{\epsilon d}{5L} }\\
    &\leq \left(1+\frac{3\epsilon}{4} \right)\frac{Z_\mu}{2Z_\nu} \E[]{R} + \frac{Z_\mu}{Z_\nu L} \cdot \frac{\epsilon}{10} \cdot \left(1+\frac{3\epsilon}{4} \right) \cdot d.
\end{align*}
By~\eqref{eq:DTV}, we have $\frac{Z_\mu}{2Z_\nu} \E[]{R} = d$. By \Cref{cond:meta}, we have $\frac{Z_\mu}{Z_\nu L} = \frac{1}{L \E[]{W}}  \leq 1$. Therefore,
\begin{align*}
    \hat{d} \leq \left(1+\frac{3\epsilon}{4}\right)d + \frac{\epsilon}{10} \cdot \left(1+\frac{3\epsilon}{4} \right) \cdot d < (1+\epsilon)d.
\end{align*}
A similar argument gives a lower bound $\hat{d} \geq (1-\epsilon)d$.
\end{proof}





\subsection{Advanced algorithm for hardcore model}\label{sec:var-main}

We give the following algorithm for hardcore models with small parameter distance.
\begin{theorem}\label{thm:hardcore-adv}
Let $\theta = 10^{-10}\frac{\epsilon^{1/4}}{n^{5/2}}$.
There exists an algorithm such that given two hardcore models $\mu$ and $\nu$ on the same graph $G=(V,E)$, and any $0 < \epsilon <1$, if $\mu$ and $\nu$ both satisfy uniqueness condition in~\eqref{eq:cond-hardcore} and $\dis(\mu,\nu) < \theta$, then it returns a random number $\hat{d}$ in time $\tilde{O}\left(\frac{n^7}{\epsilon^2}+\frac{n^{6.5}}{\epsilon^{9/4}}\right)$ such that $(1 - \epsilon)\DTV{\mu}{\nu} \leq \hat{d} \leq (1+\epsilon)\DTV{\mu}{\nu}$ with probability at least $2/3$.
\end{theorem}

Let $\lambda^\mu$ and $\lambda^\nu$ be external fields of two hardcore models $\mu$ and $\nu$, respectively. For simplicity of the notation, we denote 
\begin{align*}
    D &= \dis(\mu,\nu) = \Vert \lambda^\mu - \lambda^\nu \Vert_\infty < \theta = 10^{-10}\frac{\epsilon^{1/4}}{n^{5/2}},\\
    d &= \DTV{\mu}{\nu}.
\end{align*}  
As discussed in the algorithm overview, we divide the vertices of $G$ into two parts: the "big" vertices and the "small" vertices. Define the threshold parameter
\begin{align*}
    \kappa \defeq 10^{-9}\frac{\epsilon^{1/4}}{n^{3/2}} = \Theta\left(\frac{\epsilon^{1/4}}{n^{3/2}}\right).
\end{align*}
Define two sets of vertices $B$ and $S$ in graph $G$:
\begin{align*}
    &B = \left\{v \in V \mid \min\{\lambda_v^\mu,\lambda_v^\nu\} \geq \kappa\right\},\\
    &S = V \setminus B = \left\{v \in V \mid \min\{\lambda_v^\mu,\lambda_v^\nu\} < \kappa\right\}.
\end{align*}
Recall $\mu_B$ and $\nu_B$ are the marginal distributions of $\mu$ and $\nu$ on $B$, respectively. Let $\Omega_B \subseteq \{\pm\}^B$ be the support of both $\mu_B$ and $\nu_B$. By the definition of $B$, for any $x\in \Omega_B$, all vertices $v \in B$ with $x_v = +1$ forms an independent set in $G$. For any $x\in \Omega_B$, let $\mu^x_S$ and $\nu^x_S$ be the marginal distributions of $\mu$ and $\nu$ on $S$ conditioned on $x$. The TV-distance between $\mu$ and $\nu$ can be represented by 
\begin{align}\label{eq:d-hardcore}
   d &= \frac{1}{2}\sum_{\sigma \in \{\pm\}^V} \left|\mu(\sigma)-{\nu(\sigma)}\right| = \frac{1}{2}\sum_{x \in \Omega_B}\sum_{y \in \{\pm\}^S} \left|\mu_B(x)\mu_S^x(y)-\nu_B(x)\nu_S^x(y)\right|\notag\\
    &= \frac{1}{2}\sum_{x \in \Omega_B}\mu_B(x) {\sum_{y \in \{\pm\}^S} \left|\frac{\nu_B(x)}{\mu_B(x)}\nu_S^x(y)-\mu_S^x(y)\right|}.
\end{align}
Define the function $f: \Omega_B \to \mathbb{R}$ as
\begin{align}\label{eq:def-f}
    f(x) \defeq \frac{1}{2}\sum_{y \in \{\pm\}^S} \left|\frac{\nu_B(x)}{\mu_B(x)}\nu_S^x(y)-\mu_S^x(y)\right|.
\end{align}
The calculation shows that $\DTV{\mu}{\nu} = \E[x \sim \mu_B]{f(x)}$. In a high-level view, our algorithm wants to draw i.i.d. samples $x\sim \mu_B$ and compute values $f(x)$ and then output the average value. Formally, we have the following lemmas.
Let $n$ denote the number of vertices in $G$.

\begin{lemma}\label{lem:hardcore-adv-var}
The variance $\Var[x \sim \mu_B]{f(x)} = O_\eta(d^2) \cdot (n^3 + n/\kappa)$, where $O_\eta$ holds a constant depending only on the gap $\eta$ in the uniqueness condition in~\eqref{eq:cond-hardcore}. 
\end{lemma}

\begin{lemma}\label{lem:hardcore-adv-1}
There exists a randomized data structure satisfies that 
\begin{itemize}
    \item the data structure can be constructed in time $\tilde{O}(\frac{n^7}{\epsilon^2}+\frac{n^{6.5}}{\epsilon^{9/4}})$ and the construction succeeds with probability at least 0.99;
    \item if the data structure is constructed successfully, then given any $x \in \Omega_B$, it deterministically answers an $\hat{f}(x) \geq 0$ in time $O(n^4)$ such that  
    \[\vert \hat{f}(x)-f(x)\vert \leq \frac{\epsilon}{50} \cdot d.\]
\end{itemize}
\end{lemma}

\Cref{lem:hardcore-adv-var} can be used to control the variance of $f(x)$. \Cref{lem:hardcore-adv-1} is the main technical part of our algorithm.
We first assume both \Cref{lem:hardcore-adv-var} and \Cref{lem:hardcore-adv-1} hold and prove \Cref{thm:hardcore-adv}.
The proofs of \Cref{lem:hardcore-adv-var} and \Cref{lem:hardcore-adv-1} are given in Section \ref{sec:hardcore-adv-var} and \ref{sec:hardcore-adv-1}, respectively.
\ifthenelse{\boolean{conf}}{\begin{proof}\textbf{of \Cref{thm:hardcore-adv}}}{\begin{proof}[Proof of \Cref{thm:hardcore-adv}]}
Let $T= O(\frac{n^3 + n/\kappa}{\epsilon^2}) = O(\frac{n^3}{\epsilon^2}+\frac{n^{5/2}}{\epsilon^{9/4}})$ be large enough. 
    \begin{tcolorbox}[colback=lightgray!20, colframe=lightgray!18, coltitle=black, title={\textbf{The algorithm for hardcore model}}]
        \begin{itemize}
            \item Construct the data structure in \Cref{lem:hardcore-adv-1};
            \item Draw $T$ independent approximate samples $x_1,x_2,\ldots,x_T$ from the marginal distribution $\mu_B$ with $\DTV{\mu_B}{x_i} \leq \frac{1}{100T}$.
            \item Use the data structure to compute $\hat{f}(x_1),\hat{f}(x_2),\ldots,\hat{f}(x_T)$.
            \item Return $\hat{d}=\frac{1}{T}\sum_{i=1}^T \hat{f}(x_i)$.
            \end{itemize}
        \end{tcolorbox}

Consider an ideal algorithm $\+A^*$ that draw perfect samples $x_1,\ldots,x_T$ and exactly compute the values of $f(x_1),\ldots,f(x_T)$. Let $d^*$ denote the output $\frac{1}{T}\sum_{i=1}^Tf(x_i)$. We have $\E[]{d^*} = d$ and the variance of $d^*$ is $\frac{\Var[\mu_B]{f}}{T}$. By our choice of $T$ and the variance bound in \Cref{lem:hardcore-adv-var}, using Chebyshev's inequality, we have
\begin{align*}
    \Pr[]{|d^*-d| \geq \frac{\epsilon}{10}d} \leq 0.01.
\end{align*}

Consider our real algorithm $\+A$.
All the approximate samples $x_1,\ldots,x_T$ in $\+A$ can be coupled successfully with perfect samples with probability at least $1 - T \cdot \frac{1}{100T} = 0.99$.
Also that the data structure in $\+A$ is constructed successfully, which happen with  probability at least $0.99$.
By \Cref{lem:hardcore-adv-1}, there exists a coupling of $\+A$ and $\+A^*$ such that with probability at least 0.98, $|\hat{d}-d^*| \leq \frac{\epsilon}{50}d$. Hence, using a union bound, we have
\begin{align*}
    \Pr[]{(1-\epsilon)d \leq \hat{d} \leq (1+\epsilon)d } \geq 0.97 > \frac{2}{3}.
\end{align*}

The total running time is bounded by
\begin{align*}
 \tilde{O}\left(\frac{n^7}{\epsilon^2}+\frac{n^{6.5}}{\epsilon^{9/4}}\right) + O\left(n^4 \cdot T \right)   =  \tilde{O}\left(\frac{n^7}{\epsilon^2}+\frac{n^{6.5}}{\epsilon^{9/4}}\right). &\ifthenelse{\boolean{conf}}{}{\qedhere}
\end{align*}
\end{proof}

\subsubsection{Analyze the variance of the estimator (Proof of \texorpdfstring{\Cref{lem:hardcore-adv-var})}{}}
\label{sec:hardcore-adv-var}
Before we prove the lemma, we first remark that one can show for any $x \in \Omega_B$, $|f(x) - 1| \leq O(\frac{n}{\kappa}d)$, which implies the $O(\frac{n^2}{\kappa^2} d^2) $ variance bound. 
This bound gives a polynomial-time algorithm but the degree of polynomial is higher. 
If we use this bound, then in the proof of \Cref{thm:hardcore-adv}, it requires $O(\frac{n^2}{\kappa^2 \epsilon^2})$ samples $x \sim \mu_B$ and \Cref{lem:hardcore-adv-1} computes each value of $\hat{f}(x)$ in a super-linear time.
Alternatively, we give a more technical analysis to achieve a better dependency on $\kappa$, which gives a better running time of our algorithm.

We bound the variance of $f(x)$ by bounding the second moment of $\E[\mu_B]{f^2} = \E[x \sim \mu_B]{f^2(x)}$. We first need the following upper bound on the value of $f(x)$:
\begin{align*}
        f(x) &= \frac{1}{2}\sum_{y \in \{\pm\}^S} \left|\frac{\nu_B(x)}{\mu_B(x)}\nu_S^x(y)-\mu_S^x(y)\right|\notag\\
\text{(by triangle inequality)}\quad       &\leq \frac{1}{2}\left\vert \frac{\nu_B(x)}{\mu_B(x)} -1 \right\vert \sum_{y \in \{\pm\}^S}\nu_S^x(y) + \frac{1}{2}\sum_{y \in \{\pm\}^S}\left|\nu_S^x(y)-\mu_S^x(y)\right|\notag\\
        &= \frac{1}{2}\left\vert \frac{\nu_B(x)}{\mu_B(x)} -1 \right\vert  + \DTV{\nu_S^x}{\mu_S^x}.
\end{align*}
We have the following upper bound on the $\DTV{\nu_S^x}{\mu_S^x}$.
\begin{lemma}\label{claim:tv-bound}
    for any $x \in \Omega_B$, it holds that $\DTV{\nu_S^x}{\mu_S^x} \leq 4 nD$.
\end{lemma}
The proof of \Cref{claim:tv-bound} will be given later.
Assume \Cref{claim:tv-bound} holds.
Since $f(x) \geq 0$, we can upper bound
\begin{align}\label{eq:varf-bd}
\E[\mu_B]{f^2} &\leq \frac{1}{4}\E[x \sim \mu_B]{\left(\frac{\nu_B(x)}{\mu_B(x)} -1 \right)^2 } + 4nD \E[x \sim \mu_B]{\left\vert \frac{\nu_B(x)}{\mu_B(x)} -1 \right\vert} + 16n^2D^2\notag\\
&=  \frac{1}{4}\E[x \sim \mu_B]{\left(\frac{\nu_B(x)}{\mu_B(x)} -1 \right)^2 } + 8 nD\cdot \DTV{\nu_B}{\mu_B} + 16n^2D^2\notag\\
&\leq  \frac{1}{4}\E[x \sim \mu_B]{\left(\frac{\nu_B(x)}{\mu_B(x)} -1 \right)^2 } + 4\cdot 10^8n^2d^2,
\end{align}
where the last inequality follows from the fact $d = \DTV{\mu}{\nu} > \DTV{\mu_B}{\nu_B}$ and $d \geq \frac{1}{5000}D$.

Define a function $h(x) = \frac{\nu_B(x)}{\mu_B(x)}$. We have $\E[\mu_B]h =1 $. Our task is reduced to bound the variance $\Var[\mu_B]{h} = \Var[x \sim \mu_B]{h(x)}$. We will use the following Poincar\'e inequality (i.e. the spectral gap of the Glauber dynamics) for the marginal distribution $\mu_B$.
For any subset $\Lambda \subseteq V$, recall that $\Omega_\Lambda$ denotes the support of $\mu_\Lambda$. 
\ifthenelse{\boolean{conf}}{\begin{lemma}[\text{\citet{ChenFYZ21}}]}{\begin{lemma}[\text{\cite{ChenFYZ21}}]}\label{lem:poincare}
Since the hardcore model $(G,\lambda^\mu)$ is in the uniqueness regime in~\eqref{eq:cond-hardcore} with constant gap $\eta>0$. For any function $g: \Omega_B \to \mathbb{R}$, it holds that
\begin{align*}
    \Var[\mu_B]{g} \leq C_\eta \sum_{v \in B} \sum_{\sigma \in \Omega_{B - v}} \mu_{B - v}(\sigma) \Var[\mu^\sigma_B]{g},
\end{align*}
where $B-v$ denote the set $B \setminus \{v\}$ and $C_\eta$ is a constant depending on $\eta$.
\end{lemma}
The Poincar\'e inequality in~\cite{ChenFYZ21} is stated from the entire Gibbs distribution $\mu$. One can lift it to marginal distribution $\mu_B$. For the completeness, we give a proof in \Cref{app:poin}.

Fix a vertex $v \in B$ and a $\sigma \in \Omega_{B - v}$. 
Let $\sigma^{v_+}$ denote a configuration in $\{\pm\}^B$ obtained by extending $\sigma$ further by setting $v$ to $+1$. Define $\sigma^{v_-}$ similarly.
By the definition of variance and the definition of the function $h=\frac{\nu_B}{\mu_B}$, we can write
\begin{align*}
    \Var[\mu^\sigma_B]{h} &= \mu^\sigma_{v}(+1)\mu^\sigma_{v}(-1)(h(\sigma^{v_+}) -h(\sigma^{v_-}))^2\\
    &= \tp{ \frac{\nu_{B-v}(\sigma)}{ \mu_{B-v}(\sigma) }}^2 \mu^\sigma_{v}(+1)\mu^\sigma_{v}(-1)\left( \frac{\nu_{v}^{\sigma}(+1)}{\mu_{v}^{\sigma}(+1)} -  \frac{\nu_{v}^{\sigma}(-1)}{\mu_{v}^{\sigma}(-1)} \right)^2. 
\end{align*}
Note that $\Var[\mu^\sigma_B]{h} = 0$ if either $\mu^\sigma_v(+1) = 0$.
We assume that $\mu^\sigma_v(+1) > 0$, then for every neighbor $u$ of vertex $v$, it must hold that if $u \in B$, then $\sigma_u = -1$.
We claim the following bounds.
\begin{claim}\label{claim:p-bound}
The following bounds hold
\begin{itemize}
    \item $\lambda^\mu_v/10 \leq \mu_v^\sigma(+1) \leq \lambda_v^\mu$;
    \item $|\mu^\sigma_v(+1)-\nu^\sigma_v(+1)| \leq 80nD \mu_v^\sigma(+1) + 4D$;
    \item $ \frac{\nu_{B-v}(\sigma)}{ \mu_{B-v}(\sigma) } \leq 2$.
\end{itemize}
\end{claim}
Assume the above claim holds, which will be proved later. Let $p = \mu_v^\sigma(+1)$. By our choices of parameters, since $D \leq \theta$ is sufficient small, $\mu^\sigma_v(+1) \geq \frac{\kappa}{10} \geq 100n D \mu_v^\sigma(+1) + 4D$. We have
\begin{align*}
 \left( \frac{\nu_{v}^{\sigma}(+1)}{\mu_{v}^{\sigma}(+1)} -  \frac{\nu_{v}^{\sigma}(-1)}{\mu_{v}^{\sigma}(-1)} \right)^2 \leq \left (\frac{4D + 80nD \mu^\sigma_v(+1)}{p(1-p)}\right )^2.    
\end{align*}
Using \Cref{claim:p-bound},
the variance can be bounded by
\begin{align*}
    \Var[\mu^\sigma_B]{h} &\leq 4 p(1-p) \frac{(80nDp +4D)^2}{p^2(1-p)^2} = O(D^2) \cdot \frac{(20np+1)^2}{p(1-p)} \leq O(D^2)\cdot(n^2p + n + 1/p)\\
    &\leq  O(D^2) \cdot \left(n^2 + \frac{1}{\kappa}\right).
\end{align*}
Using \Cref{lem:poincare} and the above upper bound, we have
\begin{align}\label{eq:varh}
    \Var[\mu_B]{h} &\leq C_\eta \sum_{v \in B} \sum_{\sigma \in \Omega_{B - v}} \mu_{B - v}(\sigma) \Var[\mu^\sigma_B]{h}\notag\\ &\leq O_\eta(D^2) \sum_{v \in B}\left(n^2 + \frac{1}{\kappa}\right)\notag\\
 (\text{by \Cref{lem:TV-lower} })\quad    &= O_\eta(d^2) \cdot \left(n^3 + \frac{n}{\kappa}\right).
\end{align}    
Hence, using \eqref{eq:varf-bd}, the variance of $f(x)$ is at most
\begin{align*}
\Var[\mu_B]{f} \leq \E[\mu_B]{f^2} \leq \Var[\mu_B]{h} +  O_\eta(n^2 d^2) = O_\eta(d^2) \cdot \left(n^3 + \frac{n}{\kappa}\right).
\end{align*}

\paragraph{Proofs of technical lemmas and claims}
We first give two general lemmas (\Cref{lem:Zcond-bound} and \Cref{lem:marginal-ratio}) that will be used in later proofs. We then prove all technical lemmas and claims appeared in the proof of \Cref{lem:hardcore-adv-var}.

Define the conditional partition function.
Fix a configuration $x \in \{\pm\}^B$. Define $Z_{S,\mu}^x$ as the conditional partition function defined by
\begin{align*}
    Z_{S,\mu}^x \defeq \sum_{y \in \{\pm\}^S: \mu(x+y) > 0} \prod_{v \in S: y_v = + 1}\lambda_v^\mu.
\end{align*}
Intuitively, $Z_{S,\mu}^x$ is the total weights of $y \in \{\pm\}^S$ such that $x+y$ is a valid configuration (forms an independent set in $G$), where $x+y \in \{\pm\}^V$ is the concatenation of $x$ and $y$.
Alternatively, $Z_{S,\mu}^x$ can be interpreted as follows. Let $N_G(v)$ denote the set of neighbors of $v$ in $G$. Given $x$, one can remove all vertices $v \in S$ from $S$ such that there exists $u \in N_G(v) \cap B$ with $x_u =+1$. Let $S^x \subseteq S$ denote the set of remaining vertices:
\begin{align}\label{eq:sx}
   S^x = S \setminus \{v \in S \mid \exists u \in N_G(v)\cap B \text{ s.t. } x_u = +1 \}. 
\end{align}
Then $Z_{S,\mu}^x$ is the partition function for the hardcore model in induced subgraph $G[S^x]$. 
The following property of the conditional partition function will be used in our proofs.
\begin{lemma}\label{lem:Zcond-bound}
Suppose $\kappa + \theta < 1/(10n)$.
For any $x \in \{\pm\}^B$, it holds that 
\begin{itemize}
    \item $1 \leq Z_{S,\mu}^x,Z_{S,\nu}^x < 2 $;
    \item $|Z_{S,\mu}^x - Z_{S,\nu}^x| \leq 2 n D$.
\end{itemize}
\end{lemma}
\begin{proof}
    Since the empty set contributes the weight $1$ to both $Z_{S,\mu}^x$, so $Z_{S,\mu}^x \geq 1$. For the upper bound, Let $\lambda_{\max} = \max_{v \in S} \max \{\lambda_v^\mu,\lambda_v^\nu\}$. By the definition of $S$, we have $\lambda_{\max} \leq \kappa + D < \kappa + \theta < \frac{1}{10n}$. Hence $Z_{S,\mu}^x \leq (1+1/(10n))^n < 2$. The same bound holds for $Z_{S,\nu}^x$.

    For any independent set $I$ in graph $G[S^x]$, the difference of the weights is $\left\vert \prod_{v\in I}\lambda_v^\nu-\prod_{v\in I} \lambda_v^\mu\right\vert$. We first show the difference of the weights of $I$ is at most $(\lambda_{\max}+D)^{|I|} - \lambda_{\max}^{|I|}$. Assume $\lambda_v^\nu = \lambda_v^\mu + \delta_v$, where $|\delta_v| \leq D$. Then 
    \begin{align}\label{eq:max-diff}
        \left\vert \prod_{v\in I}\lambda_v^\nu-\prod_{v\in I} \lambda_v^\mu\right\vert &= \left\vert \prod_{v\in I}(\lambda_v^\mu + \delta_v)-\prod_{v\in I} \lambda_v^\mu\right\vert = \left\vert \sum_{A \subseteq I: A \neq \emptyset} \prod_{v \in A}\delta_v \prod_{u \in I \setminus A} \lambda_u^\mu\right\vert\notag\\
        &\leq \sum_{A \subseteq I: A \neq \emptyset} \prod_{v \in A}D \prod_{u \in I \setminus A} \lambda_{\max}\notag\\
        &= (\lambda_{\max}+D)^{|I|} - \lambda_{\max}^{|I|}. 
    \end{align}
    The number of independent sets of size $k$ in $G[S^x]$ is at most $\binom{n}{k}$, where $n$ is the number of vertices in $G$. We have
    \begin{align}\label{eq:max-diff-sum}
      \vert Z_{S,\nu}^x-Z_{S,\mu}^x \vert &\leq \sum_{k=0}^n \binom{n}{k} (\lambda_{\max}+D)^k - \sum_{k=0}^n \binom{n}{k}\lambda_{\max}^k \notag\\
         &= \left(1+\lambda_{\max}+D\right)^n - \left(1+\lambda_{\max}\right)^n\notag\\
         &= (1+\lambda_{\max})^n\left( \left(1+\frac{D}{1+\lambda_{\max}}\right)^n -1 \right)\notag\\
    \text{by ($\lambda_{\max},D <1/(10n)$)}\quad &\leq 2nD.
    \end{align}
This proves the lemma.
\end{proof}

The second general lemma we will use is the following bound on the marginal ratio.
\begin{lemma}\label{lem:marginal-ratio}
    Suppose $\kappa + \theta < 1/(10n)$ and $\theta/\kappa < 1/(10n)$.
    For any $x \in \Omega_B$, it holds that
    \begin{align*}
        \left\vert\frac{\nu_B(x)}{\mu_B(x)} -1 \right\vert \leq \frac{10n D}{\kappa}.
    \end{align*}
\end{lemma}
\begin{proof}
    Define  $g(x)=\frac{\nu_B(x)}{\mu_B(x)}$. For $x \in \Omega_B \subseteq \{\pm\}^B$ and $y \in \{\pm\}^S$, we use $x + y$ to denote a full configuration in $\{\pm\}^V$ obtained by concatenating $x$ and $y$. We have
    \begin{align*}
        &\mu_B(x)=\sum_{y\in \{\pm\}^{S}}\mu(x+y)=\frac{\prod_{v\in B:x_v=1}\lambda_v^\mu \cdot Z_{S,\mu}^x}{Z_\mu},\text{ and}\\
        &\nu_B(x)=\sum_{y\in \{\pm\}^{S}}\nu(x+y)=\frac{\prod_{v\in B:x_v=1}\lambda_v^\nu \cdot Z_{S,\nu}^x}{Z_\nu}.
    \end{align*}
Then we have
\begin{align*}
    g(x)=\frac{\nu_B(x)}{\mu_B(x)}=\left(\prod_{v\in B:x_v=1} \frac{\lambda_v^\nu}{\lambda_v^\mu} \right) \cdot \frac{Z_{S,\nu}^x}{Z_{S,\mu}^x} \cdot \frac{Z_\mu}{Z_\nu}.
\end{align*}
Denote $\alpha=\prod_{v\in B:x_v=1} \frac{\lambda_v^\nu}{\lambda_v^\mu}$, $\beta=\frac{Z_{S,\nu}^x}{Z_{S,\mu}^x}$ and $\gamma =\frac{Z_\mu}{Z_\nu}$. We analyze each term one by one.

We have ${D}/{\kappa} \leq \theta / \kappa \leq 1/(10n)$ and  
\begin{align*}
    \alpha \leq \prod_{v\in B:x_v=1}\frac{\lambda_v^\mu+D}{\lambda_v^\mu}\leq \left(1+\frac{D}{\kappa}\right)^{|B|}\leq \left(1+\frac{D}{\kappa}\right)^{n}\leq 1+\frac{2nD}{\kappa}.
\end{align*}
A similar argument gives a lower bound $\alpha \geq 1-\frac{2nD}{\kappa}$. 

For the second term $\beta$, using \Cref{lem:Zcond-bound}, we have
\begin{align*}
    \vert \beta - 1 \vert = \left\vert \frac{Z_{S,\nu}^x}{Z_{S,\mu}^x} - 1 \right\vert = \frac{\vert Z_{S,\nu}^x-Z_{S,\mu}^x \vert}{Z_{S,\mu}^x} \leq \vert Z_{S,\nu}^x-Z_{S,\mu}^x \vert \leq 2nD.
\end{align*}

Finally, for $\gamma = \frac{Z_\mu}{Z_\nu}$, we have the following equivalent form
\begin{align*}
    \gamma = \frac{Z_\mu}{Z_\nu}=\frac{\sum_{x\in\Omega_B}(\prod_{v\in B:x_v=+1}\lambda_v^\mu \cdot Z_{S,\mu}^x)}{\sum_{x\in\Omega_B}(\prod_{v\in B:x_v=+1}\lambda_v^\nu \cdot Z_{S,\nu}^x)}.
\end{align*}
Using the bound for $\alpha$ and $\beta$, since $\theta/\kappa < 1/(10n)$ and $\theta+\kappa < 1/(10n)$, we have
\begin{align}\label{eq:upperZ/Z}
   1 -  \frac{4nD}{\kappa} < \gamma \leq \left( 1 + \frac{2nD}{\kappa} \right)(1+2nD) < 1 + \frac{4nD}{\kappa}.
\end{align}
Combining the bounds for $\alpha$, $\beta$ and $\gamma$, since $D/\kappa \leq \theta/\kappa < 1/(10n)$, we have
\begin{align*}
   1 -  \frac{10nD}{\kappa} < g(x) \leq \left( 1+\frac{2nD}{\kappa}\right)(1+2nD)\left( 1+\frac{4nD}{\kappa}\right) <  1 + \frac{10nD}{\kappa}. \ifthenelse{\boolean{conf}}{}{&\qedhere}
\end{align*}

\end{proof}

\ifthenelse{\boolean{conf}}{\begin{proof}\textbf{of \Cref{claim:tv-bound}}}{\begin{proof}[Proof of \Cref{claim:tv-bound}]}
    For any $y \in \{\pm\}^S$ such that $y+x$ forms an independent set, we can bound 
    \begin{align}\label{eq:bound-diffy}
        \left|\nu_S^x(y)-\mu_S^x(y)\right| &\leq \left\vert \frac{ \prod_{v \in S: y_v= +1}\lambda_v^\nu }{Z_{S,\nu}^x } - \frac{\prod_{v \in S: y_v= +1}\lambda_v^\mu }{Z_{S,\mu}^x } \right\vert\notag\\
    \text{(by $Z_{S,\mu}^x,Z_{S,\nu}^x \geq 1$)}\quad    &\leq \left\vert Z_{S,\mu}^x \prod_{v \in S:y_v=+1}\lambda_v^\nu - Z_{S,\nu}^x \prod_{v \in S:y_v=+1}\lambda_v^\mu \right\vert\notag\\
    \text{(triangle ineq.)}\quad  &\leq Z_{S,\mu}^x\left\vert \prod_{v \in S:y_v=+1}\lambda_v^\nu -  \prod_{v \in S:y_v=+1}\lambda_v^\mu\right\vert + \left\vert Z_{S,\mu}^x - Z_{S,\nu}^x \right\vert\prod_{v \in S:y_v=+1}\lambda_v^\mu .
    \end{align}
    We bound each term separately. 
    Recall that $\lambda_{\max} = \max_{v \in S} \max \{\lambda_v^\mu,\lambda_v^\nu\} < \kappa + \theta$.
    By \Cref{lem:Zcond-bound}, we have $Z^x_{S,\mu} < 2$. Using~\eqref{eq:max-diff}, we have
    \begin{align}\label{eq:bound-diffy-1}
        Z_{S,\mu}^x\left\vert \prod_{v \in S:y_v=+1}\lambda_v^\nu -  \prod_{v \in S:y_v=+1}\lambda_v^\mu\right\vert \leq 2 \left((\lambda_{\max}+D)^{\Vert y \Vert_+} - \lambda_{\max}^{\Vert y \Vert_+}\right),
    \end{align}
    where $\Vert y \Vert_+$ is the number of $+1$s in $y$. For the second term, using \Cref{lem:Zcond-bound}, we have
    \begin{align}\label{eq:bound-diffy-2}
        \prod_{v \in S:y_v=+1}\lambda_v^\mu\ \left\vert Z_{S,\mu}^x - Z_{S,\nu}^x \right\vert \leq 2nD \lambda_{\max}^{\Vert y \Vert_+}.
    \end{align}
    The number of independent sets of size $k$ is at most $\binom{n}{k}$. We have
    \begin{align*}
        \DTV{\nu_S^x}{\mu_S^x}&= \frac{1}{2}\sum_{y \in \{\pm\}^S}\left|\nu_S^x(y)-\mu_S^x(y)\right|\\
        &\leq \sum_{k=0}^n \binom{n}{k}  \left((\lambda_{\max}+D)^k - \lambda_{\max}^k\right) +   nD \sum_{k=0}^n \binom{n}{k} \lambda_{\max}^{k}\\
    \text{(by the same calculation as~\eqref{eq:max-diff-sum})}\quad &\leq 2nD +nD (1+\lambda_{\max})^n\\
    \text{(by $\lambda_{\max} < 1 / (10n)$)}\quad &< 4nD.
    \end{align*}
This proves the upper bound on TV distance.
\end{proof}

\ifthenelse{\boolean{conf}}{\begin{proof}\textbf{of \Cref{claim:p-bound}}}{\begin{proof}[Proof of \Cref{claim:p-bound}]}
Recall that $\sigma$ is a configuration in $\Omega_{B-v}$, and $\sigma^{v_+},\sigma^{v_-}$ are configurations in $\{\pm\}^{B}$.
For simplicity of the notation, we use $Z_\mu^+$ to denote $Z_{S,\mu}^{\sigma^{v_+}}$ and $Z_\mu^-$ to denote $Z_{S,\mu}^{\sigma^{v_-}}$.
Similarly, we define $Z_\nu^+,Z_\nu^-$. We have
\begin{align*}
    \vert \mu^\sigma_v(+1) - \nu^\sigma_v(+1) \vert &= \left\vert \frac{\lambda_v^\mu Z_\mu^+}{Z_\mu^- + \lambda_v^\mu Z_\mu^+} - \frac{\lambda_v^\nu Z_\nu^+}{Z_\nu^- + \lambda_v^\nu Z_\nu^+} \right\vert = \left\vert \frac{\lambda_v^\mu Z_\mu^+Z_\nu^- - \lambda_v^\nu Z_\nu^+Z_\mu^-}{(Z_\mu^- + \lambda_v^\mu Z_\mu^+)(Z_\nu^- + \lambda_v^\nu Z_\nu^+)} \right\vert\\
    &\leq \left\vert \lambda_v^\mu Z_\mu^+Z_\nu^- - \lambda_v^\nu Z_\nu^+Z_\mu^- \right\vert\\
    &\leq Z_\nu^+Z_\mu^-|\lambda_v^\mu- \lambda_v^\nu|  + \lambda^\mu_v \vert Z_\mu^+Z_\nu^- - Z_\nu^+Z_\mu^-\vert.
\end{align*}
Using \Cref{lem:Zcond-bound}, we have $Z_\nu^+Z_\mu^-|\lambda_v^\mu- \lambda_v^\nu|\leq 4D$. We also know that $\vert Z_\mu^+ - Z_\mu^+ \vert \leq 2nD$ and $\vert Z_\mu^- - Z_\nu^- \vert \leq 2nD$. Therefore, by using the triangle inequality, we have
\begin{align*}
    \vert Z_\mu^+Z_\nu^- - Z_\nu^+Z_\mu^-\vert \leq Z_\mu^+|Z_\nu^- - Z_\mu^-| + Z_\mu^-|Z_\mu^+ - Z_\nu^+| \leq 8nD.   
\end{align*}
Therefore, we have
\begin{align*}
    \vert \mu^\sigma_v(+1) - \nu^\sigma_v(+1) \vert \leq 4D + 8nD \lambda_v^\mu.
\end{align*}

Next, we bound the value of $\mu^\sigma_v(+1)$. To sample from the distribution $\mu^\sigma_v$, one can first sample all the neighbors $u \in N_G(v) \setminus B$ of $v$, then sample the value of $v$ further conditional on the configuration of $N_G(v)$. Suppose with probability $q$, all vertices $u \in N_G(v) \setminus B$ are sampled to be $-1$. Note that $q \geq (\frac{1}{1+\lambda_{\max}})^n \geq (\frac{1}{1+\kappa+D})^n$. Since $\kappa + \theta < 1/(4n)$ and $D < \theta$, we have $q \geq \frac{1}{2}$. Conditional on this event, $v$ takes value $+1$ with probability $\frac{\lambda_v^\mu}{1+\lambda_v^\mu} \geq \frac{\lambda_v^\mu}{5}$, where we use the fact that $\lambda_v^\mu \leq \lambda_c(\Delta) \leq 4$. We have the  lower bound $\mu^\sigma_v(+1) \geq \frac{1}{10} \lambda_v^\mu$. 
On the other hand, we have the upper bound $ \mu^\sigma_v(+1) \leq \frac{\lambda_v^\mu}{1+\lambda_v^\mu} \leq  \lambda_v^\mu$. Combining together, we have
\begin{align*}
    \frac{\lambda_v^\mu}{10} \leq \mu^\sigma_v(+1) \leq  \lambda_v^\mu. 
\end{align*}
The above bound implies
\begin{align*}
    \vert \mu^\sigma_v(+1) - \nu^\sigma_v(+1) \vert \leq 4D + 80 nD \mu^\sigma_v(+1).
\end{align*}

For the last bound, using \Cref{lem:marginal-ratio}, we have
\begin{align*}
    \frac{ \nu_{B-v}(\sigma) }{\mu_{B-v}(\sigma)} = \frac{\nu_B(\sigma^{v_+})+\nu_B(\sigma^{v_-})}{\mu_B(\sigma^{v_+})+\mu_B(\sigma^{v_-})} \leq 1 + \frac{10n D}{\kappa} < 2,
\end{align*}
where the last inequality follows from the fact that $D/\kappa < 1/(10n)$.
\end{proof}

\subsubsection{Approximate the value of the estimator (Proof of \texorpdfstring{\Cref{lem:hardcore-adv-1}}{Lg})}\label{sec:hardcore-adv-1}
Define $\Vert y\Vert_+$ be the number of $+1$s in $y$. 
Let $t \geq 0$ be an integer.
The specific choice of $t$ will be fixed later.
Define the truncated function $f_t$ of $f$ defined by 
\begin{align*}
    f_t(x) \defeq \frac{1}{2}\sum_{y \in \{\pm\}^S: \Vert y\Vert_+\leq t} \left|\frac{\nu_B(x)}{\mu_B(x)}\nu_S^x(y)-\mu_S^x(y)\right|.
\end{align*}

Compared to $f$ in~\eqref{eq:def-f}, $f_t$ only includes the size at most $t$ independent sets of the induced subgraph $G[S^x]$, where $S^x$ is defined in~\eqref{eq:sx}. We have the following relation between $f_t$ and $f$.

\begin{lemma}\label{lem:appf}
  Suppose $\kappa + \theta < 1/(10n)$ and $\theta/\kappa < 1 / (10n)$.
  For any integer $t \geq 0$, any $x\in \Omega_B$,  
  \begin{align}\label{eq:error-t}
  0\leq f(x)-f_t(x)\leq 10^6\left(1+\frac{n}{10}\right)^{t+1} \kappa^t n^{t+2} \cdot d \defeq \eta(\kappa,t)\cdot d.
  \end{align}
\end{lemma}

\begin{proof} 
First, by definition, $ f(x)-f_t(x)\geq 0$.
Recall that $g(x) = \frac{\nu_B(x)}{\mu_B(x)}$.
Let $\Omega_S^x \subseteq \{\pm\}^S$ be the set over all $y$ such that $x + y$ forms an independent set in $G$.
We have
\begin{align*}
f(x) - f_t(x) &= \frac{1}{2}\sum_{y \in \Omega_S^x: \Vert y\Vert_+ \geq t+1}\left|g(x)\nu_S^x(y) - \mu_S^x(y)\right|\\
\text{(by \Cref{lem:marginal-ratio})} \quad &\leq \frac{1}{2}\sum_{y \in \Omega_S^x: \Vert y\Vert_+ \geq t+1}\left|\nu_S^x(y) - \mu_S^x(y)\right| + \frac{5 nD}{\kappa} \cdot\sum_{y \in \Omega_S^x: \Vert y\Vert_+ \geq t+1}\nu_S^x(y)
\end{align*}
Recall that $\lambda_{\max} = \max_{v \in S} \max \{\lambda_v^\mu,\lambda_v^\nu\}$. Using the result in~\eqref{eq:bound-diffy},~\eqref{eq:bound-diffy-1} and~\eqref{eq:bound-diffy-2}, we have
\begin{align*}
    \sum_{y \in \Omega_S^x: \Vert y\Vert_+ \geq t+1}\left|\nu_S^x(y) - \mu_S^x(y)\right| \leq \sum_{k=t+1}^n \binom{n}{k} \left((\lambda_{\max}+D)^k - \lambda_{\max}^k\right) +   nD \sum_{k=t+1}^n \binom{n}{k} \lambda_{\max}^{k}.
\end{align*}
Let $\phi = \kappa + D$. Then $\lambda_{\max} \leq \phi$ and $(\lambda_{\max}+D)^k - \lambda_{\max}^k \leq (\phi+D)^k - \phi^k =\phi^k( (1+D/\phi)^k - 1 )$. Note  that $D/\phi < \theta/\kappa < 1/(4n)$. The last term is at most $\phi^k\cdot (2kD)/\phi \leq \phi^{k-1}\cdot (2nD)$. Then 
\begin{align*}
    f(x) - f_t(x)  &\leq \frac{nD}{\kappa}\sum_{k=t+1}^n \binom{n}{k} \phi^{k} +   \frac{nD}{2} \sum_{k=t+1}^n \binom{n}{k} \phi^k + \frac{5nD}{\kappa}\sum_{k=t+1}^n \binom{n}{k} \phi^k \leq \frac{8nD}{\kappa} \sum_{k=t+1}^n \binom{n}{k} \phi^k\\
    &\leq \frac{8nD}{\kappa}\sum_{k\geq t+1} (n\phi)^k
\end{align*}
Note that $\phi < \kappa + \theta < 1/(4n)$ and $\phi n < 1/10$. Also note that $\phi \leq (1+1/n)\kappa$. The last term is at most $ \frac{10nD}{\kappa} \left( 1 + 1/(10n) \right)^{t+1}\kappa^{t+1} n^{t+1}$. The lemma holds by using $D \leq 5000d$.
\end{proof}

Next, we give our algorithm to approximate $f_t(x)$. We can expand $f_t$ as follows 
\begin{align}\label{eq:exp-ft}
f_t(x) &\defeq \frac{1}{2}\sum_{y \in \{\pm\}^S: \Vert y\Vert_+\leq t} \left|\frac{\nu_B(x)}{\mu_B(x)}  \nu_S^x(y)-\mu_S^x(y)\right|\notag\\
 &= \frac{1}{2}\sum_{y \in \{\pm\}^S: \Vert y\Vert_+\leq t} \left| {\left(\prod_{v\in B:x_v=1} \frac{\lambda_v^\nu}{\lambda_v^\mu} \right)} \cdot {\frac{Z_{S,\nu}^x}{Z_{S,\mu}^x}}\cdot {\frac{Z_\mu}{Z_\nu}} \cdot \nu_S^x(y)-\mu_S^x(y)\right|.
\end{align}

We now introduce the approximation of $Z_{S,\nu}^x,Z_{S,\mu}^x,\mu^x_S,\nu^x_S$ in the above formula. 
Let $\Omega^x \subseteq \{\pm\}^S$ be the set of all $y \in \{\pm\}^S$ such that $x + y$ forms an independent set in graph $G$
For any integer $t \geq 0$, let $\Omega_t^x$ be  all $y \in \Omega^x$ such that $\Vert y \Vert_+ \leq t$.
Define $\ux$ as the distribution $\mu_S^x$ restricted on $\Omega^x_t$. Formally,
\begin{align}\label{def:zx}
\forall y \in \Omega_{t}^x, \quad \ux(y) = \frac{\prod_{v \in S: y_v = +1}\lambda^\mu_v}{\Zux} \qquad\text{ and }\qquad \Zux=\sum_{\tau \in \Omega^x_{t}} \prod_{v \in S: \tau_v = +1}\lambda^\mu_v.
\end{align}
Similarly, we can define $\vx$ and $\Zvx$ from $\nu_S^x$. The following approximation lemmas hold.
\begin{lemma}\label{lem:appZ}
Suppose $\kappa + \theta < 1/(10n)$ and $\theta/\kappa < 1 / (10n)$.
For any integer $t \geq 0$ and any $x \in \Omega_B$, it holds that
\begin{align*}
    \left\vert \frac{Z_{S,\nu}^x}{Z_{S,\mu}^x} - \frac{\Zvx}{\Zux} \right\vert \leq \eta(\kappa,t)\cdot d \quad \text{and} \quad \left\vert \frac{Z_{S,\mu}^x}{Z_{S,\nu}^x} - \frac{\Zux}{\Zvx} \right\vert \leq \eta(\kappa,t) \cdot d,
\end{align*}
where $\eta(\kappa,t)$ is defined in~\eqref{eq:error-t}.
\end{lemma}
\begin{proof}
Due to the symmetry, we only consider the first term in our proof.
Define $\Omega^x_{> t}=\Omega^x \setminus\Omega^x_{t}$.
Because $Z_{S,\mu}^x,\Zux\geq 1$, similar to the proof of \Cref{lem:appf}, we have
\begin{align*}
    &\left\vert \frac{Z_{S,\nu}^x}{Z_{S,\mu}^x} - \frac{\Zvx}{\Zux} \right\vert \leq 
    \left\vert \Zux Z_{S,\nu}^x - Z_{S,\mu}^x \Zvx \right\vert\\
    &=\left\vert \ \sum_{\tau \in \Omega^x_{t}} \prod_{v \in S: \tau_v = +1}\lambda^\mu_v \sum_{\varphi \in \Omega^x} \prod_{u \in S: \varphi_u = +1}\lambda^\nu_u - \sum_{\tau \in \Omega^x_{ t}} \prod_{v \in S: \tau_v = +1}\lambda^\nu_v \sum_{\varphi \in \Omega^x} \prod_{u \in S: \varphi_u = +1}\lambda^\mu_u\right\vert\\
    &=\left\vert \sum_{\tau \in \Omega^x_{ t}}\sum_{\varphi \in \Omega^x_{> t}}\left(\prod_{v \in S: \tau_v = +1}\lambda^\mu_v \prod_{u \in S: \varphi_u = +1}\lambda^\nu_u-\prod_{v \in S: \tau_v = +1}\lambda^\nu_v\prod_{u \in S: \varphi_u = +1}\lambda^\mu_u \right ) \right\vert\\
    &\leq  \sum_{i=0}^{t}\sum_{j=t+1}^{n}\binom{n}{i}\binom{n}{j}\left((\lambda_{\max}+D)^{i+j}-\lambda_{\max}^{i+j} \right).
\end{align*}
In the proof of \Cref{lem:appf}, we already show that  $(\lambda_{\max}+D)^k - \lambda_{\max}^k \leq \phi^{k-1}\cdot (2kD)$, where $\phi = \kappa + \theta$. Note that $\binom{n}{j}\leq n^j/(t+1)!$ and for any $k$, there are at most $(t+1)$ pairs $(i,j)$ s.t. $i+j=k$. We have
\begin{align*}
 \left\vert \frac{Z_{S,\nu}^x}{Z_{S,\mu}^x} - \frac{\Zvx}{\Zux} \right\vert  \leq \frac{2(n+t)D}{\phi t!}   \sum_{k=t+1}^{n+t}{n^k}\cdot \phi^{k} \leq \frac{4nD}{\kappa} \sum_{k \geq t+1}(n\phi)^t.  
\end{align*}
The last term is at most $\eta(\kappa,t) \cdot d$.
\end{proof}

\begin{lemma}\label{lem:sum-bound}
Suppose $\kappa + \theta < 1/(10n)$ and $\theta/\kappa < 1 / (10n)$. For any integer $t \geq 0$,    it holds that 
\begin{align*}
\sum_{y \in \Omega^x_t} \vert \ux(y) - \mu^x_S(y) \vert \leq \eta(\kappa,t) \cdot d \quad \text{and} \quad \sum_{y \in \Omega^x_t} \vert \vx(y) - \nu^x_S(y) \vert \leq \eta(\kappa,t) \cdot d.
\end{align*}
\end{lemma}

\begin{proof}
Due to the symmetry of $\mu$ and $\nu$, we only prove the first term here. We know that $\Zux\leq Z_{S,\mu}^x$,
\begin{align*}
    \sum_{y \in \Omega^x_t} \vert \ux(y) - \mu^x_S(y) \vert&=\sum_{y \in \Omega^x_t} \vert \frac{\prod_{v\in S:y_v=1}\lambda_v^\mu}{\Zux} - \frac{\prod_{v\in S:y_v=1}\lambda_v^\mu}{Z_{S,\mu}^x} \vert =\frac{Z_{S,\mu}^x-\Zux}{\Zux Z_{S,\mu}^x}\sum_{y \in \Omega^x_t}\prod_{v\in S:y_v=1}\lambda_v^\mu\\
    &=\frac{Z_{S,\mu}^x-\Zux}{\Zux Z_{S,\mu}^x}\Zux=\frac{Z_{S,\mu}^x-\Zux}{Z_{S,\mu}^x}\leq Z_{S,\mu}^x-\Zux\\
    &=\sum_{y\in \Omega_{>t}^x}\prod_{v\in S:y_v=1}\lambda_v^\mu\leq \sum_{y\in \Omega_{>t}^x}\lambda_{\max}^{\Vert y \Vert_+}\leq \sum_{k=t+1}^n\binom{n}{k}\phi^k\leq \frac{1}{(t+1)!}\sum_{k\geq t+1}(n\phi)^k.
\end{align*}
The last term is at most $\eta(\kappa,t) \cdot d$.
\end{proof}

Now, we are ready to give our algorithm.
First consider the ratio $R = \frac{Z_\nu}{Z_\mu}$ in~\eqref{eq:exp-ft}. We have the following algorithm to approximate $R$. 

\begin{lemma}\label{lem:alg-ratio}
If $\eta(\kappa,t)\leq \frac{\epsilon}{200}$, $\theta/\kappa < 1/(10n)$, and $\theta + \kappa < 1/(10n)$, where $\eta(\kappa,t)$ is defined in~\eqref{eq:error-t}, then there exists a randomized algorithm that computes a random number $\tilde{R}$ in time $T' \cdot \TS_G(T'/10^3) + T' \cdot O(n^t)$, where $T'=O(\frac{n^3 + n/\kappa}{\epsilon^2})$ and $\TS_G(\cdot)$ is the cost of sample oracle for $\mu$, such that with probability at least $0.99$,
\begin{align}\label{eq:tildeR}
   \left\vert \tilde{R} - \frac{Z_\nu}{Z_\mu} \right\vert \leq \frac{\epsilon}{100} d.
\end{align}
\end{lemma}

We first prove \Cref{lem:hardcore-adv-1} assuming \Cref{lem:alg-ratio}. \Cref{lem:alg-ratio} will be proved later.

\ifthenelse{\boolean{conf}}{\begin{proof}\textbf{of \Cref{lem:hardcore-adv-1}}}{\begin{proof}[Proof of \Cref{lem:hardcore-adv-1}]}
We fix the parameter $t = 4$. 
Recall that $\kappa = 10^{-9}\frac{\epsilon^{1/4}}{n^{3/2}}$ and $\theta = 10^{-10}\frac{\epsilon^{1/4}}{n^{5/2}}$.
We can verify that $\eta(\kappa,t) = 10^6\left(1+\frac{n}{10}\right)^{t+1} \kappa^t n^{t+2} \leq \frac{\epsilon}{200}$, $\theta + \kappa < 1/ (10n)$, and $\theta / \kappa < 1/(10n)$.

In the construction step, we use \Cref{lem:alg-ratio} to compute the random number $\tilde{R}$.
We say the construction step succeeds if~\eqref{eq:tildeR} is satisfied, which happens with probability at least $0.99$. 

In the query step, given any $x \in \Omega_B$, our data structure answers the following $\hat{f}(x)$:
\begin{align}\label{eq:hatf}
    \hat{f}(x) = \frac{1}{2}\sum_{y \in \Omega^x_t} \left| {\left(\prod_{v\in B:x_v=1} \frac{\lambda_v^\nu}{\lambda_v^\mu} \right)} \cdot \frac{\Zvx}{\Zux}\cdot \tilde{R} \cdot \vx(y)-\ux(y)\right|.
\end{align}

We bound the approximation error of $\hat{f}$. Define $A = \left(\prod_{v\in B:x_v=1} \frac{\lambda_v^\nu}{\lambda_v^\mu} \right) \cdot \frac{\Zvx}{\Zux}\cdot \tilde{R}$ and $B = \left(\prod_{v\in B:x_v=1} \frac{\lambda_v^\nu}{\lambda_v^\mu} \right) \cdot \frac{Z_{S,\nu}^x}{Z_{S,\mu}^x}\cdot \frac{Z_\nu}{Z_\mu}$. By \Cref{lem:Zcond-bound} and the analysis in the proof of \Cref{lem:marginal-ratio}, we have $\left(\prod_{v\in B:x_v=1} \frac{\lambda_v^\nu}{\lambda_v^\mu} \right) < (1+\theta/\kappa)^n < 2$, $\frac{Z_{S,\nu}^x}{Z_{S,\mu}^x}\leq 2$, and $\frac{Z_\nu}{Z_\mu}\leq 2$. Also note that $d \leq 1$. Using \Cref{lem:appZ} and our assumption on $\tilde{R}$, it holds that 
\begin{align*}
    |A-B| \leq \frac{\epsilon}{25}d.
\end{align*}
Using triangle inequality, we can bound
\begin{align*}
    |f^t(x) - \hat{f}(x)| \leq B\sum_{y \in \Omega^x_t}\vert \vx(y) - \nu^x_S(y) \vert + |A-B|\sum_{y \in \Omega^x_t}\vx(y) + \sum_{y \in \Omega^x_t}\vert \ux(y) - \mu^x_S(y) \vert.
\end{align*}
Using the fact $B\leq 8$ and \Cref{lem:sum-bound}, we have $ |f^t(x) - \hat{f}(x)| \leq \frac{\epsilon}{10}d$. By \Cref{lem:appf},
\begin{align*}
     |f(x) - \hat{f}(x)| \leq \frac{\epsilon}{8}d.
\end{align*}

The construction step takes time $T' \cdot \TS_G(T'/10^3) + T' \cdot O(n^4)$, where $T' = O(\frac{n^3}{\epsilon^2}+\frac{n^{5/2}}{\epsilon^{9/4}})$. Since the hardcore model satisfies the uniqueness condition in~\eqref{eq:cond-hardcore}, we have $\TS_G(T'/10^3) = O_\eta(\Delta n\log\frac{n}{\epsilon})$~\citep{CFYZ22,CE22}. 
The total construction time is $\tilde{O}(\frac{n^7}{\epsilon^2}+\frac{n^{6.5}}{\epsilon^{9/4}})$.
For each query, the running time is dominated by computing distributions $\ux$ and $\vx$, which is $O(n^t)=O(n^4)$.
\end{proof}

Finally, we prove \Cref{lem:alg-ratio}. 

\ifthenelse{\boolean{conf}}{\begin{proof}\textbf{of \Cref{lem:alg-ratio}}}{\begin{proof}[Proof of \Cref{lem:alg-ratio}]}
Recall our definition of $\mu_B$ and $\nu_B$ is that for any $x \in \Omega_B$,
\begin{align*}
    \mu_B(x)=\frac{\prod_{v\in B:x_v=1}\lambda_v^\mu Z_{S,\mu}^x}{Z_\mu},\quad \nu_B(x)=\frac{\prod_{v\in B:x_v=1}\lambda_v^\nu Z_{S,\nu}^x}{Z_\nu}.
\end{align*}
We can compute that 
\begin{align*}
    \frac{Z_\nu}{Z_\mu}=\sum_{x\in\Omega_B}\mu_B(x)\frac{\prod_{v\in B:x_v=1}\lambda_v^\nu Z_{S,\nu}^x}{\prod_{v\in B:x_v=1}\lambda_v^\mu Z_{S,\mu}^x}=\mathbf{E}_{x \sim \mu_B}\Big[{\underbrace{\frac{\prod_{v\in B:x_v=1}\lambda_v^\nu Z_{S,\nu}^x}{\prod_{v\in B:x_v=1}\lambda_v^\mu Z_{S,\mu}^x}}_{\defeq Q(x)}}\Big].
\end{align*}
We estimate $Z_\nu/Z_\mu$ by sampling $x$ from $\mu_B$, approximating $Q(x)$ and taking the average. We propose the algorithm as follow. Let $T' = O (\frac{n^3 + n/\kappa}{\epsilon^2})$ be a sufficiently large integer.
\begin{itemize}
    \item Draw $T'$ independent approximate samples $x_1,\cdots,x_{T'}$ from $\mu_B$ with $\DTV{\mu_B}{x_i}\leq \frac{1}{1000T'}$. 
    \item Compute $\tilde{Q}(x)=\frac{\prod_{v\in B:x_v=1}\lambda_v^\nu \Zvx}{\prod_{v\in B:x_v=1}\lambda_v^\mu \Zux}$ for $x=x_1,\cdots x_{T'}$, where $\Zux,\Zvx$ are defined in~\eqref{def:zx}.
    \item Compute $\tilde{R} = \frac{1}{T'}\sum_{i=1}^{T'}\tilde{Q}(x_i)$.
\end{itemize}

First, for any $x\in \Omega_B$,
by \Cref{lem:appZ}, since $\eta(\kappa,t)\leq \frac{\epsilon}{200}$ and $D/\kappa\leq \theta/\kappa\leq 1 / (10n)$,
\begin{align*}
    \left \vert Q(x)-\tilde{Q}(x)\right \vert&=\left \vert \frac{\prod_{v\in B:x_v=1}\lambda_v^\nu Z_{S,\nu}^x}{\prod_{v\in B:x_v=1}\lambda_v^\mu Z_{S,\mu}^x}-\frac{\prod_{v\in B:x_v=1}\lambda_v^\nu \Zvx}{\prod_{v\in B:x_v=1}\lambda_v^\mu \Zux} \right \vert=\left\vert \frac{Z_{S,\nu}^x}{Z_{S,\mu}^x} - \frac{\Zvx}{\Zux}\right \vert \prod_{v\in B:x_v=1}\frac{\lambda_v^\nu}{\lambda_v^\mu}\\
    &\leq \frac{\epsilon}{200}d \left(\frac{\kappa+D}{\kappa}\right)^n \leq \frac{\epsilon d}{150}.
\end{align*}
Recall $h(x)=\nu_B(x)/\mu_B(x)$, the variance of $Q(x)$ is
\begin{align*}
    \Var[\mu_B]{Q} = \Var[x\sim \mu_B]{\frac{\prod_{v\in B:x_v=1}\lambda_v^\nu Z_{S,\nu}^x}{\prod_{v\in B:x_v=1}\lambda_v^\mu Z_{S,\mu}^x}}=\Var[x\sim \mu_B]{\frac{Z_\nu \nu_B(x)}{Z_\mu \mu_B(x)}}=\frac{Z_\nu^2}{Z_\mu^2}\Var[\mu_B]{h}.
\end{align*}
In~\eqref{eq:upperZ/Z}, we showed that $\frac{Z_\mu}{Z_\nu}\leq 1+\frac{4nD}{\kappa}$, which implies $\frac{Z_\nu^2}{Z_\mu^2}\leq \left(1+\frac{4nD}{\kappa} \right )^2\leq e^{4/5}$. In~\eqref{eq:varh}, we also proved that $\Var[\mu_B]{h}\leq O_\eta(d^2) \cdot \left(n^3 + \frac{n}{\kappa}\right)$. We can conclude that 
\begin{align*}
    \Var[\mu_B]{Q} \leq O_\eta(d^2) \cdot \left(n^3 + \frac{n}{\kappa}\right).
\end{align*}

Assume that we have an ideal algorithm that draw perfect samples $x_1,\ldots,x_{T'}$ and exactly compute $Q(x_1),\cdots,Q(x_{T'})$ and compute $R^*=\frac{1}{T'}\sum_{i=1}^{T'}Q(x_i)$. Note $T' = O (\frac{n^3 + n/\kappa}{\epsilon^2})$. By Chebyshev's inequality, if $T'$ is sufficiently large, we have
\begin{align*}
    \Pr[]{|R^*-Z_\nu/Z_\mu|\geq \frac{\epsilon d}{300}}\leq 0.005.
\end{align*}
Note that  $\left \vert Q(x)-\tilde{Q}(x)\right \vert \leq \frac{\epsilon d}{150}$ for all $x \in \Omega_B$.
All the approximate samples $x_1,\cdots x_{T'}$ can be coupled successfully with perfect samples with probability at least $1-T'\cdot \frac{1}{1000T'}=0.999$.
We can couple our algorithm with ideal algorithm such that $|R^*-\tilde{R}|\leq \frac{\epsilon d}{150}$ with probability at least $0.99$. By a union bound, with probability at least 0.99,
\begin{align*}
    \left \vert \tilde{R}-\frac{Z_\nu}{Z_\mu} \right \vert \leq \frac{\epsilon d}{150}+\frac{\epsilon d}{300} = \frac{\epsilon d}{100}.
\end{align*}

The running time of our algorithm is $T' \cdot \TS_G(T'/10^3) + T' \cdot O(n^t)$ because each $\tilde{Q}(x)$ can be computed in time $O(n^t)$.
\end{proof}



\section{Proofs of algorithmic results}\label{sec:proof-main}

\subsection{The general algorithm (Proof of \texorpdfstring{\Cref{thm:Ising-1}}{Lg})}\label{sec:proof-gen}

\paragraph{Compute marginal lower bound}
In our main theorem, the input instance is promised to be $b$-marginally bounded for some parameter $b$. However, the specific value of $b$ is not given to the algorithm. The following algorithm computes the tight value of marginal lower bound.  
\begin{lemma}\label{lem:alg-b}
There exists an algorithm such that given any hardcore or Ising model $\mu$ in graph $G=(V,E)$, 
it returns a value $b$ in time at most $(1/b)^{O(1/b)} n$ for hardcore model and in time at most $O(n+m)$ for Ising model such that 
\begin{align}\label{eq:def-b}
    b = \min \{\mu^\sigma_v(c) \mid v\in V, \sigma \text{ is a feasible partial pinning, and } \mu^\sigma_v(c) > 0 \}.
\end{align}
\end{lemma}
The proof of \Cref{lem:alg-b} will be given later. With this lemma, we can assume that the value of $b$ is known to the algorithm.

\paragraph{Pre-processing step}
For Ising model, we need the following pre-processing step to reduce the general Ising model to the soft-Ising model.
Recall that an Ising model $(G=(V,E),J,h)$ is said to be \emph{soft} if $h \in \mathbb{R}^V$ instead of $h \in (\mathbb{R} \cup \{\pm \infty\})^V$. There are three cases: 
\begin{enumerate}
    \item Case 1: If there exists $v$ such that ($h_v^\mu =  +\infty,h_v^\nu = - \infty$)  or ($h_v^\mu = -\infty, h_v^\nu = +\infty$), we can direct compute that $\DTV{\mu}{\nu}=1$.
    \item Case 2: If there exists $v$ such that ($h_v^\mu =  \pm \infty,h_v^\nu \neq \pm \infty$) or ($h_v^\mu \neq \pm \infty, h_v^\nu = \pm \infty$), without loss of generality, we consider the case ($h_v^\mu =  +\infty,h_v^\nu \neq \pm \infty$). We have 
    \begin{align*}
        \DTV{\mu}{\nu}\geq |\mu_v(-)-\nu_v(-)|=|0-\nu_v(-)|=\nu_v(-) \geq b,
    \end{align*} 
    where the last inequality due to the marginal lower bound.
    We use additive-error algorithm in \Cref{thm:Approximate-Gibbs} with additive error $b\epsilon$ in this case. The running time is 
    \begin{align}\label{eq:time-pre}
        O\tp{\TC_G\tp{\frac{b\epsilon}{4}}+ \frac{1}{b^2\epsilon^2}\tp{ \TW_{G} + \TS_G\tp{\frac{b\epsilon}{4}}}} =O\tp{\TC_G\tp{\frac{b\epsilon}{4}}+ \frac{1}{b^2\epsilon^2}\tp{\TS_G\tp{\frac{b\epsilon}{4}}}},
    \end{align}
    where the equation holds because $\TW_{G} = O(n+m)$ and we can assume $\TS_G(\cdot),\TC_G(\cdot)$ is at least $\Omega(n+m)$ since the algorithm needs to read all vertices and all edges.
    
    \item Case 3: For all $v\in V$, if $h_v^\nu=\pm \infty$ or $h_v^\mu=\pm \infty$, then $h_v^\nu=h_v^\mu$. These vertices are fixed to some value with probability 1. By the standard self-reducibility, one can remove all these vertices and change external fields of neighbors to obtain two soft-Ising models. Formally, one can go through all vertices $v \in V$ whose value is fixed as $c \in \{\pm\}$, for every free neighbor $u$ of $v$, update $h^\mu_u \gets h^\mu_u + J^\mu_{uv}c$ and $h^\nu_u \gets h^\nu_u + J^\nu_{uv}c$.

    We remark that (1) the new soft-Ising models also have the same marginal lower bound because we only remove vertices whose marginal lower bound is 1; (2) to sample from the soft-Ising model, one can call sampling oracle on original model and do a projection; to approximately count the partition function, one can use the approximating counting oracle on the original model, because two partition functions differ only by an easy-to-compute factor.
\end{enumerate}

We also do the pre-processing step for hardcore model. A hardcore model $(G,\lambda)$ is said to be soft if $\lambda \in \mathbb{R}_{>0}^V$ instead of $\lambda \in \mathbb{R}_{\geq 0}^V$. There two cases.
\begin{enumerate}
    \item Case 1: There exists $v$ such that $(\lambda^\mu_v = 0,\lambda^\nu_v > 0)$ or $(\lambda^\mu_v > 0,\lambda^\nu_v = 0)$, then $\DTV{\mu}{\nu} \geq b$. We use additive-error algorithm to solve the problem in time~\eqref{eq:time-pre}.
    \item Case 2: For all $v$, $\lambda^\mu_v = 0$ if and only if $\lambda^\nu_v = 0$. We can simply remove all such vertices and work on the soft-hardcore model on the remaining graph. Again, the marginal lower bound and sampling/approximate counting oracles also work for new soft-hardcore model.
\end{enumerate}

\paragraph{The main algorithm}
Since we work on soft models, $b\leq \frac{1}{2}$. 
Define the parameters
\begin{align*}
    \CC = \begin{cases}
         b^3 &\text{hardcore model},\\
         \frac{b^2}{2} &\text{Ising model}.
     \end{cases}, \text{ and } 
     \Thre = \begin{cases}
         \frac{b}{2(1-b)n} &\text{hardcore model},\\
         \frac{1}{2(n+3m)} &\text{Ising model}.
     \end{cases}
\end{align*}
The algorithm computes the parameter distance $\dis(\mu,\nu)$ in time $O(n + m)$.
If $\dis(\mu,\nu) \geq \theta$, then by \Cref{lem:TV-lower}, $\DTV{\mu}{\nu} \geq \theta \CC$, we use additive-error algorithm in \Cref{thm:Approximate-Gibbs} with additive error $\theta \CC \epsilon$. Similar to~\eqref{eq:time-pre}, the running time is 
\begin{align}\label{eq:time-add}
   O\tp{\TC_G\tp{\frac{\theta \CC \epsilon}{4}}+ \frac{1}{(\theta \CC \epsilon)^2}\tp{\TS_G\tp{\frac{\theta \CC \epsilon}{8}}}}.
\end{align}

Next, assume that $\dis(\mu,\nu) < \theta$. 
We use the basic algorithm in \Cref{thm:alg-main}.
We have the following two lemmas for the soft-hardcore and soft-Ising models.

\begin{lemma}\label{lem:par-hardcore}
Let $\mu$ and $\nu$ be two soft-hardcore models satisfying $b$-marginal lower bound and $\dis(\mu,\nu) \leq \theta$. Then $\mu$ and $\nu$ satisfy \Cref{cond:meta} with $K = {4n}/{(b\CC)}$ and $L = 2$.
\end{lemma}

\begin{lemma}\label{lem:par-Ising}
Let $\mu$ and $\nu$ be two soft-Ising models satisfying $\dis(\mu,\nu) \leq \theta$ and $\DTV{\mu}{\nu} \geq \CC \dis(\mu,\nu)$. Then $\mu$ and $\nu$ satisfy \Cref{cond:meta} with $K = 4(n+m)/\CC$ and $L = 2$.
\end{lemma}

Assuming the above two lemmas, we can use \Cref{thm:alg-main} to solve the problem in time 
\begin{align}\label{eq:time-multi}
  O\tp{\TC_G\tp{\frac{\epsilon}{4}} + T \cdot \TS_G\tp{\frac{1}{100T}}}, \quad\text{where } T = O\tp{\frac{L^2K^2}{\epsilon^2}}.
\end{align}
The final running time of our algorithm is dominated by the maximum of~\eqref{eq:time-pre},~\eqref{eq:time-add},~\eqref{eq:time-multi}, and the running time in \Cref{lem:alg-b}. Since both $\TS_G,\TC_G$ are non-increasing functions, the running time of our algorithm is at most
\begin{align*}
    C_b \cdot \frac{N^2}{\epsilon^2} \TS_G \tp{ \frac{\epsilon^2}{C_b N^2}} + \TC_G\tp{\frac{\epsilon}{C_b N}} + C'_b N,
\end{align*}
where $C_b, C'_b \geq 1$ are parameters depending only on $b$ and for hardcore model, $N = n$; for Ising model, $N = n+ m$. 
For hardcore model
\begin{align*}
    C_b = \mathrm{poly}\tp{\frac{1}{b}}, \quad C'_b  = \tp{\frac{1}{b}}^{O(\frac{1}{b})};
\end{align*}
and for Ising model,
\begin{align*}
    C_b = \mathrm{poly}\tp{\frac{1}{b}}, \quad C'_b = O(1).
\end{align*}
The parameter $C'_b$ comes from the running time in \Cref{lem:alg-b}.

\begin{remark}\label{remark:b}
\Cref{thm:Ising-1} presents a simplified version that assumes the marginal lower bound \( b \) to be a constant. However, our algorithm applies to both the Ising and hardcore models with an arbitrary marginal lower bound \( b \), where \( b \) may depend on the size of the input. The running time of our algorithm is given by
\begin{align*}
    \mathrm{poly}\tp{\frac{1}{b}} \cdot \frac{N^2}{\epsilon^2} \TS_G \tp{ \frac{\mathrm{poly}(b) \epsilon^2}{N^2}} + \TC_G\tp{\frac{\mathrm{poly}(b) \epsilon}{ N}} + C'_b N,
\end{align*}    
\begin{itemize}
    \item \textbf{Ising model:} The parameter \( C'_b = O(1) \), so a polynomial-time reduction from TV-distance estimation to sampling and approximate counting exists for the Ising model with marginal lower bound \( b \geq \frac{1}{\mathrm{poly}(n)} \).
    \item \textbf{Hardcore model:} The parameter  \( C'_b = (1/b)^{O(1/b)} \), which comes from the running time in \Cref{lem:alg-b}. One can improve the last term \( C'_b N \) to \( \mathrm{poly}(n) \cdot \TC_G(\frac{1}{10}) \) by assuming a slightly stronger approximate counting oracle. In \Cref{thm:Ising-1}, we only assume approximate counting oracles for \( \mu \) and \( \nu \). If we further assume that approximate counting oracles work for all conditional distributions induced by \( \mu \) and \( \nu \), then by going through the proof of \Cref{lem:alg-b}, we can obtain an algorithm that computes \( b' \) in time \( \mathrm{poly}(n) \cdot \TC_G(\frac{1}{10}) \), such that with high probability, the value \( b' \) satisfies \( \frac{b}{2} \leq b' \leq b \) for \( b \) defined in~\eqref{eq:def-b}. This \( b' \) is also a marginal lower bound and provides a constant approximation to the true lower bound. We can then use this \( b' \) in the remainder of the algorithm. All the subsequent proofs follow for this $b'$.  Hence, given the stronger approximate counting oracle, the polynomial-time reduction exists for the hardcore model with marginal lower bound $b \geq \frac{1}{\mathrm{poly}(n)}$. 
\end{itemize}
\end{remark}

Finally, we give the proofs of technical lemmas.
We need the following property about the soft hardcore model.
\begin{lemma}\label{lem:const-degree}
Let $0 < b < 1$. Suppose a soft-hardcore model $(G,\lambda)$ is $b$-marginally bounded. For any vertex $v \in V$, let $\deg_v^{\text{free}}$ denote the number of free neighbors $u$ of $v$ such that $\lambda_u > 0$. For any $v \in V$, it holds that $\deg_v^{\text{free}} \leq \frac{\ln b}{\ln (1- b)}$.
\end{lemma}
\begin{proof}
Let $N^2(v)$ denote the set of vertices with distance 2 to vertex $v \in V$ in graph $G$. 
Let $\sigma$ be a pinning that fixes all vertices in $N^2(v)$ to the value $-1$. 
Let $N^{\text{free}}(v)=\{v_1,v_2,\ldots,v_\ell\}$ denote the set of free neighbors of $v$, where $\ell = \deg^{\text{free}}_v$. 
Let $\mu$ denote  the Gibbs distribution.
Conditional on $\sigma$, $v$ takes value $+$ only if all free neighbors of $v$ take the value $-$. 
Since the hardcore model is soft, $v$ takes $+$ with a positive probability so that the marginal lower bound applies.
We have
\begin{align*}
    b\leq \mu_v^\sigma(+) \leq \prod_{j = 1}^\ell \mu_{v_j}(-\mid \sigma \text{ and } (\forall k < j, v_k \gets -)). 
\end{align*}
For any free neighbor $v_j$, conditional on $\sigma$ and $ v_k \gets -$ for all $k < j$, $v_j$ takes $+$ with a positive probability so that $v_j$ takes $+$ with probability at least $b$. Hence,
\begin{align*}
    b \leq (1-b)^{\deg^{\text{free}}_v}. 
\end{align*}
Note that $1 - b < 1$. This proves the upper bound of the free degree.
\end{proof}

\ifthenelse{\boolean{conf}}{\begin{proof}\textbf{of \Cref{lem:alg-b}}}{\begin{proof}[Proof of \Cref{lem:alg-b}]}
For hardcore model, if $\lambda_u=0$ for some $u\in V$, we fix $u=-$ and consider the remained subgraph. If the subgraph has no vertices, we just return b = 1. 

When $c=-$, for each any partial configuration $\tau$ on $\Lambda \subseteq V$ where $v \notin \Lambda$, consider the marginal probability $\mu_v^\tau(-)$.
Then $\mu_v^\tau(-)$ is a convex combination of $\mu_v^\sigma(-)$'s, where $\sigma$ is a partial configuration on $V  \setminus \{v\}$.
For any $\sigma$, $\mu^\sigma_v(-)$ is positive.
Due to the conditional independence, we only need to consider the worst pinning of on $N(v)$, where $N(v)$ is the set of all neighbors of the vertex $v$.
It is easy to see $\mu^\sigma_v(-) \geq \frac{1}{1+\lambda_v^\mu}$, and equality is achieved when $\sigma$ fixes the values of all neighbors to be $-$. 

Now we consider the case $c = +$ and fix a vertex $v$. Consider a partial configuration $\sigma\in \{\pm\}^\Lambda$ for subset $\Lambda \subseteq V\setminus \{v\}$. If there exists a vertex $u\in N(v)$ such that $u\in \Lambda$, to make $\mu_v^\sigma(+)$ nonzero, $\sigma(u)$ must be $-$. 
Assume $u \in N(v)$ and $\sigma_v = -$.
We consider another subset $\sigma'$ on $\Lambda' = \Lambda \setminus \{u\}$ such that $\sigma'_{\Lambda'} = \sigma_{\Lambda'}$. Define the notation
\begin{align*}
    w_\mu(\sigma,v = +) = \sum_{\tau \in \{\pm\}^V:\tau_\Lambda  = \sigma \land \tau_v = +} w_\mu(\tau).
\end{align*}
We have $w_\mu(\sigma,v = +) = w_\mu(\sigma',v = +)$, because $\tau_v = +$ forces all vertices in $N(v)$ to take the value $-$. On the other hand, $w_\mu(\sigma,v = -) \leq w_\mu(\sigma',v = -)$, because $u$ in $\sigma_0$ is free and it can either take $-$ or $+$.
Our goal is to find a condition such that $v$ takes $+$ with the minimum positive probability. We can assume $N(v) \cap \Lambda = \emptyset$.
Again, $\mu^\sigma_v(+)$ is a convex combination over all $\mu^\tau_v(+)$, where the feasible partial configuration $\tau \in \{\pm\}^{V\setminus (N(v)\cup v)}$ that fixes the value of all vertices except $N(v) \cup \{v\}$. We only need to consider the worst case of $\tau$. Note that
\begin{align*}
    \frac{w_\mu(\tau,v=+)}{w_\mu(\tau,v=-)} = \frac{\lambda_v^\mu}{\sum_{\rho \in \{\pm\}^{N(v)},w_\mu(\tau,\rho,-)>0 }\prod_{u\in N(v),\rho_u=+}\lambda_u^\mu}.
\end{align*}
It is easy to verify when $\tau = \tau_0$ such that for all $u \in {V\setminus (N(v)\cup v)}$, $\tau_0(u)= -$,
the above ratio obtains its minimum, because other $\tau$ may forbid some possible $\rho$ in the summation. Our algorithm for computing the value of $b$ is:
\begin{itemize}
    \item Compute $b_0=\frac{1}{1+\max_{v\in V}(\lambda_v^\mu)}$.
    \item For each $v\in V$, compute $m_v=\mu_{v}^{\tau_0}(+)$ by enumerating all independent sets of $N(v)$.
    \item Output $b=\min \{b_0,\min_{v\in V} \{m_v\}\}$.
\end{itemize}

By \cref{lem:const-degree}, because we already remove all vertices with zero $\lambda_v$, we have $N(v)\leq \frac{\ln(b)}{\ln (1-b)}$.
Let $k = \max_{v \in V}|N(v)| = O(\frac{1}{b}\log \frac{1}{b})$. 
The running time of the algorithm is 
\[O(n 2^k k^2) = \tp{\frac{1}{b}}^{O(\frac{1}{b})} n.\] 
The running time $O(k2^k)$ is for the exact computation of the $\mu^{\tau_0}_v(+)$. As stated in \Cref{remark:b}, given the approximate counting oracle for conditional distributions, we can compute an approximate value $p_v$ such that $\frac{1}{2} \mu^{\tau_0}_v(+) \leq p_v \leq \mu^{\tau_0}_v(+)$ with high probability.

For Ising model, similar to the pre-processing step, we can first remove all $v$ with $|h_v| = \infty$ and then change the external fields on all neighbors of $v$. After this step, we only need to consider a soft-Ising model  $(G,J,h)$ in the remaining graph $G$.

We also analyze when $\mu_v^\sigma(c)$ obtains the minimum. 
Since we deal with soft-Ising model, any $\sigma \in \{\pm\}^V$ appears with positive probability. Since $\mu_v^\sigma(c)$ is a convex combination of $\mu^\tau_v(c)$, where $\tau$ is a pinning on $V \setminus \{v\}$. Due to the conditional independence, if all $N(v)$ is fixed, then other vertices do not influence on the marginal distribution at $v$. Hence, to minimize $\mu_v^{\sigma}(c)$, we only need to consider $\sigma\in \{\pm\}^{N(v)}$. The marginal distribution can be written as
\begin{align*}
\frac{\mu_v^\sigma(c)}{\mu_v^\sigma(-c)}=\frac{\exp(\sum_{u\in N(v)}J_{vu}\sigma_u c+h_v c)}{\exp(-\sum_{u\in N(v)}J_{vu}\sigma_u c-h_v c)} =\exp\tp{2\sum_{u\in N(v)}J_{vu}\sigma_u c+2h_v c}.
\end{align*}
To find a $\sigma$ that minimizes $\sum_{u\in N(v)}J_{vu}^\mu\sigma_u c$, we greedily assign $\sigma_u \in \{\pm\}$ according to the sign of $J_{vu}$. The final result is for any $c \in \{\pm\}$,
\begin{align*}
\min_{\sigma\in \{\pm\}^{N(v)}} \mu_v^\sigma(c)=\frac{g(v,c)}{g(v,c)+1/g(v,c)}\text{, where } g(v,c)= \exp\tp{h_vc -\sum_{u\in N(v)}|J_{vu}|}.  
\end{align*}
Our algorithm is:
\begin{itemize}
    \item For each $v\in V$, compute $g(v) = \min_{c \in \{\pm\}}\frac{g(v,c)}{g(v,c)+1/g(v,c)}$.
    \item Output $b=\min_{v\in V} \{g(v)\}$.
\end{itemize}
The running time is $O(n+m)$.
\end{proof}

\ifthenelse{\boolean{conf}}{\begin{proof}\textbf{of \Cref{lem:par-hardcore}}}{\begin{proof}[Proof of \Cref{lem:par-hardcore}]}
For all $\sigma\in \{\pm\}^V$, if $\mu(\sigma)>0$, $\sigma$ corresponds to an independent set of $G$. Because $\nu$ is soft-hardcore, then $\nu(\sigma)>0$, so $\nu$ is absolutely continuous with respect to $\mu$. 

For each $v\in V$, consider $\sigma = (-1)^{V\setminus v}$. Because $\mu$ and $\nu$ are both soft-hardcore models and satisfy the $b$-marginal lower bound, $\mu_{v}^\sigma(+),\mu_{v}^\sigma(-),\nu_{v}^\sigma(+),\nu_{v}^\sigma(-)\geq b$. We have 
$\frac{\lambda_u^x}{1+\lambda_u^x},\frac{1}{\lambda_u^x+1}\geq b$ for all $v\in V$, $x\in \{\mu,\nu\}$, which means
\begin{align*}
    \frac{b}{1-b}\leq \lambda_v^x\leq \frac{1-b}{b}.
\end{align*}
The above inequality means that $b\leq \frac{1}{2}$. We can compute the ratio of the weight
\begin{align*}
    \frac{w_\nu(\sigma)}{w_\mu(\sigma)}&=\prod_{v:\sigma(v)=+}\frac{\lambda_v^\nu}{\lambda_v^\mu}\leq \prod_{v:\sigma(v)=+} \frac{\lambda_v^\mu+\dis(\mu,\nu)}{\lambda_v^\mu}\\
    &\leq \left(\frac{\lambda_v^\mu+\dis(\mu,\nu)}{\lambda_v^\mu}\right)^n\leq \left(1+\frac{(1-b)\dis(\mu,\nu)}{b}\right)^n.
\end{align*}
For hardcore model, $\dis(\mu,\nu) \leq \theta=\frac{b}{2(1-b)n}$, so 
\begin{align*}
    \frac{w_\nu(\sigma)}{w_\mu(\sigma)}\leq 1+\frac{3n(1-b)\dis(\mu,\nu)}{b}.
\end{align*}
Similarly, $\frac{w_\nu(\sigma)}{w_\mu(\sigma)}\geq 1-\frac{n(1-b)\dis(\mu,\nu)}{b}$. By \Cref{lem:TV-lower}, $\DTV{\mu}{\nu} \geq \CC \dis(\mu,\nu)$, then
\begin{align*}
    \sqrt{\Var{W}}&\leq \frac{4n(1-b)\dis(\mu,\nu)}{b}\leq \frac{4n(1-b)}{b\CC}\DTV{\mu}{\nu} < \frac{4n}{b\CC}\DTV{\mu}{\nu} \text{, and}\\
    \E{W}&\geq \min_{\sigma} \frac{w_\nu(\sigma)}{w_\mu(\sigma)} \geq 1-\frac{1}{2}=\frac{1}{2}.
\end{align*}
This verifies \Cref{cond:meta}.
\end{proof}

\ifthenelse{\boolean{conf}}{\begin{proof} \textbf{of \Cref{lem:par-Ising}}}{\begin{proof}[Proof of \Cref{lem:par-Ising}]}
Since both $\mu$ and $\nu$ are soft-Ising models, the absolutely continuous condition $\nu \ll \mu$ holds.
    Let $J^\mu,h^\mu$ and $J^\nu,h^\nu$ be the interaction matrices and external field vectors of $\mu$ and $\nu$, respectively. Denote $D = \dis(\mu,\nu)$.
    For any $\sigma \in \{\pm\}^V$, we have
    \begin{align*}
        \frac{w_\nu(\sigma)}{w_\mu(\sigma)} = \exp\left( \sum_{u\in V} (h^\nu_u -h^\mu_u) \sigma_u + \sum_{\{u,v\}\in E} (J^\nu_{uv} - J^\mu_{uv})\sigma_u\sigma_v \right). 
    \end{align*}
    By the definition of parameter distance in \Cref{def:dis}, we have 

    \begin{align*}
        \left| \sum_{u\in V} (h^\nu_u -h^\mu_u) \sigma_u + \sum_{\{u,v\}\in E} (J^\nu_{uv} - J^\mu_{uv})\sigma_u\sigma_v \right| &\leq \sum_{u\in V}|h^\nu_u-h^\mu_u|+\sum_{\{u,v\}\in E} |J^\nu_{uv}-J^\mu_{uv}|
        \\
        &\leq \left(\sum_{u\in V} (\deg_u+1)+\sum_{\{u,v\}\in E}1 \right)D\\&=(n+3m)D,
    \end{align*}
    it implies
    \begin{align*}
        \exp(-(n+3m)D) \leq \frac{w_\nu(\sigma)}{w_\mu(\sigma)} \leq \exp((n+3m)D).
    \end{align*}
    For soft-Ising model, $\theta = \frac{1}{2(n+3m)}$ and $D < \theta$, so that 
    \begin{align*}
       1 - (n+3m)D \leq \frac{w_\nu(\sigma)}{w_\mu(\sigma)} \leq 1 + 3(n+3m)D.
    \end{align*}
    Hence, $\sqrt{\Var{W}}\leq 4(n+3m)D \leq \frac{4(n+3m)}{\CC}\DTV{\mu}{\nu}$, and $\E[]{W} \geq 1 - \frac{1}{2} = \frac{1}{2}$.
\end{proof}


\subsection{The improved algorithms for hardcore model in the uniqueness regime} \label{sec:hardcore-alg-proof}

\ifthenelse{\boolean{conf}}{\begin{proof}\textbf{of \Cref{thm:hardcore-1}}}{\begin{proof}[Proof of \Cref{thm:hardcore-1}]}
Consider a hardcore model $(G,\lambda)$, where $\lambda \in \mathbb{R}_{\geq 0}^V$, that satisfies the uniqueness condition.
Define threshold $\theta$ for hardcore model as 
\begin{align*}
  \theta = 10^{-10}\frac{\epsilon^{1/4}}{n^{5/2}} = \Theta\left(\frac{\epsilon^{1/4}}{n^{5/2}}\right). 
\end{align*}
If $\dis(\mu,\nu) > \theta$, then by \Cref{lem:TV-lower}, we know that $\DTV{\mu}{\nu} \geq \frac{\theta}{5000}$. 
we use the algorithm \Cref{thm:Approximate-Gibbs} to achieve the additive error $\frac{\epsilon D}{5000}$. 
For the hardcore model in the uniqueness regime, we have $\TS_G(\delta)=O_\eta(\Delta n \log \frac{n}{\delta})$ and $\TC_G(\delta) = \tilde{O}_\eta(\frac{\Delta n^2}{\delta^2})$. Note that $\TW_G = O(n)$. The running time for this case is at most 
\begin{align*}
 O\left(\TC_G\left(\frac{\epsilon D}{20000}\right)+ \frac{1}{\epsilon^2D^2}\left(n + \TS_G\left(\frac{\epsilon D}{20000}\right)\right)\right) = \tilde{O}_\eta\left( \frac{\Delta n^2}{\epsilon^2 \theta^2}\right) = \tilde{O}_\eta\left( \frac{\Delta n^7}{\epsilon^{5/2}}\right).
\end{align*}
If $\dis(\mu,\nu) \leq \theta$, we use \Cref{thm:hardcore-adv} with running time $\tilde{O}_\eta\left(\frac{n^7}{\epsilon^2}+\frac{n^{6.5}}{\epsilon^{9/4}}\right)$. The over all running time is the maximum of two cases, which is $\tilde{O}_\eta\left( \frac{\Delta n^7}{\epsilon^{5/2}}\right)$.
\end{proof}

The choice of $\theta$ is closely related to the choices of $t$ and $\kappa$ in the proof of \Cref{lem:hardcore-adv-1}. We choose the parameters to minimize the exponent on $n$ in the running time of \Cref{thm:hardcore-1}.

\ifthenelse{\boolean{conf}}{\begin{proof}\textbf{of \Cref{thm:hardcore-2}}}{\begin{proof}[Proof of \Cref{thm:hardcore-2}]}
Now, we further assume that $\Delta = O(1)$ and $\lambda^\pi_v = \Omega(1)$ or $0$ for all $v\in V$ and $\pi \in \{\mu,\nu\}$.
We do a similar pre-processing step as that in \Cref{sec:proof-gen}.
Suppose there exists $v$ such that $(\lambda^\mu_v = \Omega(1),\lambda^\nu_v = 0)$ or $(\lambda^\mu_v = \Omega(1),\lambda^\nu_v = 0)$. Say we are in the first case. Then $\nu_v(+) = 0$ and $\mu_v(-) \geq \frac{\lambda^\mu_v}{1+\lambda^\mu_v}(\frac{1}{1+\lambda_c(\Delta)})^\Delta = \Omega(1)$. The total variation $\DTV{\mu}{\nu}=\Omega(1)$. We can use \Cref{thm:Approximate-Gibbs} to solve the problem in time $\tilde{O}(\frac{n^2}{\epsilon^2})$. For all $v \in V$ with $\lambda^\mu_v = \lambda^\nu_v=0$, we can remove $v$. Hence, we can assume $\Omega(1)=\lambda^\pi_v \leq (1 - \eta)\lambda_c(\Delta)$ for all $v \in V$ and $\pi \in \{\mu,\nu\}$.

In this case, we use a different threshold $\theta_0 = \Theta(\frac{1}{\Delta n}) = \Theta(\frac{1}{n})$ because $\Delta=O(1)$. Suppose $\dis(\mu,\nu) \leq \theta_0$, then by~\eqref{eq:w-var}, we have $\Var[]{W} = O(n^2)$. It is easy to verify that $\frac{Z_\nu}{Z_\mu} = \Theta(1)$.  \Cref{thm:alg-main} gives an algorithm in time $\tilde{O}_\eta(n^3/\epsilon^2)$. Let us assume $D = \dis(\mu,\nu) > \theta_0$. If we directly apply \Cref{thm:Approximate-Gibbs} to achieve the additive error $\frac{\epsilon D}{5000}$, then the running time would be
\begin{align*}
    O\left(\TC_G\left(\frac{\epsilon D}{20000}\right)+ \frac{1}{\epsilon^2D^2} \TS_G\left(\frac{\epsilon D}{20000}\right)\right).
\end{align*}
The second term $\frac{1}{\epsilon^2D^2} \TS_G\left(\frac{\epsilon D}{20000}\right)= \tilde{O}_\eta(n^3/\epsilon^2) $. But the bottleneck is the first term $\TC_G\left(\frac{\epsilon D}{20000}\right) = \tilde{O}_\eta(\frac{ n^2}{\epsilon^2D^2}) = \tilde{O}_\eta(\frac{ n^2}{\epsilon^2 \theta_0^2}) = \tilde{O}_\eta(n^4/\epsilon^2)$. However, we can improve the first term by noting that the algorithm in \Cref{thm:Approximate-Gibbs} only needs to approximate the ratio $\frac{Z_\nu}{Z_\mu}$ with relative-error $O(\epsilon D)$ and we show such ratio can be approximated in time $\tilde{O}_\eta(n^3/\epsilon^2)$. We can construct a sequence of $\lambda^{(0)},\lambda^{(1)},\cdots,\lambda^{(\ell)}$ such that  $\lambda^{(0)} = \lambda^\mu$, $\lambda^{(\ell)} = \lambda^{\nu}$ and other $\lambda^{(i)}$ are defined as follows. For any $v \in V$, let $\delta_v = \frac{\lambda^{\nu}_v}{\lambda^{\mu}_v}$. For any $1 \leq i \leq \ell$,  $\lambda^{(i)}_v$ is defined by $\lambda^{(i)}_v = \lambda^\mu_v \delta_v^{i/\ell}$.  We choose $\ell$ such that $\ell = \Theta(1+ nD)$. Note that $\delta_v  =1 \pm O(D)$. We have $\delta_v^{1/\ell} = 1 \pm O(\frac{1}{n})$.  Let $w_i$ be the weight function induced by $\lambda^{(i)}$. 
Let $Z_i$ be the partition function induced by $w_i$. 
Let $\mu_i$ be the Gibbs distribution induced by $w_i$.
We further define $Z_{\ell + 1}$ by setting $\lambda^{(\ell + 1)}_v = \lambda^\mu_v \delta_v^{(\ell+1)/\ell}$. 
Define random variable $W_i$ as 
\begin{align*}
    W_i = \frac{w_{i}(X)}{w_{i-1}(X)}, \quad \text{where } X \sim \mu_{i-1}.
\end{align*}
Define $W \defeq \prod_{i=1}^{\ell}W_i$, where $W_i$'s are mutually independent. It is easy to verify that 
\begin{align*}
\E[]{W} = \frac{Z_\nu}{Z_\mu} = \frac{Z_\ell}{Z_0}, \quad\text{and } \Var[]{W} \leq \E[]{W^2} = \prod_{i=1}^{\ell}\E[]{W_i^2} = \frac{Z_{\ell + 1}Z_{\ell}}{Z_0Z_1}.
\end{align*}
We have the following bound
\begin{align*}
 \frac{\Var[]{W}}{(\E[]{W})^2} \leq   \frac{\E[]{W^2}}{(\E[]{W})^2} & \leq \frac{Z_{\ell+1}}{Z_\ell} \cdot \frac{Z_0}{Z_1} = O(1).
\end{align*}
The last equality follows from the fact that $\delta_v^{1/\ell} = 1 \pm O(\frac{1}{n})$.
Hence, to achieve $O(\epsilon D)$ relative-error, we can draw $O(\frac{1}{\epsilon^2 D^2})$ samples of $W$, each sample costs $\tilde{O}_\eta(n \ell) = \tilde{O}_\eta(n + n^2 D)$ time. The total running time is 
\begin{align*}
    \tilde{O}_\eta\left(\frac{n + n^2 D}{\epsilon^2 D^2}\right) = \tilde{O}_\eta\left( \frac{n}{\epsilon^2 \theta_0^2} + \frac{n^2}{\epsilon^2 \theta_0} \right) = O_\eta\left( \frac{n^3}{\epsilon^2}\right). \ifthenelse{\boolean{conf}}{}{&\qedhere}&
\end{align*}
\end{proof}

\subsection{The algorithm for marginal distributions (Proof of \texorpdfstring{\Cref{thm:many-vertex-alg}}{Lg})}\label{sec:marginthm}

\Cref{thm:many-vertex-alg} is a simple corollary of \Cref{thm:approx-margin-tv}.
\ifthenelse{\boolean{conf}}{\begin{proof}\textbf{of \Cref{thm:many-vertex-alg}}}{\begin{proof}[Proof of \Cref{thm:many-vertex-alg}]}
For hardcore model $(G,\lambda^\mu)$, $Z^\sigma_{\mu}$ where $\sigma \in \{\pm\}^S$ is the partition function of $(G[\Lambda],\lambda_\Lambda^\mu)$. The set $\Lambda$ is obtained from $V$ by removing all vertices in $S$ together with all neighbors $u$ of vertices $v \in S$ such that $\sigma_v = +1$. If $(G, \lambda^\mu)$ satisfies the uniqueness condition, then $(G[\Lambda], \lambda_\Lambda^\mu)$ also satisfies the uniqueness condition. Hence, by the previous results in~\cite{CFYZ22,CE22} and \cite{SVV09}, for both $\mu$ and $\nu$, the approximate conditional counting oracle with $\TC_G(\epsilon)=O(\frac{\Delta n^2}{\epsilon^2} \mathrm{polylog}\frac{n}{\epsilon})$ exists and the sampling oracle with $\TS_G(\epsilon)=O(\Delta n \mathrm{polylog}\frac{n}{\epsilon})$ exists. The theorem follows from \Cref{thm:approx-margin-tv}.
\end{proof}

\ifthenelse{\boolean{conf}}
{\section{Hardness of approximating the TV-distance on a subset of vertices}\label{sec:hard}

In this section, we prove \Cref{thm:many-vertex}. Let $G_i$ be the graph defined as above. 
Let $n_i = n - i + 1$ denote the number of vertices in $G_i$.
One can construct a graph $G_i'$ by adding a set $\Lambda$ of $\ell$ isolated vertices to $G_i$.  Let $N = n_i + \ell$. Let $k(\cdot)$ be the function in \Cref{thm:many-vertex}. Note that $k(N) = N - \lceil N^\alpha \rceil$. Since $\alpha$ is a constant, one can set $\ell = n^{\Omega(1/\alpha)} = \mathrm{poly}(n)$ so that $k(N) \leq \ell$.
Hence, the size of $G_i'$ is a polynomial in $n$.

Let $\mu^{(i)}_{\text{new}}$ be the hardcore model on $G_i'$ such that the external fields on vertices in $G_i$ are 1 and the external fields on $\lambda$ are 0. Let $\nu^\alpha_{\text{new}}$ be the hardcore model on $G_i'$ such that the external fields on vertices in $G_i$ are $\lambda^\alpha$ and the external fields on $\Lambda$ are 0. Let $S$ be a subset of vertices containing vertex $i$ and $k(N) - 1$ vertices in $\Lambda$. It holds that 
\begin{align*}
 \DTV{\mu_{\text{new},S}^{(i)}}{\nu_{\text{new},S}^\alpha} = \DTV{\mu^{(i)}_i}{\nu^\alpha_i}.   
\end{align*}
In words, the total variation distance between marginal distributions on $S$ projected from $\mu^{(i)}_{\text{new}}$ and $\nu^\alpha_{\text{new}}$ is the same as the total variation distance between marginal distributions on vertex $i$ projected from $\mu^{(i)}$ and $\nu^\alpha$. Note that both $\mu^{(i)}_{\text{new}}$ and $\nu^\alpha_{\text{new}}$ satisfy the uniqueness condition. \Cref{thm:many-vertex} can be verified by going through the same reduction for \Cref{thm:one-vertex}.}
{\section{Proofs of \#P-hardness results}\label{sec:hard}
In this section, we prove \Cref{thm:one-vertex} and \Cref{thm:many-vertex}. Our starting point is the \#P-hardness for exactly counting the number of independent sets in a graph.

\ifthenelse{\boolean{conf}}{\begin{proposition}[\text{\citet[Theorem 4.2]{DyerG00}}]}{\begin{proposition}[\text{\cite[Theorem 4.2]{DyerG00}}]}\label{prop:hard}
\ifthenelse{\boolean{conf}}{The following problem \#\textsf{Ind}(3) is \#P-complete.
Input: a graph $G=(V,E)$ with maximum degree $\Delta = 3$; Output: the exact number of independent sets in $G$.}{The following problem \#\textsf{Ind}(3) is \#P-complete.
\begin{itemize}
    \item Input: a graph $G=(V,E)$ with maximum degree $\Delta = 3$; 
    \item Output: the exact number of independent sets in $G$.
\end{itemize}}
\end{proposition}

The above problem is exactly computing the partition function of $(G,\lambda)$ with $\lambda_v = 1$ for all $v \in V$.
Let $n = |V|$ and $V = \{1,2,\ldots,n\}$. Let $\mu_{G,\bm{1}}$ denote the uniform distribution over all independent sets in $G$, which is the hardcore distribution in $G$ when $\lambda_v = 1$ for all $v \in V$. Define
\begin{align}\label{eq:def_pi}
    p_i = \Pr[X \sim \mu_{G,\bm{1}}]{X_i = 0 \mid \forall 1\leq j \leq i - 1, X_j = 0},
\end{align}
which is the probability that the vertex $i$ is not in a random independent set $X$ conditional on all $j < i$ not being in $X$. By definition, the total number of independent set is $Z = \frac{1}{\mu_{G,\boldsymbol{1}}(\boldsymbol{0})} = \prod_{i=1}^n \frac{1}{p_i}$.
Suppose for any $i \in [n]$, we can compute $\hat{p}_i$ such that 
\begin{align}\label{eq:tar}
   (1 - 4^{-n})p_i \leq \hat{p}_i \leq (1+4^{-n})p_i.
\end{align}
Let $\hat{Z} = \prod_{i=1}^n \frac{1}{\hat{p}_i}$ and it holds that $(1-3^{-n})Z\leq \hat{Z} \leq (1+3^{-n})Z$. Note that $Z \leq 2^n$. We have $|\hat{Z} - Z| \leq 3^{-n}Z \leq 1.5^{-n} < 0.01$. We can round $\hat{Z}$ to the nearest integer to recover $Z$. Hence, $\#\text{Ind}(3)$ can be reduced to the following high-accuracy marginal estimation problem.
\begin{problem}\label{problem:mar} The high-accuracy marginal estimation problem is defined by
\begin{itemize}
    \item Input: a graph $G=(V,E)$ with $n$ vertices and maximum degree $\Delta = 3$;
    \item Output: $n$ numbers $(\hat{p}_i)_{i \in [n]}$ such that for all $ i \in [n]$, $(1 - 4^{-n})p_i \leq \hat{p}_i \leq (1+4^{-n})p_i$.
\end{itemize}
\end{problem}

Let $k(n) = 1$ for all $n \in \mathbb{N}$ be a constant function. We show that if there is a $\mathrm{poly}(n)$ time algorithm for \Cref{label:prob-mar} if the input error bound $\epsilon = \mathrm{poly}(n)$ and both two input hardcore models satisfy the uniqueness condition, then \Cref{problem:mar} can also be solved in $\mathrm{poly}(n)$ time. \Cref{thm:one-vertex} follows from \Cref{prop:hard}.

Fix an integer $i \in [n]$. Let $G_i$ denote the induced graph $G[S_i]$, where $S_i =\{j \in [n] \mid j \geq i\}$ is the set of vertices with label at least $i$. Let $\mu^{(i)}$ denote the uniform distribution over all independent set in graph $G_i$. In other words, $\mu^{(i)}$ is the Gibbs distribution of hardcore model $(G_i,\boldsymbol{1})$. Then $p_i$ in~\eqref{eq:def_pi} is the marginal distribution on vertex $i$ projected from $\mu^{(i)}$. If the maximum degree of $G_i$ is at most $2$, then $G_i$ is a set of disconnected lines or circles and $p_i$ can be computed exactly in polynomial time. We can assume the maximum degree of $G_i$ is $3$.  

Let $\alpha \geq 0$. Define vector $\lambda^\alpha \in \mathbb{R}^{S_i}$ by
\begin{align}\label{eq:def-lambda-alpha}
    \lambda_j^\alpha = \begin{cases}
        \frac{\alpha}{1 - \alpha} &\text{if } j = i; \\
        0 &\text{if } j \neq i.
    \end{cases}
\end{align}
Let $\nu^\alpha$ denote the Gibbs distribution of $(G_i,\lambda^\alpha)$. Note that $\lambda_c(3) = 4 > 1$. The following observation is easy to verify.
\begin{observation}\label{ob:uniq}
Both $\mu^{(i)}$ and $\nu^\alpha$ satisfies the uniqueness condition in~\eqref{eq:cond-hardcore} if $\alpha \leq \frac{1}{2}$.
\end{observation}

By the definition of $\nu^\alpha$, it is easy to see $\nu^\alpha_i(+1) = \alpha$  and
\begin{align*}
    \DTV{\mu^{(i)}_i}{\nu^\alpha_i} = \vert \mu^{(i)}_i(+1) - \alpha \vert = \vert p_i - \alpha \vert.
\end{align*}
Let $\+A(\alpha)$ be the algorithm such that given $\alpha \in [0,\frac{1}{2}]$, it returns a number $\hat{d}$ such that $ \frac{d_{\text{TV}}(\mu^{(i)}_i,\nu^\alpha_i)}{1+\epsilon}\leq \hat{d} \leq (1+\epsilon)d_{\text{TV}}(\mu^{(i)}_i,\nu^\alpha_i)$, where $\epsilon = \mathrm{poly}(n)$. By \Cref{ob:uniq}, if the polynomial-time algorithm for the problem in \Cref{thm:one-vertex} exists, then $\+A(\alpha)$ runs in $\mathrm{poly}(n)$ time.
We then can use the following algorithm to solve \Cref{label:prob-mar} for $p_i$ in $\mathrm{poly}(n)$ time. Thus, the hardness result in \Cref{thm:one-vertex} follows from \Cref{prop:hard}.

\ifthenelse{\boolean{conf}}{\begin{algorithm2e}[ht]}{
\begin{algorithm}[ht]
}\label{alg:high-TV}
    \caption{Algorithm for high-accuracy marginal estimation}
    Let $\alpha \gets \frac{1}{2}$ and $\epsilon = \mathrm{poly}(n)$ be the parameter assumed by algorithm $\+A$\;
    \For{$t$ from 1 to $50n (1+\epsilon)^2$}{
        $\hat{d} \gets \+A(\alpha)$\;
        $\alpha \gets \alpha - \hat{d}/(1+\epsilon)$\;
        if the bit length of $\alpha$ is more than $100n$, then round $\alpha$ up to the nearest number that has bit length at most $100n$\;
    }
    \Return $\hat{p}_i = \alpha$.
\ifthenelse{\boolean{conf}}{\end{algorithm2e}}{
\end{algorithm}
}
The above algorithm runs in $\mathrm{poly}(n)$ time. We show that the output $\hat{p}$ satisfies \eqref{eq:tar}.

Let $\alpha_t$ be the value of $\alpha$ after the $t$-th iteration. We first show that $p_i \leq \alpha_t$ for all $t$. 
 At the beginning, $\alpha_0 = 1/2$. Since $\mu^{(i)}$ is a uniform distribution over all independent sets, we have $\mu^{(i)}_i(+1) \leq \frac{1}{2}$. By the assumption of algorithm $\+A$, $\+A(\alpha_{t})/(1+\epsilon) \leq  d_{\text{TV}}(\mu^{(i)}_i,\nu^{\alpha_t}_i) = \alpha_t - p_i$. Hence, $\alpha_{t+1} \geq \alpha_t -  \+A(\alpha_{t})/(1+\epsilon) \geq p_i $. 
 
We next bound the value of $\alpha_t - p_i$. At the beginning, $\alpha_0 = \frac{1}{2}$ so that $\alpha_0 - p_i \leq \frac{1}{2}$. Note that $\alpha_{t+1} < \alpha_t -  \frac{\+A(\alpha_{t})}{1+\epsilon}+2^{-90n} \leq \alpha_t -  \frac{\alpha_t-p_i}{(1+\epsilon)^2}+2^{-90n} $, where $2^{-90n}$ is an upper bound of rounding error. The inequality implies that 
\begin{align*}
    \alpha_{t+1} - p_i \leq \left( 1 - \frac{1}{(1+\epsilon)^2} \right)(\alpha_t - p_i) + 2^{-90n}.
\end{align*}
Note that $\hat{p}_i = \alpha_{50n(1+\epsilon)^2}$.
Solving the recurrence implies that
\begin{align*}
    0 \leq \hat{p}_i - p_i \leq \exp\left( -\frac{50n(1+\epsilon)^2}{(1+\epsilon)^2} \right)\cdot \frac{1}{2} + (1+\epsilon)^2 \cdot 2^{-90n} < 2^{-40n}.
\end{align*}
Note that $p_i$ is at least $1/2^n$. Hence, the output $\hat{p}$ satisfies \eqref{eq:tar}.

\begin{remark*}
In the above proof, while the TV-distance between the two Gibbs distributions \(\mu^{(i)}\) and \(\nu^{\alpha}\) is large (because their parameter distance is 1), the TV distance between their marginal distributions at vertex \(i\) can be arbitrarily small. This highlights the distinction between the TV-distance of marginal distributions and the TV-distance of the entire distribution.
\end{remark*}

\subsection{Hardness of approximating the TV-distance on a subset of vertices}

We now prove \Cref{thm:many-vertex}. Let $G_i$ be the graph defined as above. 
Let $n_i = n - i + 1$ denote the number of vertices in $G_i$.
One can construct a graph $G_i'$ by adding a set $\Lambda$ of $\ell$ isolated vertices to $G_i$.  Let $N = n_i + \ell$. Let $k(\cdot)$ be the function in \Cref{thm:many-vertex}. Note that $k(N) = N - \lceil N^\alpha \rceil$. Since $\alpha$ is a constant, one can set $\ell = n^{\Omega(1/\alpha)} = \mathrm{poly}(n)$ so that $k(N) \leq \ell$.
Hence, the size of $G_i'$ is a polynomial in $n$.

Let $\mu^{(i)}_{\text{new}}$ be the hardcore model on $G_i'$ such that the external fields on vertices in $G_i$ are 1 and the external fields on $\lambda$ are 0. Let $\nu^\alpha_{\text{new}}$ be the hardcore model on $G_i'$ such that the external fields on vertices in $G_i$ are $\lambda^\alpha$ and the external fields on $\Lambda$ are 0. Let $S$ be a subset of vertices containing vertex $i$ and $k(N) - 1$ vertices in $\Lambda$. It holds that 
\begin{align*}
 \DTV{\mu_{\text{new},S}^{(i)}}{\nu_{\text{new},S}^\alpha} = \DTV{\mu^{(i)}_i}{\nu^\alpha_i}.   
\end{align*}
In words, the total variation distance between marginal distributions on $S$ projected from $\mu^{(i)}_{\text{new}}$ and $\nu^\alpha_{\text{new}}$ is the same as the total variation distance between marginal distributions on vertex $i$ projected from $\mu^{(i)}$ and $\nu^\alpha$. Note that both $\mu^{(i)}_{\text{new}}$ and $\nu^\alpha_{\text{new}}$ satisfy the uniqueness condition. \Cref{thm:many-vertex} can be verified by going through the same reduction for \Cref{thm:one-vertex}.
}

\ifthenelse{\boolean{conf}}{}{\section*{Acknowledgment}
Weiming Feng and Hongyang Liu gratefully acknowledge the support of ETH Z\"urich, where part of this work was conducted.  
Weiming Feng acknowledges the support of Dr.\ Max R\"ossler, the Walter Haefner Foundation, and the ETH Z\"urich Foundation during his affiliation with ETH Z\"urich.}

\bibliographystyle{alpha}
\bibliography{refs}
\appendix
\section{NP-hardness of approximating the TV-distance} \label{app:proof}
In this section, we prove the hardness results of approximating the TV-distance between two Gibbs distributions beyond the uniqueness threshold. 
The proof is based on the technique developed in~\cite{BGMMPV24ICLR}.
We define the following instance family for TV-distance approximation.
\begin{problem}\label{prob:hardcore-tv}
Let $\Delta \geq 3$ and $\lambda > \lambda_c(\Delta) = \frac{(\Delta-1)^{\Delta-1}}{(\Delta-2)^{\Delta}}$ be two constants.
\begin{itemize}
    \item \emph{Input}: two hardcore models $(G,\lambda^\mu)$ and $(G,\lambda^\nu)$ defined on the same graph $G = (V,E)$ with maximum degree at most $\Delta$, which specifies two Gibbs distributions $\mu$ and $\nu$ respectively, and an error bound $0<\epsilon<1$. 
    There exists a vertex $v^* \in V$ such that 
    the external fields $\lambda^\mu_v = \lambda^\nu_v = \lambda$ for all $v \neq v^*$, $\lambda^\nu_{v^*} = \infty$, and $\lambda^\mu_{v^*} = \lambda$.
    \item \emph{Output}: a number $\hat{d}$ such that $|\DTV{\mu}{\nu} - \hat{d}| \leq \epsilon$.
\end{itemize}
\end{problem}

The two input hardcore models are not in the uniqueness regime because $\lambda > \lambda_c(\Delta)$.  The output only requires to approximate the TV-distance up to an additive error of $\epsilon$, which is weaker requirement than the relative-error approximation. The hardness result for this problem implies the hardness of relative-error approximation.

\begin{theorem} \label{thm:hardcore-tv}
There is no FPRAS for \Cref{prob:hardcore-tv} unless $\textbf{NP}=\textbf{RP}$.
\end{theorem}

To prove \Cref{thm:hardcore-tv}, we need the following lemma, which can be abstracted from the proofs in~\cite{BGMMPV24ICLR}.

\ifthenelse{\boolean{conf}}{\begin{lemma}[\citet{BGMMPV24ICLR}]}{\begin{lemma}[\cite{BGMMPV24ICLR}]}
\label{lem:hardcore-tv-lemma}
Let $\mu$ be a distribution over $\{\pm\}^V$.
Let $v \in V$ and $c \in \{\pm\}$ with $\mu_v(c) > 0$.
Let $\mu^{vc}$ be the distribution over $\{\pm\}^{V}$ obtained from $\mu$ by conditioning on $v$ taking value $c$. Then, $ \DTV{\mu}{\mu^{vc}} = \mu_v(-c)$.
\end{lemma}
\begin{proof}
For any $\sigma \in \{\pm\}^V$ with $\sigma_v = c$, we have $\mu^{vc}(\sigma) \geq \mu(\sigma)$. For any $\tau \in \{\pm\}^V$ with $\tau_v =  -c$, we have $\mu^{vc}(\tau) = 0 \leq \mu(\tau)$. Therefore, $\DTV{\mu}{\mu^{vc}} = \sum_{\tau \in \{\pm\}^V:\tau_v = -c}\mu(\tau) = \mu_v(-c)$.
\end{proof}

\ifthenelse{\boolean{conf}}{\begin{proof}\textbf{of \Cref{thm:hardcore-tv}}}{\begin{proof}[Proof of \Cref{thm:hardcore-tv}]}
Let $\Delta \geq 3$ and $\lambda > \lambda_c(\Delta) = \frac{(\Delta-1)^{\Delta-1}}{(\Delta-2)^{\Delta}}$ be two constants.
By the hardness results in~\cite{SlyS12,GalanisSV16}, unless $\textbf{NP}=\textbf{RP}$, there is no FPRAS for approximating the partition function of the hardcode model $\mathbb{S}=(G,(\lambda_v)_{v \in V})$ with $\epsilon$-relative error, where $\lambda_v = \lambda$ for all $v \in V$ and $G$ has the maximum degree $\Delta$. By the standard counting-to-sampling reduction~\citep{JVV86}, approximating the partition function $Z_{\mathbb{S}}$ is equivalent to approximating the probability of $\pi(\sigma^\emptyset)$, where $\pi$ is the Gibbs distribution of $\mathbb{S}$ and $\sigma^\emptyset_v = -1$ for all $v \in V$. In other words, $\sigma^\emptyset$ corresponds to the empty set.
Let us number all the vertices in $V$ as $\{1,2,\cdots,n\}$. To approximate $\pi(\sigma^\emptyset)$, we need to approximate the probability $p_i \defeq \pi_i(-1 \mid \forall j < i, j \text{ takes value } -1)$  with relative error $O(\frac{\epsilon}{n})$. This probability is the same as the marginal probability of $i$ in the induced subgraph $G[V\setminus\{1,2,\cdots,i-1\}]$. We show how to approximate $p_1 = \pi_1(-1)$. We let $\mathbb{S}^\mu  = \mathbb{S}$ and define $\mathbb{S}^\nu$ for $\mathbb{S}$ by changing $\lambda^\nu_1$ to $\infty$. 
By \Cref{lem:hardcore-tv-lemma}, we have $\DTV{\mu}{\nu} = \mu_1(-1) = \pi_1(-1)$. It is easy to verify that $\pi_1(-1) \geq \frac{1}{1+\lambda} = \Omega(1)$ has a constant lower bound. Hence, if we can solve \Cref{prob:hardcore-tv} with $O(\epsilon/n)$-additive error, we can approximate $\pi_1(-1)$ with $O(\epsilon/n)$-relative error. The same argument can be applied to other probabilities by considering the instances in induced subgraphs. Hence, if \Cref{prob:hardcore-tv} admits an FPRAS, then there is an FPRAS for approximating the partition function of $\mathbb{S}$.
\end{proof}

For the Ising model, one can also verify the instance family stated after \Cref{Cor:Ising} is a hard instance family for approximating the TV-distance with additive error. The proof is similar to the one for the hardcore model. 
Let $\Delta \geq 3$ be a constant and $\beta < 0$ be a constant with $\exp(2\beta) < \frac{\Delta-2}{\Delta}$.
Let $\mathbb{S} = (G,J,h)$ such that $J_{uv} = \beta$ for all $\{u,v\} \in E$ and $h_v = 0$ for all $v \in V$. 
The starting point is the \text{NP}-hardness of approximating the partition function of $\mathbb{S}$~\citep{SlyS12,GalanisSV16}. 
Then one can go through the same reduction to estimate $p_i \defeq \pi_i(-1 \mid \forall j < i, j \text{ takes value } -1)$. We can construct $\mathbb{S}^\mu = (G,J,h^\mu)$ and $\mathbb{S}^\nu = (G,J,h^\nu)$ such that for all $j < i$, $h^\mu_j = h^\nu_j = -\infty$, for all $k > i$, $h^\mu_k = h^\nu_k = 0$, and $h^\mu_i = 0,h^\nu_i = \infty$. By \Cref{lem:hardcore-tv-lemma}, we have $\DTV{\mu}{\nu} = p_i$. One can verify $p_i = \Omega(1)$ so that $O(\epsilon/n)$-additive error approximation of $\DTV{\mu}{\nu}$ is equivalent to $O(\epsilon/n)$-relative error approximation of $p_i$. The hardness result follows from the same argument as in the proof of \Cref{thm:hardcore-tv}.

\section{Poincar\'e inequality for marginal distribution}\label{app:poin}
The Poincar\'e inequality for hardcore model in the uniqueness regime is established in~\cite{ChenFYZ21}. The paper considers the hardcore model $(G,\lambda)$ such that all $\lambda_v$ are the same. By verifying the technical condition in\ifthenelse{\boolean{conf}}{~\citet[Theorem 1.9]{ChenFYZ21}}{~\cite[Theorem 1.9]{ChenFYZ21}}, the following Poincar\'e inequality also holds for hardcore model with different $\lambda_v$ satisfying~\eqref{eq:cond-hardcore}. For any function $g: \Omega \to \mathbb{R}$, where $\Omega$ is the support of the distribution, we have
\begin{align*}
\Var[\mu]{g} \leq C_\eta \sum_{v \in V} \sum_{\sigma \in \Omega_{V -v}} \mu_{V - v}(\sigma) \Var[\mu^\sigma]{g}.
\end{align*}
Let $g': \Omega_B \to \mathbb{R}$ be an arbitrary function. 
Define $g(x) = g'(x_B)$ for all $x \in \Omega$. It holds that $\E[\mu]{g} =\E[\mu_B]{g'}$ and $\E[\mu]{g^2} = \E[\mu_B]{(g')^2}$. We have
\begin{align*}
    \Var[\mu_B]{g'} &= \Var[\mu]{g}  \leq C_\eta \sum_{v \in V} \sum_{\sigma \in \Omega_{V -v}} \mu_{V - v}(\sigma) \Var[\mu^\sigma]{g}\\
\text{($g(\sigma)$ depends on $\sigma_B$)}\quad    &= C_\eta \sum_{v \in B}  \sum_{\sigma \in \Omega_{V -v}} \mu_{V - v}(\sigma) \Var[\mu^\sigma]{g}.
\end{align*}
Fix a partial assignment $\tau \in \Omega_{B-v}$. Note that $\E[\mu^\tau]{g} = \E[\mu^\tau_B]{g'}$ and  $\E[\mu^\tau]{g^2} = \E[\mu^\tau_B]{(g')^2}$. We have $\Var[\mu^\tau]{g} = \Var[\mu^\tau_B]{g'}$. Let $X \sim \mu^\tau$ and $Y = g(X)$. By the law of total variance,
\begin{align*}
  \Var[\mu^\tau_B]{g'} &= \Var[\mu^\tau]{g} = \Var[]{Y}\\
   &= \E[]{\Var[]{Y \mid X_{V-v}}} + \Var[]{\E[]{Y \mid X_{V-v}}}\\
   &\geq \E[]{\Var[]{Y \mid X_{V-v}}}\\
   &= \sum_{\sigma \in \Omega^\tau_{V-v}} \mu^\tau_{V-v}(\sigma) \Var[\mu^\sigma]{g},
\end{align*}
where $\Omega^\tau_{V-v} \subseteq \{\pm\}^{V-v}$ is the support of $\mu^\tau_{V-v}$.
Combining with the above inequality, we have
\begin{align*}
    \Var[\mu_B]{g'} &\leq C_\eta \sum_{v \in B}  \sum_{\tau \in \Omega_{B -v}} \mu_{B-v}(\tau) \sum_{\sigma \in \Omega^\tau_{V-v}} \mu^\tau_{V-v}(\sigma) \Var[\mu^\sigma]{g}\\ &\leq C_\eta \sum_{v \in B}  \sum_{\tau \in \Omega_{B -v}} \mu_{B-v}(\tau) \Var[\mu^\tau_B]{g'}.
\end{align*}
This proves the Poincar\'e inequality for the marginal distribution.

Alternatively, one can also use the fact the the Poincar\'e inequality is equivalent to the decay of $\chi^2$-divergence in the down walk of Glauber dynamics. The results follows from the data processing inequality. See \ifthenelse{\boolean{conf}}{~\citet[Section 6.1]{FGJW23}}{~\cite[Section 6.1]{FGJW23}} for more details.

\end{document}